\documentclass[aip,jmp,amsmath,amssymb,amsopn,superscriptaddress,preprint,author-numerical]{revtex4-1}

\usepackage{graphicx,color}
\usepackage{textcomp}
\usepackage{braket,slashed}
\usepackage{hyperref}
\usepackage{enumerate}

\usepackage{bm}

\usepackage[mathscr]{euscript}
\usepackage{dsfont}

\newcommand{\operator}[1]{\ensuremath{\hat{#1}}}
\newcommand{\voperator}[1]{\ensuremath{\mathds{#1}}}
\newcommand{\manifold}[1]{\ensuremath{\mathcal{#1}}}
\newcommand{\group}[1]{\ensuremath{\mathsf{#1}}}
\newcommand{\vectorspace}[1]{\ensuremath{\mathbb{#1}}}

\DeclareMathOperator*{\tr}{tr}
\DeclareMathOperator*{\varmin}{min}

\DeclareMathOperator{\order}{\mathscr{O}}

\DeclareMathOperator{\rank}{rank}
\DeclareMathOperator{\End}{\mathbb{L}}

\newcommand{\id}{\mathrm{id}}
\newcommand{\one}{\mathds{1}}

\newcommand{\rmd}{\ensuremath{\mathrm{d}}}
\newcommand{\rmi}{\ensuremath{\mathrm{i}}}
\newcommand{\rme}{\ensuremath{\mathrm{e}}}
\newcommand{\defis}{\ensuremath{\triangleq}}

\newcommand{\opartial}{\ensuremath{\overline{\partial}}}
\newcommand{\vz}{\ensuremath{\bm{z}}}
\newcommand{\ovz}{\ensuremath{\overline{\bm{z}}}}
\newcommand{\oz}{\ensuremath{\overline{z}}}
\newcommand{\vw}{\ensuremath{\bm{w}}}
\newcommand{\ovw}{\ensuremath{\overline{\bm{w}}}}

\newcommand{\vv}{\ensuremath{\bm{v}}}

\newcommand{\rket}[1]{\ensuremath{|#1)}}
\newcommand{\rbra}[1]{\ensuremath{(#1|}}
\newcommand{\rbraket}[1]{\ensuremath{(#1)}}

\newcommand{\hilbert}{\ensuremath{\vectorspace{H}}}
\newcommand{\Tplane}{\ensuremath{\vectorspace{T}}}
\newcommand{\varM}{\ensuremath{\manifold{M}}}

\usepackage{amsthm}
\newtheorem{theorem}{Theorem}
\newtheorem{corollary}[theorem]{Corollary}
\newtheorem{lemma}[theorem]{Lemma}

\theoremstyle{remark}

\theoremstyle{definition}
\newtheorem{definition}{Definition}

\begin{document}

\title[Geometry of Matrix Product States]{Geometry of Matrix Product States:\\metric, parallel transport and curvature}
\author{Jutho Haegeman}
\email{jutho.haegeman@gmail.com}
\affiliation{Vienna Center for Quantum Science and Technology, Faculty of Physics, University of Vienna, Austria}
\affiliation{Faculty of Physics and Astronomy, University of Ghent, Krijgslaan 281 S9, 9000 Gent, Belgium}
\author{Micha\"{e}l Mari\"{e}n}
\affiliation{Faculty of Physics and Astronomy, University of Ghent, Krijgslaan 281 S9, 9000 Gent, Belgium}
\author{Tobias J Osborne}
\affiliation{Leibniz Universit\"{a}t Hannover, Institute of Theoretical Physics, Appelstrasse 2, D-30167 Hannover, Germany}
\affiliation{Leibniz Universit\"{a}t Hannover, Riemann Center for Geometry and Physics, Appelstrasse 2, D-30167 Hannover, Germany}
\author{Frank Verstraete}
\affiliation{Faculty of Physics and Astronomy, University of Ghent, Krijgslaan 281 S9, 9000 Gent, Belgium}
\affiliation{Vienna Center for Quantum Science and Technology, Faculty of Physics, University of Vienna, Austria}

\begin{abstract}
We study the geometric properties of the manifold of states described as (uniform) matrix product states. Due to the parameter redundancy in the matrix product state representation, matrix product states have the mathematical structure of a (principal) fiber bundle. The total space or bundle space corresponds to the parameter space, \textit{i.e.}\ the space of tensors associated to every physical site. The base manifold is embedded in Hilbert space and can be given the structure of a K\"{a}hler manifold by inducing the Hilbert space metric. Our main interest is in the states living in the tangent space to the base manifold, which have recently been shown to be interesting in relation to time dependence and elementary excitations. By lifting these tangent vectors to the (tangent space) of the bundle space using a well-chosen prescription (a principal bundle connection), we can define and efficiently compute an inverse metric, and introduce differential geometric concepts such as parallel transport (related to the Levi-Civita connection) and the Riemann curvature tensor.

\end{abstract}

\maketitle

\tableofcontents
\clearpage

\section{Introduction}
The most powerful method for studying one-dimensional gapped quantum spin systems is without doubt the density matrix renormalization group \cite{1992PhRvL..69.2863W,1993PhRvB..4810345W}. The density matrix renormalization group can be interpreted as a variational method that selects the best approximation to the true ground state of the system within the set of states known as matrix product states \cite{1995PhRvL..75.3537O,1997PhRvB..55.2164R}. The history of matrix product states dates back to before the development of the density matrix renormalization group, when they were referred to as valence bond states \cite{1987PhRvL..59..799A,1988CMaPh.115..477A,Klumper:1991fk,Klumper:1992uq} or finitely correlated states \cite{1991JPhA...24L.185F,1992CMaPh.144..443F}. In fact, related constructions were already developed in the study of classical statistical mechanics a few decades ago \cite{Kramers:1941kx,baxter:650}.

The development of matrix product states benefited greatly from insights regarding entanglement that were being formulated within the field of quantum information theory. This resulted in new algorithms for studying systems with periodic boundary conditions\cite{2004PhRvL..93v7205V}, time evolution\cite{2004PhRvL..93d0502V,2004PhRvL..93g6401W,2004JSMTE..04..005D,2006PhRvL..97o7202O} and for systems at finite temperature or with dissipative dynamics\cite{2004PhRvL..93t7204V,2004PhRvL..93t7205Z}. In addition, the basic structure of the matrix product state was generalized to more arbitrary networks of contracted tensors, in order to cope with different settings. Specific examples include the multiscale entanglement renormalization ansatz\cite{2007PhRvL..99v0405V,2008PhRvL.101k0501V} for critical systems, or the projected entangled-pair states for higher-dimensional systems\cite{2004cond.mat..7066V,2007PhRvA..75c3605M,2008AdPhy..57..143V}.

So far, most studies focussed on the physical properties of these states. The mathematical and geometric properties of the whole set of matrix product states with a fixed bond dimension have so far received less attention. Recently, however, new algorithms for simulating time-evolution\cite{2011arXiv1103.0936H} and for studying excitation spectra\cite{2012PhRvB..85c5130P,2012PhRvB..85j0408H} were developed that inherently depend on the tangent space to the set of matrix product states. However, such an approach is only justified if we can identify this set as a smooth manifold embedded in the total Hilbert space of quantum states. The purpose of this paper is to make this identification and elaborate on many of the details behind the construction lying at the heart of the aforementioned algorithms. We do not discuss the algorithms themselves. The current paper restricts to the identification of the matrix product state construction as a principal fiber bundle, followed by a discussion of the differential geometric properties of its base manifold (the set of states in Hilbert space). A more detailed presentation of the aforementioned algorithms including the relationship between them is given elsewhere\footnote{J.~Haegeman \textit{et al.}, \textit{in preparation}}. 

The first section of this paper introduces standard concepts from differential geometry applied to the case of a variational manifold, \textit{i.e.} a set of quantum states $\ket{\Psi(\vz)}$ ---depending on some complex variational parameters $z^i$, $i=1,\ldots,m$--- that is assumed to form a smooth manifold embedded in the total Hilbert space $\hilbert$ of the problem. We also discuss the geometric properties of Hilbert space itself, and give a brief summary of the theory of complex manifolds. In the second section we study generic matrix product states with open boundary conditions. We discuss the conditions that need to be imposed in order to identify the matrix product state representation as a principal fiber bundle, from which we can derive that the resulting set of physical states is a smooth manifold, more specifically a K\"{a}hler manifold, using standard tools from (complex) differential geometry. The same construction is repeated in the third section for the case of uniform matrix product states with periodic boundary conditions, which requires some non-trivial modifications. These are the main results of this paper. Having a differential structure, we then go on in both sections by introducing tangent vectors to the manifold of physical states. The parameterization thereof requires the introduction of a principal bundle connection, which can be fine-tuned in order to simplify the induced Hilbert space metric. This was one of the main results used in the aforementioned algorithms. Finally, we complete the geometric description of the manifold of matrix product states by also deriving the Levi-Civita connection and Riemann curvature tensor for this manifold. 

Identical tensor network decompositions for higher-order tensors have been independently developed in the field of applied mathematics. The main interest there is on the matrix product structure (tensor train decomposition)\cite{Oseledets:2009fk} and on the tree-tensor structure (hierarchical Tucker format)\cite{Hackbusch:2009uq}. Only this year was it proven that the set of states defined by these formats does indeed constitute a smooth manifold\cite{Holtz:2012kx,Uschmajew:2012vn} using techniques similar to ours. The present paper deviates from previous results as quantum states live in a complex rather than a real vector space. Secondly, we allow for matrix product states with periodic boundary conditions. In general, loops in a tensor network have to be treated carefully because the resulting set of states might not be closed\cite{2011arXiv1105.4449L}. Therefore, they have not been considered in aforementioned papers. In addition, we have a natural geometry induced from the Hilbert space in which these manifolds are embedded, which allows us to define a metric, a Levi-Civita connection and a Riemann curvature tensor.

Finally, we would also like to draw attention to the work of \citet{Sidles:2009fk}, where variational classes of quantum states multilinear in the variational parameters were also recognized as K\"{a}hler manifolds. Generic matrix product states ---and in fact more general tensor networks--- do indeed fulfill this multilinearity property. However, the translation-invariant uniform matrix product states of  Section~\ref{s:umps} do not. \citet{Sidles:2009fk} does not discuss the principal fiber bundle structure and the corresponding principal bundle connection, which is considered a key result of the current manuscript. However, Ref.~\onlinecite{Sidles:2009fk} discusses in great detail the curvature properties of these manifolds and the physical relevance thereof, in particular in relation to the error resulting from approximating arbitrary quantum states within the variational set. As such, \citet{Sidles:2009fk} can be considered complementary to our manuscript, and we only provide a brief discussion of the curvature properties for the case of uniform matrix product states in Subsection~\ref{ss:umps:curvature}.

\section{Variational manifolds as K\"{a}hler manifolds}
The state of an isolated quantum system is described by a vector $\ket{\Psi}$ living in a Hilbert space $\hilbert$. If the dimension of the Hilbert space is too large to be handled numerically, one often resorts to variational classes of quantum states, \textit{i.e.} subsets of $\hilbert$ that hopefully contain a good approximation to the physically interesting states of the problem. While it is possible to construct variational classes that are subspaces, there are many interesting classes for which the vector space structure is lost in the variational subset. However, the way in which the variational class is constructed often suggests that the variational subset can still be given the structure of a smooth manifold to which we can induce the geometric properties of the underlying Hilbert space. Under general conditions which are discussed in the Subsection~\ref{ss:var:varmanifold}, the subset will be a complex manifold, or to be even more specific, a K\"{a}hler manifold, and a beautiful new structure becomes available to study the approximated quantum problem within the variational manifold, which is further explored in the third and fourth subsection. As complex differential geometry might not be part of the standard toolbox of our target audience, we provide a very concise introduction to the subject in Subsection~\ref{ss:var:complexgeometry}, by reviewing the minimal set of definitions required to understand the remainder of this paper.

\subsection{Crash course in complex geometry}
\label{ss:var:complexgeometry}
This subsection mainly serves to introduce notation, and we refer to many excellent references for a more detailed treatment of the theory of complex manifolds and proofs for the corresponding theorems and lemmas \cite{Fritzsche:2002kx,Nakahara:2003ys,Huybrechts:2004vn}. In particular, we follow the notation convention of Ref.~\onlinecite{Nakahara:2003ys}. 
\begin{definition}[Complex manifold] A topological space $\varM$ is a complex manifold of complex dimension $m$ if 
\begin{enumerate}[(i)]
\item $\varM$ is provided with an atlas $\{(\manifold{U}_i,\phi_i)\}$, \textit{i.e.}\ a family of charts $(\manifold{U}_i,\phi_i)$, where $\{\manifold{U}_i\}$ is a family of open sets that covers $\varM$ and $\phi_i$ is a homeomorphism from $\manifold{U}_i$ to an open subset of $\mathbb{C}^m$. 
\item Given $\manifold{U}_i$ and $\manifold{U}_j$ such that $\manifold{U}_i \cap \manifold{U}_j \neq \emptyset$, the transition map $\psi_{ij}=\phi_j \circ \phi_i^{-1}$ from $\phi_i(\manifold{U}_i \cap \manifold{U}_j)$ to $\phi_j(\manifold{U}_i \cap \manifold{U}_j)$ is holomorphic, \textit{i.e.} the limit
\begin{displaymath}
\lim_{\vz \to \vz_0} \frac{\psi_{ij}(\vz)-\psi_{ij}(\vz_0)}{\vz-\vz_0}
\end{displaymath}
exists for every $\vz_0\in \phi_i(\manifold{U}_i \cap \manifold{U}_j)$.
\end{enumerate}
\end{definition}
Henceforth, we always refer to the complex dimension of any manifold we encounter and denote it as $\dim \varM\defis \dim_{\mathbb{C}} \varM$. The corresponding real dimension is twice the complex dimension and is denoted as $\dim_{\mathbb{R}}\varM=2\dim_{\mathbb{C}} \varM$. Clearly, the most elementary complex manifold is $\mathbb{C}^m$ itself.

\begin{definition}[Holomorphic map] Let $f:\varM\to\manifold{N}$ be a continuous map between two complex manifolds $\varM$ and $\mathcal{N}$ with respective (complex) dimensions $m$ and $n$ and respective atlases $\{(\manifold{U}_i,\phi_i)\}$ and $\{(\mathcal{V}_j,\psi_j)\}$. Define the open sets $\manifold{U}_{i,j}=\manifold{U}_i \cap f^{-1}(\mathcal{V}_j)$ with $f^{-1}(\mathcal{V}_j)$ the preimage of $V_j$. The map $f$ is called \emph{holomorphic} if the $n$ components of $\psi_j \circ f \circ \phi_i^{-1}:\mathbb{C}^m \to \mathbb{C}^n$ are holomorphic in any of the $m$ variables in the open subset $\phi_i(\manifold{U}_{i,j})\subset \mathbb{C}^m$. 
\end{definition}
If $f$ is injective and surjective, then $f$ is also a diffeomorphism and it is called a \emph{biholomorphism} because its inverse $f^{-1}:\manifold{N}\to\varM$ is also holomorphic. Correspondingly, the complex manifolds $\manifold{M}$ and $\manifold{N}$ are said to be \emph{biholomorphic}.

For any smooth manifold $\varM$, the tangent space $T_p \varM$ at a point $p\in\varM$ is the vector space of all directional derivatives of functions $f:\varM\to \mathbb{R}$ at the point $p$. For a complex manifold $\varM$ with $\dim \varM=m$, $T_p \varM$ is a real vector space with real dimension $2m$ and a basis is given by the partial derivatives
\begin{displaymath}
\left\{ \left(\frac{\partial\ }{\partial x^1}\right)_p,\ldots,\left(\frac{\partial\ }{\partial x^m}\right)_p;\left(\frac{\partial\ }{\partial y^1}\right)_p,\ldots,\left(\frac{\partial\ }{\partial y^m}\right)_p\right\}.
\end{displaymath}
The dual space or cotangent space $T_p^\ast \varM$ is spanned by the basis vectors
\begin{displaymath}
\left\{\left(\rmd x^1\right)_p,\ldots,\left(\rmd x^m\right)_p; \left(\rmd y^1\right)_p, \ldots, \left(\rmd y^m\right)_p\right\},
\end{displaymath}
which satisfy
\begin{equation}
\begin{split}
\left\langle \left(\rmd x^i\right)_p, \left(\frac{\partial\ }{\partial x^j}\right)_p\right\rangle=\delta^i_j,\qquad&\left\langle \left(\rmd x^i\right)_p, \left(\frac{\partial\ }{\partial y^j}\right)_p\right\rangle=0,\\
\left\langle \left(\rmd y^i\right)_p, \left(\frac{\partial\ }{\partial x^j}\right)_p\right\rangle=0,\qquad&\left\langle \left(\rmd y^i\right)_p, \left(\frac{\partial\ }{\partial y^j}\right)_p\right\rangle=\delta^i_j
\end{split}
\end{equation}
An \emph{almost complex structure} $J_p$ is introduced as a smooth (real) tensor field $J_p:T_p\varM\to T_p\varM$ via the prescription
\begin{align}
J_p=\left(\rmd x^i\right)_p \otimes \left(\frac{\partial\ }{\partial y^i}\right)_p - \left(\rmd y^i\right)_p \otimes \left(\frac{\partial\ }{\partial x^i}\right)_p.
\end{align}
This definition does not depend on the chosen coordinate map $\phi$ for complex manifolds $\varM$. Note that $J_p^2=-\id_{T_p\varM}$.

By also considering complex-valued functions $f\colon\varM\to\mathbb{C}$, it becomes useful to study the \emph{complexified tangent space}.
The complexified vector space $T_p\varM^{\mathbb{C}}$ is obtained by extending the vector space to all complex linear combinations of the basis vectors (see \cite{Roman:2005zr} for a more rigorous definition). It is a complex vector space with complex dimension $2m$. A new basis $\{(\partial_j)_p,(\opartial_{\overline{\jmath}})_p\}_{i=1,\ldots,m}$ for $T_p\varM^{\mathbb{C}}$ is obtained by defining
\begin{align}
\partial_j&\defis \frac{\partial\ }{\partial z^j}=\frac{1}{2}\left(\frac{\partial\ }{\partial x^j}-\rmi \frac{\partial\ }{\partial y^j}\right),&\opartial_{\overline{\jmath}}\defis \frac{\partial\ }{\partial \oz^{\overline{\jmath}}}=\frac{1}{2}\left(\frac{\partial\ }{\partial x^j}+\rmi \frac{\partial\ }{\partial y^j}\right).
\end{align}
The use of barrred indices for the complex conjugate variables will become clear when introducing a metric at the end of this subsection. Analogously, a complexified dual space $(T_p\varM^{\mathbb{C}})^{\ast} \equiv (T_p^{\ast} \varM)^{\mathbb{C}}$ is introduced and the corresponding basis $\{ (\rmd z^j)_p; (\overline{\rmd z}^{\overline{\jmath}})_p\}$ is defined by
\begin{align}
\rmd z^j &= \rmd x^j +\rmi \rmd y^j,&\overline{\rmd z}^{\overline{\jmath}}&=\rmd x^j-\rmi \rmd y^j.
\end{align}
In these new bases, the almost complex structure $J$ is given as the tensor field
\begin{equation}
J_p=\rmi \left(\rmd z^j\right)_p \otimes \left(\partial_j\right)_p-\rmi \left(\overline{\rmd z}^{\overline{\jmath}}\right)_p \otimes \left(\opartial_{\overline{\jmath}}\right)_p,
\end{equation}
 so that its matrix representation is diagonal. The complexified tangent space $T_p\varM^{\mathbb{C}}$ is decomposed into two sectors $T_p\varM^{\pm}$ corresponding to the eigenspaces of $J_p$ with eigenvalue $\pm \rmi$. The \emph{holomorphic tangent space} $T_p\varM^+$ is a complex vector space with complex dimension $m$ and is thus isomorphic to $\mathbb{C}^m$. It is spanned by $\{(\partial_j)_p\}_{j=1,\ldots,m}$ and corresponds to those directional derivatives that annihilate anti-holomorphic functions $f:\varM\to\mathbb{C}$. Similarly, the \emph{anti-holomorphic tangent space} $T_p\varM^-$ is spanned by $\{(\opartial_{\overline{\jmath}})_p\}_{j=1,\ldots,m}$ and annihilates holomorphic functions $f:\varM\to \mathbb{C}$. By extending these definitions to the whole tangent bundle $T\varM$, the holomorphic and anti-holomorphic tangent bundles $T\varM^{\pm}$ are obtained so that $T\varM^{\mathbb{C}}=T\varM^{+}\oplus T\varM^{-}$.

The following two lemmas are used throughout this paper.
\begin{lemma} Let $\varM$ be a complex manifold with $\dim \varM=m$ and with atlas $\{(\manifold{U}_i,\phi_i)\}$. Any open subset $\manifold{N}$ of $\varM$ is a complex manifold with $\dim \manifold{N}=\dim \varM$ and atlas $\{(\manifold{U}_i\cap \manifold{N}, \phi_i)\}$. In addition, the tangent space $T_p \manifold{N}$, complexified tangent space $T_p\manifold{N}^{\mathbb{C}}$ and (anti-)holomorphic tangent space $T_p\manifold{N}^{\pm}$ equal the corresponding tangent spaces $T_p\varM$, $T_p\varM^{\mathbb{C}}$ and $T_p\varM^{\pm}$ for any $p\in\manifold{N}$. 
\label{lemma:var:opensubsetmanifold}
\end{lemma}
\begin{lemma} The Cartesian product $\varM \times \manifold{N}$ of two complex manifolds $\varM$ and $\manifold{N}$ with $\dim \varM = m$ and $\dim \manifold{N}=n$, is a complex manifold with $\dim \varM\times\manifold{N}=m+n$.
\label{lemma:var:productmanifold}
\end{lemma}

Next, we can endow a complex manifold with a metric. 
\begin{definition}[Hermitian metric]
Let the complex manifold $\varM$ have a Riemannian metric $G$, \textit{i.e.}\ for every $p\in\varM$, $G_p$ is a symmetric positive-definite bilinear form on the real vector space $T_p \varM$. Due to its linearity, it can trivially be extended to $T_p\varM^{\mathbb{C}}$. If $G_p(J_p X,J_p Y)=G_p(X,Y)$ for every $X,Y\in T_p\varM^{\mathbb{C}}$ and every $p\in\varM$, then $G$ is called a \emph{Hermitian metric} and $\varM$ is a \emph{Hermitian manifold}.
\end{definition}
Having a set of complex coordinates $\vz\in\mathbb{C}^m$ for $\varM$, it can easily be checked that a Hermitian metric $G$ is of the form
\begin{equation}
G=g_{i\overline{\jmath}}(\ovz,\vz) \rmd z^i \otimes \overline{\rmd z}^{\overline{\jmath}}+g_{\overline{\imath}j}(\ovz,\vz) \overline{\rmd z}^{\overline{\imath}} \otimes \rmd z^i
\end{equation}
with $g_{i\overline{\jmath}}=g_{\overline{\jmath}i}$ due to the symmetry property of the Riemannian metric and $g_{ij}=0$, $g_{\overline{\imath}\overline{\jmath}}=0$. In addition, we have that $g_{i\overline{\jmath}}=\overline{g_{j\overline{\imath}}}$. A Hermitian metric $G_p$ at $p\in\varM$ defines an inner product ---a positive-definite Hermitian form--- $g_p$ on $T_p\varM^+$ by
\begin{equation}
g_p\colon T_p\varM^+\times T_p\varM^+\to \mathbb{C}\colon (X,Y)\mapsto g_p(X,Y)=G_p(X,\overline{Y}).
\end{equation}
Note that $\overline{Y}\in T_p\varM^{-}$. If $(\partial_i)_p$ is chosen as basis for $T_p\varM^{+}$, the matrix notation for $g_p$ is $g_{\overline{\imath},j}$, \textit{i.e.}\ $g_p(X,Y)=\overline{Y}^{\overline{\imath}} g_{\overline{\imath}j} X^j$. Vice versa, any inner product $g_p$ on $T_p\varM^{+}$ defines a Hermitian metric $G$. Note that standard notation is $g$ for the full Riemannian metric of $T_p\varM^{\mathbb{C}}$ and $h$ for the Hermitian inner product on $T_p\varM^+$. We do not adhere to this convention because $h$ typically represents a Hamiltonian (density) in the quantum literature, which would result in a source of confusion. In addition, as explained in the next subsection, we are mainly interested in the elements of $T_p\varM^{+}$ throughout the remainder of this paper, and  simply refer to $g$ as representing the metric. As noted, $g$ completely determines the full Hermitian metric $G$.

If $\varM$ is a Hermitian manifold with Hermitian metric $G$, the \emph{K\"{a}hler form} of $G$ is defined as the $2$-form $\Omega$ with prescription
\begin{equation}
\Omega_p\colon T_p\varM^{\mathbb{C}}\times T_p\varM^{\mathbb{C}}\to \mathbb{C}\colon (X,Y)\mapsto \Omega_p(X,Y)=G_p(J_p X,Y).
\end{equation}
In coordinates $\vz\in\mathbb{C}^m$, the $2$-form $\Omega$ is given by $\Omega=\rmi g_{i,\overline{\jmath}}\rmd z^i\wedge \overline{\rmd z}^{\overline{\jmath}}$ and can seen to be real ($\overline{\Omega}=\Omega$). An interesting subclass of Hermitian manifolds are the so-called \emph{K\"{a}hler manifolds}.
\begin{definition}[K\"{a}hler manifold]
Let $\varM$ be a Hermitian manifold with Hermitian metric $G$ and corresponding K\"{a}hler form $\Omega$. The manifold $\varM$ is a K\"{a}hler manifold if $\Omega$ is closed, \textit{i.e.}\ $\rmd \Omega=0$. The corresponding metric is called a K\"{a}hler metric.
\end{definition}
In coordinates $\vz\in\mathbb{C}^m$, a K\"{a}hler manifold satisfies $\partial_i g_{j\overline{k}}=0$ and $\opartial_{\overline{\imath}} g_{j\overline{k}}=0$. Locally, this implies the existence of a \emph{K\"{a}hler potential} $K(\ovz,\vz)$ such that
\begin{equation}
g_{i\overline{\jmath}}(\ovz,\vz)=\partial_i \opartial_{\overline{\jmath}} K(\ovz,\vz).
\end{equation}
By also defining the inverse metric such that
\begin{align}
g^{i\overline{\jmath}}(\ovz,\vz)g_{\overline{\jmath}k}(\ovz,\vz)&=\delta^i_k,&
g_{\overline{\imath}j}(\ovz,\vz) g^{j\overline{k}}(\ovz,\vz)&=\delta^{\overline{k}}_{\overline{\imath}},&
\end{align}
one can easily derive the Levi-Civita connection. The only non-vanishing components of the connection are given by
\begin{equation}
\begin{split}
\Gamma_{ij}^{\ \;k}(\ovz,\vz)=g^{k\overline{m}}(\ovz,\vz) \partial_{i} g_{\overline{m}j}(\ovz,\vz)=g^{k\overline{m}}(\ovz,\vz)\opartial_{\overline{m}}\partial_{i}\partial_{j} K(\ovz,\vz),\\
\Gamma_{\overline{\imath} \overline{\jmath}}^{\ \;\overline{k}}(\ovz,\vz)=g^{\overline{k}m}(\ovz,\vz) \opartial_{\overline{\imath}}g_{\overline{\jmath}m}(\ovz,\vz)=g^{m \overline{k}}(\ovz,\vz)\partial_{m}\opartial_{\overline{\imath}}\opartial_{\overline{\jmath}}K(\ovz,\vz).
\end{split}
\end{equation}
Because the Levi-Civita connection has no non-zero mixed components, the Riemannian geometry of K\"{a}hler manifolds is compatible with the complex structure, \textit{i.e.} holomorphic tangent vectors $X\in\mathbb{T}\varM^+$ are parallel transported into holomorphic tangent vectors. Finally, the only non-zero components of the Riemann tensor are
\begin{equation}
\begin{split}
R_{i\overline{\jmath} k\overline{l}}&=g_{\overline{l} m}(\ovz,\vz) \opartial_{\overline{\jmath}}\Gamma_{ik}^{\ \;m}(\ovz,\vz)\\
&=\opartial_{\overline{\jmath}}\partial_{i}\opartial_{\overline{l}}\partial_{k}K(\ovz,\vz)-(\opartial_{\overline{\jmath}}\opartial_{\overline{l}}\partial_{m} K(\ovz,\vz)) g^{m\overline{n}}(\ovz,\vz)(\partial_{i}\opartial_{\overline{n}}\partial_{k} K(\ovz,\vz))\\
&=\opartial_{\overline{\jmath}}\partial_{i}g_{\overline{l}k}(\ovz,\vz)-\Gamma_{\overline{l}\overline{\jmath}}^{\ \;\overline{m}}(\ovz,\vz)g_{\overline{m}n}(\ovz,\vz)\Gamma_{ik}^{\ \;n}(\ovz,\vz)\end{split}
\end{equation}
in combination with
\begin{equation}
R_{\overline{\jmath}i k \overline{l}}=R_{i \overline{\jmath} \overline{l} k}=-R_{i\overline{\jmath}k\overline{l}}=-R_{\overline{\jmath}i \overline{l}k}=-R_{k\overline{l}i\overline{\jmath}}.\label{eq:var:curvaturesym}
\end{equation}

As a final topic of this introduction, we introduce the concept of a complex Lie group
\begin{definition}[Complex Lie Group] A complex Lie group $\mathsf{G}$ is a Lie group $\mathsf{G}$ that has the structure of a complex manifold and for which the group operations of multiplication
\begin{equation}
\mathsf{G}\times\mathsf{G}\to \mathsf{G}\colon (G_1,G_2) \mapsto G_1 G_2
\end{equation}
and taking the inverse
\begin{equation}
\mathsf{G}\to \mathsf{G}\colon G \mapsto G^{-1}
\end{equation}
are holomorphic maps.
\end{definition}
Clearly, the complex general linear group $\mathsf{GL}(D,\mathbb{C})$ of invertible complex $D\times D$ matrices is a complex Lie group with complex dimension $\dim \mathsf{GL}(D,\mathbb{C})=D^2$. Unlike the real general linear group $\mathsf{GL}(D,\mathbb{R})$, the complex case $\mathsf{GL}(D,\mathbb{C})$ is connected. We conclude this introductory review with another lemma.
\begin{lemma} The direct product group $\mathsf{G}_1\times \mathsf{G}_2$ of two complex Lie groups $\mathsf{G}_1$ and $\mathsf{G}_2$ is a complex Lie group with $\dim \mathsf{G}_1\times \mathsf{G}_2=\dim \mathsf{G}_1 + \dim \mathsf{G}_2$.\label{lemma:var:productgroup}
\end{lemma}

\subsection{Complex geometry of Hilbert space}
\label{ss:var:hilbert}
This paper is concerned with the study of quantum systems, the state of which is described by a vector in some Hilbert space $\hilbert$. By introducing a basis $\{\ket{e_i}\}_{i=1,\ldots,\dim \hilbert}$ in a finite-dimensional Hilbert space $\hilbert$ and identifying states $\ket{\Psi}=z^i\ket{e_i}\in \hilbert$ with $\vz \in\mathbb{C}^{\dim \hilbert}$, we obtain a biholomorphism between $\hilbert$ and the complex Euclidean space $\mathbb{C}^{\dim \hilbert}$. Alternatively, we can interpret the relation between $\ket{\Psi}$ and $\vz$ as a globally defined coordinate chart, by which $\hilbert$ satisfies the conditions for being a complex manifold. 

Physically interesting functions on $\hilbert$ are of the form
\begin{equation}
f_O\colon \hilbert \to \mathbb{C}\colon \ket{\Psi}\mapsto f_O(\ket{\Psi})=\braket{\Psi|\operator{O}|\Psi},\label{eq:var:deff}
\end{equation}
or the normalized version
\begin{equation}
\tilde{f}_O\colon \hilbert \to \mathbb{C}\colon \ket{\Psi}\mapsto \tilde{f}_O(\ket{\Psi})=\frac{\braket{\Psi|\operator{O}|\Psi}}{\braket{\Psi|\Psi}}.\label{eq:var:defnormalizedf}
\end{equation}
If $\ket{\Psi}=z^i\ket{e_i}$ and $\ket{\Phi}=w^i\ket{e_i}$ with $\vz,\vw\in\mathbb{C}^{\dim \hilbert}$, then the action of the holomorphic tangent vector $w^i \left(\partial_i\right)_{\ket{\Psi}}$ on $f_O$ is given by
\begin{displaymath}
w^i \left(\partial_i\right)_{\ket{\Psi}} f_O = \braket{\Psi|\operator{O}|\Phi}.
\end{displaymath}
Consequently, we can identify $w^i(\partial_i)_{\ket{\Psi}}\in T_{\ket{\Psi}}\hilbert^+$ with $\ket{\Phi}\in\hilbert$ for any base point $\ket{\Psi}\in\hilbert$. We thus conclude that $T_{\ket{\Psi}}\hilbert^+ \cong \hilbert$ for any $\ket{\Psi}\in\hilbert$. A similar argument shows that $T_{\ket{\Psi}}\hilbert^- \cong \hilbert^\ast$. From this identification, $T_{\ket{\Psi}}\hilbert^+$ can be endowed with a natural inner product
\begin{equation}
g^{(\hilbert)}_{\ket{\Psi}}\colon \hilbert\times \hilbert\to\mathbb{C}\colon (\ket{\Phi_1},\ket{\Phi_2})\mapsto g^{(\hilbert)}_{\ket{\Psi}}(\ket{\Phi_1},\ket{\Phi_2})\defis \braket{\Phi_2|\Phi_1},
\end{equation}
the matrix elements of which are given by
\begin{equation}
g^{(\hilbert)}_{\overline{\imath}j}=g^{(\hilbert)}_{j\overline{\imath}}=\braket{e_{\overline{\imath}}|e_j},
\end{equation}
with $\bra{e_{\overline{\imath}}}$ the linear functional associated to the basis vector $\ket{e_i}$ according to the Riesz representation theorem. In case of an orthonormal basis, the full Riemannian metric $G^{(\hilbert)}$ for the tangent space $T_{\ket{\Psi}}\hilbert^{\mathbb{C}}$ reduces to
\begin{equation}
G^{(\hilbert)}_{\ket{\Psi}}=\sum_{i=1}^{\dim \hilbert}\left(\rmd z^i\right)_{\ket{\Psi}}\otimes \left(\overline{\rmd z}^{\overline{\imath}}\right)_{\ket{\Psi}}\qquad \Leftrightarrow \qquad \left[G^{(\hilbert)}_{\ket{\Psi}}\right] = \begin{bmatrix} 0 & \operator{\one}\\ \operator{\one} & 0 \end{bmatrix},
\end{equation}
with $\operator{\one}$ the identity operator on $\hilbert$. Clearly, $\hilbert$ is globally flat and the Levi-Civita connection and Riemann curvature tensor vanish everywhere. In particular, it is K\"{a}hler manifold with $K(\ovz,\vz)=\braket{\Psi|\Psi}=\oz^{\overline{\imath}} g_{\overline{\imath} j} z^j$. 

While it is common practice to study quantum mechanics using state vectors in affine Hilbert space $\hilbert$, physical states correspond to rays of such vectors and should be identified with the elements of the projective Hilbert space $P(\hilbert)\defis \hilbert/\mathsf{GL}(1,\mathbb{C})\cong \mathbb{C}P^{\dim \hilbert-1}$. Here, $\mathsf{GL}(1,\mathbb{C})$ is the multiplicative abelian group of norm and phase changes which acts on $\hilbert$ as
\begin{equation}
\Gamma\colon\hilbert\times \mathsf{GL}(1,\mathbb{C})\to \hilbert\colon (\ket{\Psi},\lambda) \mapsto \Gamma(\ket{\Psi},\lambda)=\lambda \ket{\Psi}.
\end{equation}
Indeed, physical results are related to normalized expectation values $\tilde{f}_O$ as in Eq.~(\ref{eq:var:defnormalizedf}), which is invariant under the action of $\mathsf{GL}(1,\mathbb{C})$ and can thus be restricted to $P(\hilbert)$. While it is often easier to work in the affine space $\hilbert$, it turns out that using the projective structure of state space is required when studying uniform matrix product states in the thermodynamic limit, in order to avoid a number of unpleasant divergences.

The ray of states containing a vector $\ket{\Psi}\in\hilbert$ is denoted as $[\ket{\Psi}]\in P(\hilbert)$. However, in all calculations, we denote the elements $[\ket{\Psi}]\in P(\hilbert)$ using any representative $\ket{\Psi'}=\lambda\ket{\Psi}$ of the ray $[\ket{\Psi}]$ and consider these as a set of homogeneous coordinates for the elements of $P(\hilbert)$, rather than trying to define a new set of stereographic or orthographic coordinates which are not globally well defined. Often, formulas greatly simplify when choosing a normalized representative, \textit{i.e.}\ $\braket{\Psi'|\Psi'}=1$.

 The tangent space $T_{[\ket{\Psi}]}P(\hilbert)$ is the quotient vector space $T_{\ket{\Psi}}\hilbert/\!\sim\ \cong\hilbert/\!\sim$, where two vectors $\ket{\Phi_1},\ket{\Phi_2}\in\hilbert$ are equivalent  ($\ket{\Phi_1}\sim\ket{\Phi_2}$) if $\ket{\Phi_1}-\ket{\Phi_2}=\alpha \ket{\Psi}$ for some $\alpha\in\mathbb{C}$. We can find a unique representative $\ket{\Phi}\in\hilbert$ for every tangent vector in $T_{[\ket{\Psi}]}P(\hilbert)$ by imposing a condition such as
\begin{equation}
\braket{\Psi|\Phi}=0,\label{eq:var:condtangentvectorsprojectivehilbert}
\end{equation}
and we obtain $T_{[\ket{\Psi}]}P(\hilbert)\cong \hilbert^{\perp}_{\ket{\Psi}}$, where $\hilbert^{\perp}_{\ket{\Psi}}$ is the orthogonal complement of the one-dimensional subspace spanned by $\ket{\Psi}$. A representation $\ket{\Phi}$ that does not satisfy this condition can be transformed into one that does by acting with
\begin{equation}
\operator{P}_{\ket{\Psi}}^{\perp}=\operator{\one}-\operator{P}_{\ket{\Psi}}=\one-\frac{\ket{\Psi}\bra{\Psi}}{\braket{\Psi|\Psi}},
\end{equation}
the orthogonal projector onto $\hilbert^{\perp}_{\ket{\Psi}}$. Projective Hilbert space $P(\hilbert)$ is still a K\"{a}hler manifold if one endows it with the \emph{Fubini-Study metric} $\tilde{g}^{(\hilbert)}_{\ket{\Psi}}$ for $T_{[\ket{\Psi}]}P(\hilbert)^{+}$, which is defined as \cite{Nakahara:2003ys}
\begin{equation}
\tilde{g}^{(\hilbert)}_{\ket{\Psi}}(\ket{\Phi_1},\ket{\Phi_2})=\frac{\braket{\Psi|\Psi}\braket{\Phi_2|\Phi_1}-\braket{\Phi_2|\Psi}\braket{\Psi|\Phi_1}}{\braket{\Psi|\Psi}^2}=\frac{\braket{\Phi_2|\operator{P}_{\ket{\Psi}}^{\perp}|\Phi_1}}{\braket{\Psi|\Psi}}.
\end{equation}
The Fubini-Study metric corresponds to the infinitesimal version of the normalized overlap $\braket{\Psi|\Psi'}/(\braket{\Psi|\Psi}\braket{\Psi'|\Psi'})^{1/2}$ of two quantum states. Note that, using representations $\ket{\Phi_{1,2}}$ that satisfy Eq.~(\ref{eq:var:condtangentvectorsprojectivehilbert}) and a base point representation $\ket{\Psi}$ satisfying $\braket{\Psi|\Psi}=1$, the Fubini-Study metric reduces to the ordinary metric $g_{\ket{\Psi}}^{(\hilbert)}$. In terms of the homogeneous coordinates $\ket{\Psi}$, the corresponding K\"{a}hler potential is given by $\tilde{K}=\log N=\log \braket{\Psi|\Psi}$.

Henceforth, we are most interested in the physical states living in the holomorphic tangent space. Often, we fail to mention the restriction to the holomorphic part and just refer to this as the tangent space containing tangent vectors. For submanifolds $\varM\subset\hilbert$, the (holomorphic) tangent space at some point $\ket{\Psi}\in\varM$ satisfies $T_{\ket{\Psi}}\varM^{+}\subset \hilbert$, and we always denote (holomorphic) tangent vectors of complex submanifolds $\varM\subset \hilbert$ as vectors $\ket{\ }\in \hilbert$, rather than as a directional derivatives. While the latter is the standard approach in modern differential geometry and allows one to study the manifold intrinsically, the former facilitates a geometric interpretation of our results, making them (hopefully) more accessible to people with less background in differential geometry.

\subsection{Complex variational manifolds}
\label{ss:var:varmanifold}
For many interesting systems, the dimension of the Hilbert space is too large to allow for an exact solution, neither analytically nor numerically. One powerful approach to obtain approximate results is by restricting to a set of variational ansatz states $\ket{\Psi(\vz)}$ depending on a number of parameters $z^i$ with $i=1,\ldots,m$ where typically $m\ll \dim \hilbert$ so that this set is easier to handle. Throughout this paper, we restrict to ansatzes $\Psi$ for which the parameters $\vz$ can take complex values in some open subset $\manifold{U}$ of $\mathbb{C}^m$, either directly or via analytic continuation.  Furthermore, we impose the additional restriction that the map $\Psi:\manifold{U}\to \hilbert$ is holomorphic.
\begin{definition}[Variational subset] The variational subset $\varM$ corresponding to a variational ansatz $\Psi:\manifold{U}\subset \mathbb{C}^m\to \hilbert$ is defined as the image of $\Psi$, that is
\begin{equation}
\varM\defis \Psi(\manifold{U})= \{\ket{\Psi(\vz)} | \vz \in \manifold{U}\}.
\end{equation}
\end{definition}
The restriction to holomorphic maps $\Psi$ is not sufficient to give any differentiable structure to $\varM$. Without further conditions on $\Psi$, we can not conclude that the variational subset is a manifold, \textit{e.g.}\ we cannot exclude the possibility that $\varM$ intersects itself. Throughout the remainder of this section, we assume that $\Psi$ is injective, so that $\varM$ is biholomorphic to the complex manifold $\manifold{U}$  (see Lemma~\ref{lemma:var:opensubsetmanifold}) and is therefore an embedded complex submanifold of $\hilbert$. 

Under the injectivity assumption, a holomorphic inverse map $\Psi^{-1}\colon \varM\to \manifold{U}\subset\mathbb{C}^m$ can be defined, so that $(\varM,\Psi^{-1})$ can be interpreted as a global coordinate chart for $\varM$. However, we refrain from doing so, since the injectivity restriction on $\Psi$ will be lifted in the next section, in which case we no longer have a coordinate chart $\Psi^{-1}$.

With a slight abuse of notation, we use the same symbol $\Psi$ to denote the associated antiholomorphic map $\Psi\colon\manifold{U}\to\hilbert^\ast\colon\vz\mapsto\bra{\Psi(\ovz)}$, with $\hilbert^\ast$ the dual space of linear functionals on $\hilbert$. We explicitly denote the antiholomorphic dependence of the bras on the variational parameters. For the holomorphic map $\Psi\colon\manifold{U}\to \varM$, we also define the tangent map $\rmd \Psi_{\vz}\colon T_{\vz}\manifold{U}^+\to T_{\ket{\Psi(\vz)}}\varM^{+}$ with $T_{\vz}\manifold{U}^+\equiv (T_{\vz}\mathbb{C}^{m})^{+}\cong \mathbb{C}^m$ (see Lemma~\ref{lemma:var:opensubsetmanifold}) and $T_{\ket{\Psi(\vz)}}\varM^{+}\subset T_{\ket{\Psi(\vz)}}\hilbert^{+}\cong \hilbert$. We define a pushforward $\rmd\Psi_{\vz}(w^i\left.\partial_i\right|_{\vz})$ of tangent vectors $w^i\left.\partial_i\right|_{\vz}\in(T_{\vz}\mathbb{C}^m)^{+}$ to the holomorphic tangent space $(T_{\ket{\Psi(\vz)}}\varM)^{+}\subset \hilbert$. For the partial derivatives of $\Psi$ at a point $\vz\in\manifold{U}$, we introduce the notation
\begin{equation}
\partial_i \Psi\colon \manifold{U} \to \hilbert\colon\vz \mapsto \ket{\partial_i\Psi(\vz)}\defis \partial_i \left.\ket{\Psi(\vz)}\right|_{\vz}.
\end{equation}
The tangent map $\rmd \Psi_{\vz}$ is then given by the prescription
\begin{equation}
\rmd \Psi_{\vz}\colon \mathbb{C}^m\to T_{\ket{\Psi(\vz)}}\varM^{+}\subset\hilbert\colon w^i(\partial_i)_{\vz}\mapsto w^i\ket{\partial_i \Psi(\vz)}\defis \ket{\Phi(\vw;\vz)},
\end{equation}
which defines a new map $\Phi\colon\mathbb{C}^m\times\mathbb{C}^m \to\hilbert$. Finally, we also define the holomorphic tangent bundle $T\varM^{+}\subset T\hilbert^{+}\cong \hilbert\times \hilbert$. The pushforward of $\Psi$ induces a bundle map, \textit{i.e.}\ a map between between the tangent bundles $T\manifold{U}^{+}\cong \mathbb{C}^m\times\manifold{U}$ and $T\varM^{+}\subset \hilbert \times \hilbert$, which acts as
\begin{equation}
\rmd\Psi\colon \mathbb{C}^m\times\manifold{U}\to T\varM^{+}\colon (\vw;\vz)\mapsto (\ket{\Phi(\vw;\vz)}; \ket{\Psi(\vz)}).
\end{equation}

\subsection{Metric, connection and curvature in affine Hilbert space}
We can induce the standard Hilbert metric onto $\varM$ and then define a \emph{pullback metric} $g=\Psi^{\ast}g^{(\hilbert)}$ that is given by
\begin{equation}
g_{\vz}\colon \mathbb{C}^m\times\mathbb{C}^m\to\mathbb{C}\colon
(\vw_{1},\vw_{2})\mapsto g_{\vz}(\vw_{1},\vw_{2})=g_{\ket{\Psi(\vz)}}^{(\hilbert)}(\rmd\Psi_{\vz}(\vw_1),\rmd\Psi_{\vz}(\vw_2)).\label{eq:var:pullbackmetric}
\end{equation}
We can further simplify $g_{\vz}(\vw_{1},\vw_{2})$ as
\begin{displaymath}
g_{\vz}(\vw_{1},\vw_{2})=\braket{\Phi(\ovw_2;\ovz)|\Phi(\vw_1;\vz)}=\ovw_{2}^{\overline{\imath}} \braket{\opartial_{\overline{\imath}}\Psi(\ovz)|\partial_j\Psi(\vz)} \vw_1^j.
\end{displaymath}
We now switch to a coordinate-based notation, and define the entries of the Hermitian metric as
\begin{equation}
g_{\overline{\imath}j}(\ovz,\vz)=\braket{\opartial_{\overline{\imath}} \Psi(\ovz)|\partial_j\Psi(\vz)},
\end{equation}
which is indeed Hermitian ($g_{\overline{\imath}j}=\overline{g_{\overline{\jmath}i}}$) and positive definite, due to the injectivity of $\Psi$. Note that we still use the term pullback metric when the injectivity of $\Psi$ is abandoned and $g$ can become degenerate. At that point, the pullback metric $g$ is no longer a proper metric according to the strict definition. We discuss the consequences when encountering this issue for the first time in the next subsection. It can be checked that in combination with the metric $g$ defined in Eq.~(\ref{eq:var:pullbackmetric}), $\varM$ is a K\"{a}hler manifold with K\"{a}hler potential $K(\ovz,\vz)=N(\ovz,\vz)$, with $N(\ovz,\vz)\defis \braket{\Psi(\ovz)|\Psi(\vz)}$ the norm function. Indeed, it was shown that any complex submanifold of a K\"{a}hler manifold is also a K\"{a}hler manifold \cite{Nakahara:2003ys}. We define an inverse metric with nonzero entries $g^{i\overline{\jmath}}(\ovz,\vz)=g^{\overline{\jmath}i}(\ovz,\vz)$ such that $g^{i\overline{\jmath}}(\ovz,\vz) g_{\overline{\jmath}k}(\ovz,\vz)=\delta^i_{\ k}$ and can then easily derive the Levi-Civita connection using the results from Subsection~\ref{ss:var:complexgeometry}. The only non-vanishing components of the connection are given by
\begin{align}
\Gamma_{ij}^{\ \;k}&=g^{k\overline{m}} \braket{\opartial_{\overline{m}}\Psi|\partial_{i}\partial_{j} \Psi},&
\Gamma_{\overline{\imath} \overline{\jmath}}^{\ \;\overline{k}}&=g^{m \overline{k}} \braket{\opartial_{\overline{\imath}}\opartial_{\overline{\jmath}} \Psi | \partial_{m} \Psi},
\end{align}
where we have omitted the arguments $\vz$ and $\ovz$ for the sake of brevity. Similarly, the only non-zero components of the Riemann tensor are 
\begin{equation}
R_{i\overline{\jmath} k\overline{l}}=\braket{\opartial_{\overline{\jmath}}\opartial_{\overline{l}}\Psi|\partial_{i}\partial_{k}\Psi}-\braket{\opartial_{\overline{\jmath}}\opartial_{\overline{l}}\Psi|\partial_{m}\Psi}g^{m\overline{n}} \braket{\opartial_{\overline{n}}\Psi|\partial_{i}\partial_{k}\Psi}
\end{equation}
in combination with the symmetries in Eq.~(\ref{eq:var:curvaturesym}).

\subsection{Metric, connection and curvature in projective Hilbert space}
\label{ss:var:projective}
When trying to associate a manifold $\tilde{\varM}\subset P(\hilbert)$ to the original manifold $\varM\subset\hilbert$, there are a number of possibilities depending on the nature of $\varM$, and thus on the nature of the map $\Psi$. Even if the variational subset $\varM$ is guaranteed to be a manifold, this does not automatically imply that there are no singularities or self-intersections in the set
\begin{equation}
\tilde{\varM}=\{[\ket{\Psi(\vz)}], \forall \vz\in\manifold{U}\}.
\end{equation}
We now assume that the map $\Psi$ is sufficiently regular in order to be able to define $\tilde{\varM}$ as a complex manifold for some open domain $\manifold{U}$ of the parameter space. 

For example, if the map $\Psi$ satisfies
\begin{displaymath}
\forall \vz\in\manifold{U}, \forall \lambda \in \mathbb{C}\colon \lambda \ket{\Psi(\vz)}\in \varM \Leftrightarrow\lambda=1,
\end{displaymath}
so that $\varM$ contains at most a single representative on every ray of vectors in $\hilbert$, then we can immediately define a map $\tilde{\Psi}\colon\manifold{U}\to P(\hilbert)$ by setting
\begin{equation}
\tilde{\Psi}\colon\manifold{U}\to P(\hilbert)\colon\vz \mapsto \left[ \ket{\Psi(\vz)} \right]
\end{equation}
with $\left[ \ket{\Psi(\vz)} \right]\in P(\hilbert)$ the ray to which $\ket{\Psi(\vz)}$ belongs. The map $\tilde{\Psi}$ is still injective and we can induce the Fubini-Study metric onto $\varM$ and define a positive definite pullback metric $\tilde{g}$ on $\manifold{U}$.

Alternatively, it could be the case that $\varM$ contains parts of rays of vectors, such that
\begin{equation}
\forall \vz\in\manifold{U}, \exists \epsilon_{\vz} >0: \alpha \in B_{\epsilon_{\vz}}(0)\Rightarrow
\rme^{\alpha} \ket{\Psi(\vz)}\in \varM\label{eq:var:projectivecase2}
\end{equation}
with $B_\epsilon(0)=\{\alpha \in\mathbb{C}|\vert \alpha\vert<\epsilon\}$. A particular subcase of this type is when the map $\Psi$ satisfies $\ket{\Psi(\lambda \vz)} = f(\lambda) \ket{\Psi(\vz)}$ for any $\lambda\in\mathsf{GL}(1,\mathbb{C})$ for which $\lambda \vz \in \manifold{U}$, where $f$ is necessarily holomorphic. In that case we could define an injective map from $P(\manifold{U})\subset \mathbb{C}P^{m-1} \to P(\hilbert)$ that equals $\Psi$ when expressed in terms of homogeneous coordinates for both projective spaces. However, from a practical point of view it is more convenient to have a parameter space with an affine structure, and to work with the map $\tilde{\Psi}\colon\manifold{U}\to P(\hilbert)\colon\vz \mapsto \left[ \ket{\Psi(\vz)} \right]$, which is no longer injective in the case where $\Psi$ satisfies Eq.~\eqref{eq:var:projectivecase2}.

The pullback $\tilde{g}_{\vz}$ of the induced Fubini-Study metric $\tilde{g}^{(\hilbert)}_{\ket{\Psi(\vz)}}$  to $T_{\vz} \manifold{U}^{+}=\mathbb{C}^m$ is given by
\begin{multline}
\tilde{g}_{\vz}\colon\mathbb{C}^m\times\mathbb{C}^m\to \mathbb{C}\colon \\
(\vw_1,\vw_2)\mapsto \tilde{g}_{\vz}(\vw_1,\vw_2)=\frac{\braket{\Phi(\ovw_2;\ovz)|\Phi(\vw_1;\vz)}}{\braket{\Psi(\ovz)|\Psi(\vz)}}-\frac{\braket{\Phi(\ovw_2;\ovz)|\Psi(\vz)}\braket{\Psi(\ovz)|\Phi(\vw_1;\vz)}}{\braket{\Psi(\ovz)|\Psi(\vz)}^2}
\end{multline}
Correspondingly, the entries $\tilde{g}_{\overline{\imath}j}$ are given by
\begin{equation}
\tilde{g}_{\overline{\imath}j}(\ovz,\vz)=\frac{\braket{\opartial_{\overline{\imath}}\Psi(\ovz)|\partial_j\Psi(\ovz)}}{\braket{\Psi(\ovz)|\Psi(\vz)}}-\frac{\braket{\opartial_{\overline{\imath}}\Psi(\ovz)|\Psi(\vz)}\braket{\Psi(\ovz)|\partial_j\Psi(\ovz)}}{\braket{\Psi(\ovz)|\Psi(\vz)}^2},
\end{equation}
and originate from a K\"{a}hler potential
\begin{equation}
\tilde{K}(\ovz,\vz)=\log\left(N(\ovz,\vz)\right)=\log\left(\braket{\Psi(\ovz)|\Psi(\vz)}\right)
\end{equation}
In the first case, where $\varM$ contains at most a single representative on every ray of vectors in $\hilbert$, $\tilde{g}_{\overline{\imath}j}$ is still a positive definite matrix and can readily be inverted. In the second case, where $\varM$ contains parts of rays of vectors, the pullback metric is degenerate and no longer constitutes a proper Riemannian metric, due to the non-injectivity of $\tilde{\Psi}$. Hence, there is no unique way to define an inverse metric. A more formal treatment in terms of fiber bundles and bundle connections is given in the next section, where even the map $\Psi$ itself is non-injective. For now, we continue in a more intuitive way. Given the validity of Eq.~\eqref{eq:var:projectivecase2}, there must exist some $\vv(\vz)\in T_{\vz}\manifold{U}^{+}\cong\mathbb{C}^m$ such that $\ket{\Phi(\vv(\vz);\vz)}=v^i(\vz) \ket{\partial_i \Psi(\vz)}=\ket{\Psi(\vz)}$. For the particular subcase where $\ket{\Psi(\lambda\vz)}=f(\lambda)\ket{\Psi(\vz)}$, we obtain $\vv(\vz)=\vz/f'(1)$ by differentiating this relation with respect to $\lambda$ at $\lambda=1$. It can easily be checked that $\tilde{g}_{\overline{\imath} j}(\ovz,\vz) v^j(\vz) = 0$. Hence, we can define an equivalence relation between tangent vectors $\vw\in T_{\vz}\manifold{U}^+\cong\mathbb{C}^m$ as $\vw_1\sim \vw_2$ if $\vw_1-\vw_2 = \alpha \vv(\vz)$ for some $\alpha\in\mathbb{C}$. The tangent space $T_{[\ket{\Psi(\vz)}]}\tilde{\varM}$ is isomorphic to $\mathbb{C}^m/\!\sim$ and we need to impose a condition in order to associate a unique vector $\vw\in\mathbb{C}^m$ to tangent vectors in $T_{[\ket{\Psi(\vz)}]}\tilde{\varM}$. Unlike in the total Hilbert space, working with the orthogonal complement of $\vv(\vz)$ would not be a natural choice, as there is no intrinsic notion of orthogonality and this choice depends on the parameterization of the manifold. Instead, the transition from affine Hilbert space to projective Hilbert space implies that the natural condition to impose on tangent vectors $\vw$ is that
\begin{equation}
\braket{\Psi(\ovz)|\Phi(\vw;\vz)}=\braket{\Psi(\ovz)|\partial_i \Psi(\vz)}w^i = 0.\label{eq:var:condtangentvectorsparameterspace}
\end{equation}
We can then define a pseudo-inverse of the pullback metric with entries $\tilde{g}^{i\overline{\jmath}}(\ovz,\vz)$ that satisfy
\begin{equation}
\tilde{g}^{i\overline{\jmath}}(\ovz,\vz) \tilde{g}_{\overline{\jmath}k}(\ovz,\vz) =\delta^i_k- v^i(\vz) \frac{\braket{\Psi(\ovz)|\partial_k\Psi(\vz)}}{\braket{\Psi(\ovz)|\Psi(\vz)}}.\label{eq:pseudoinversemetric}
\end{equation}
The right-hand side acts as a projector (not an orthogonal one) that transforms any vector $\vw$ into an equivalent one that satisfies the condition in Eq.~(\ref{eq:var:condtangentvectorsparameterspace}). Put differently, the pullback metric acts as a proper inner product in the subspace of $\mathbb{C}^m$ that satisfies Eq.~(\ref{eq:var:condtangentvectorsparameterspace}). 

We carefully proceed by introducing a Levi-Civita connection and a Riemann curvature tensor using the pullback metric and its pseudo-inverse rather than a proper metric and its inverse. We check the consistency of this approach at the end of this subsection. For the non-zero entries of the connection $\tilde{\Gamma}$, we obtain
\begin{align}
\tilde{\Gamma}_{ij}^{\ \;k}&=\tilde{g}^{k\overline{m}} \partial_{i} \tilde{g}_{\overline{m}j},&
\tilde{\Gamma}_{\overline{\imath} \overline{\jmath}}^{\ \;\overline{k}}&=\tilde{g}^{m \overline{k}} \opartial_{\overline{\imath}}\tilde{g}_{\overline{\jmath}m},
\label{eq:var:levicivita}
\end{align}
with
\begin{equation}
\begin{split}
\partial_{i} \tilde{g}_{\overline{m}j}=&\frac{\braket{\opartial_{\overline{m}}\Psi|\partial_{i}\partial_{j} \Psi}}{\braket{\Psi|\Psi}}-\frac{\braket{\opartial_{\overline{m}}\Psi|\Psi}\braket{\Psi|\partial_i\partial_j\Psi}}{\braket{\Psi|\Psi}^2}-\frac{\braket{\opartial_{\overline{m}}\Psi|\partial_i \Psi}\braket{\Psi|\partial_j\Psi}}{\braket{\Psi|\Psi}^2}\\
&-\frac{\braket{\opartial_{\overline{m}}\Psi|\partial_j \Psi}\braket{\Psi|\partial_i\Psi}}{\braket{\Psi|\Psi}^2}+2\frac{\braket{\opartial_{\overline{m}}\Psi|\Psi}\braket{\Psi|\partial_i\Psi}\braket{\Psi|\partial_j\Psi}}{\braket{\Psi|\Psi}^3}.
\end{split}\label{eq:var:levicivitaaux}
\end{equation}
As before, we have omitted the arguments $\vz$ and $\ovz$ for the sake of brevity. To compute the Riemann curvature tensor, we need an expression for $\opartial_{\overline{m}}\tilde{g}^{i\overline{l}}$, which can be obtained from applying $\opartial_{\overline{m}}$ to Eq.~(\ref{eq:pseudoinversemetric}) and multiplying with $\tilde{g}^{k\overline{l}}$, resulting in
\begin{displaymath}
\left(\opartial_{\overline{m}} \tilde{g}^{i\overline{\jmath}}\right)\left(\delta^{\overline{l}}_{\overline{\jmath}}-\frac{\braket{\opartial_{\overline{j}}\Psi|\Psi}}{\braket{\Psi|\Psi}}\overline{v}^{\overline{l}}\right)+\tilde{g}^{i\overline{\jmath}}(\opartial_{\overline{m}}\tilde{g}_{\overline{\jmath}k})\tilde{g}^{k\overline{l}}=-v^{i}\left(\delta^{\overline{l}}_{\overline{m}}-\frac{\braket{\opartial_{\overline{m}}\Psi|\Psi}}{\braket{\Psi|\Psi}}\overline{v}^{\overline{l}}\right),
\end{displaymath}
the unique solution to which is given by
\begin{equation}
\opartial_{\overline{m}} \tilde{g}^{i\overline{l}}=-\tilde{g}^{i\overline{\jmath}}(\opartial_{\overline{m}}\tilde{g}_{\overline{\jmath}k})\tilde{g}^{k\overline{l}}-v^i\delta^{\overline{l}}_{\overline{m}}.
\end{equation}
Here we used that
\begin{displaymath}
g_{\overline{\imath}j}v^j=0\qquad\mathrm{and}\qquad g^{i\overline{\jmath}}\braket{\opartial_{\overline{\jmath}}\Psi|\Psi}=0,
\end{displaymath}
which are the conditions that led to the precise form of the right hand side of Eq.~(\ref{eq:pseudoinversemetric}). Differentiating these identities further leads to
\begin{displaymath}
(\partial_{i} \tilde{g}_{\overline{m}j}) v^j = - \tilde{g}_{\overline{m}j} \partial_{i}v^j\qquad\mathrm{and}\qquad (\partial_{i} \tilde{g}_{\overline{m}j}) \overline{v}^{\overline{m}} = 0,
\end{displaymath}
since $\vv$ depends on $\vz$ holomorphically. With these identities at hand, we can show that the covariant derivative $\nabla_i$ annihilates the metric
\begin{equation}
\nabla_i \tilde{g}_{\overline{\jmath}k}\defis \partial_i \tilde{g}_{\overline{\jmath} k} - \tilde{\Gamma}^{\ \;l}_{ik}\ \tilde{g}_{\overline{\jmath} l}=0.
\end{equation}
Hence, the defining relation of the Levi-Civita connection is still preserved, even though we have been using a pseudo-inverse of a degenerate pullback metric. In addition, the identities above are required to prove that the entries of the Riemann curvature tensor $\tilde{R}$ are still given by
\begin{equation}
\begin{split}\tilde{R}_{i\overline{\jmath} k\overline{l}}&=\tilde{g}_{\overline{l} m} \opartial_{\overline{\jmath}}\tilde{\Gamma}_{ik}^{\ \;m}=\opartial_{\overline{\jmath}}\partial_{i}\tilde{g}_{\overline{l}k}-\left(\opartial_{\overline{\jmath}}\tilde{g}_{\overline{l}m}\right)\tilde{g}^{m\overline{n}}\left(\partial_i\tilde{g}_{\overline{n}k}\right)
\end{split}
\label{eq:var:riemmanprojective}
\end{equation}
with
\begin{equation}
\begin{split}
\opartial_{\overline{\jmath}}\partial_{i}\tilde{g}_{\overline{l}k}=&\frac{\braket{\opartial_{\overline{\jmath}}\opartial_{\overline{l}}\Psi|\partial_i\partial_k\Psi}}{\braket{\Psi|\Psi}}-\frac{\braket{\opartial_{\overline{\jmath}}\opartial_{\overline{l}}\Psi|\Psi}\braket{\Psi|\partial_i\partial_k\Psi}}{\braket{\Psi|\Psi}^2}\\
&-\frac{\braket{\opartial_{\overline{\jmath}}\Psi|\partial_{i}\Psi}\braket{\opartial_{\overline{l}}\Psi|\partial_k\Psi}}{\braket{\Psi|\Psi}^2}-\frac{\braket{\opartial_{\overline{\jmath}}\Psi|\partial_{k}\Psi}\braket{\opartial_{\overline{l}}\Psi|\partial_i\Psi}}{\braket{\Psi|\Psi}^2}\\
&-\frac{\braket{\opartial_{\overline{\jmath}}\opartial_{\overline{l}}\Psi|\partial_{k}\Psi}\braket{\Psi|\partial_i\Psi}}{\braket{\Psi|\Psi}^2}-\frac{\braket{\opartial_{\overline{\jmath}}\opartial_{\overline{l}}\Psi|\partial_{i}\Psi}\braket{\Psi|\partial_k\Psi}}{\braket{\Psi|\Psi}^2}\\
&-\frac{\braket{\opartial_{\overline{l}}\Psi|\partial_{i}\partial_{k}\Psi}\braket{\opartial_{\overline{\jmath}}\Psi|\Psi}}{\braket{\Psi|\Psi}^2}-\frac{\braket{\opartial_{\overline{\jmath}}\Psi|\partial_{i}\partial_{k}\Psi}\braket{\opartial_{\overline{l}}\Psi|\Psi}}{\braket{\Psi|\Psi}^2}\\
&+2\frac{\braket{\opartial_{\overline{\jmath}}\opartial_{\overline{l}}\Psi|\Psi}\braket{\Psi|\partial_i\Psi}\braket{\Psi|\partial_k\Psi}}{\braket{\Psi|\Psi}^3}+2\frac{\braket{\Psi|\partial_i\partial_k\Psi}\braket{\opartial_{\overline{\jmath}}\Psi|\Psi}\braket{\opartial_{\overline{l}}\Psi|\Psi}}{\braket{\Psi|\Psi}^3}\\
&+2\frac{\braket{\opartial_{\overline{l}}\Psi|\partial_{k}\Psi}\braket{\Psi|\partial_i\Psi}\braket{\opartial_{\overline{\jmath}}\Psi|\Psi}}{\braket{\Psi|\Psi}^3}+2\frac{\braket{\opartial_{\overline{\jmath}}\Psi|\partial_{k}\Psi}\braket{\Psi|\partial_i\Psi}\braket{\opartial_{\overline{l}}\Psi|\Psi}}{\braket{\Psi|\Psi}^3}\\
&+2\frac{\braket{\opartial_{\overline{l}}\Psi|\partial_{i}\Psi}\braket{\Psi|\partial_k\Psi}\braket{\opartial_{\overline{\jmath}}\Psi|\Psi}}{\braket{\Psi|\Psi}^3}+2\frac{\braket{\opartial_{\overline{\jmath}}\Psi|\partial_{i}\Psi}\braket{\Psi|\partial_k\Psi}\braket{\opartial_{\overline{l}}\Psi|\Psi}}{\braket{\Psi|\Psi}^3}\\
&-6\frac{\braket{\opartial_{\overline{\jmath}}\Psi|\Psi}\braket{\opartial_{\overline{l}}\Psi|\Psi}\braket{\Psi|\partial_i\Psi}\braket{\Psi|\partial_k\Psi}}{\braket{\Psi|\Psi}^4}.
\end{split}
\end{equation}
Note that for any tensor $T$, the contraction of a covariant index $i$ with $v^i$ or of any contravariant index $j$ with $\braket{\Psi|\partial_j\Psi}$ should be zero, if we want to be able to consistently lower and raise indices with the metric and its pseudo-inverse. This can explicitly be checked for the Riemann curvature tensor. However, from the identities above we obtain
\begin{equation}
\tilde{\Gamma}_{ij}^{\ \;k} v^i=-\tilde{g}^{k\overline{m}}\tilde{g}_{\overline{m}i}(\partial_j v^i)=-\partial_j v^k+v^k\frac{\braket{\Psi|\partial_i \Psi}}{\braket{\Psi|\Psi}}(\partial_j v^i)\neq 0
\end{equation}
which is compatible with the Levi-Civita connection not being a proper tensor.

\section{Geometry of generic matrix product states}
\label{s:gmps}
In this section we study the geometry of generic matrix product states (MPS) for finite one-dimensional lattices, which were identified \cite{1995PhRvL..75.3537O,1997PhRvB..55.2164R} as the variational states implicitly created by White's density-matrix renormalization group \cite{1992PhRvL..69.2863W,1993PhRvB..4810345W}. Subsection~\ref{ss:gmps:definition} first summarizes the definition and key properties of the variational set. Secondly, in Subsection~\ref{ss:gmps:affinefiberbundle} we discuss the properties of the representation $\Psi$ itself, as studied in great detail in Ref.~\onlinecite{2006quant.ph..8197P}. We identify matrix product states as having the structure of a principal fiber bundle and introduce the necessary concepts and definitions. In particular, we identify the parameter space as the total space or bundle space and the manifold $\varM_{\mathrm{MPS}}\subset \hilbert$ as the base space. Subsection~\ref{ss:gmps:tangent} introduces the tangent space of these manifolds and discusses the relation between the tangent bundle of the parameter space and the tangent bundle $T\varM_{\mathrm{MPS}}$. In order to associate a unique parameterization to each tangent vector of the physical space, we need to introduce what is called a bundle connection. In Subsection~\ref{ss:gmps:projectivefiberbundle}, we repeat the whole construction for the manifold $\tilde{\varM}_{\mathrm{MPS}}\subset P(\hilbert)$. Finally, in Subsection~\ref{ss:gmps:metric} we define the pullback metric and introduce an efficient parameterization for tangent vectors such that the metric becomes the identity. 

\subsection{Definition and properties}
\label{ss:gmps:definition}
Consider a one-dimensional lattice $\mathcal{L}$ with $\vert\mathcal{L}\vert=N$ sites labeled by the integer $n\in\mathcal{L}=\{1,\ldots,N\}$. Every site $n$ contains a $q_{n}$-dimensional quantum variable, so that the local Hilbert space $\hilbert_{n}\cong \mathbb{C}^{q_{n}}$ is spanned by a basis $\{\ket{s_{n}}\mid s_{n}=1,\ldots,q_{n}\}$. The total Hilbert space is given by $\hilbert_{\mathcal{L}}=\bigotimes_{n=1}^{N} \hilbert_{n}$ and is spanned by the product basis
\begin{equation}
\ket{\{s_{n}\}}\defis \ket{s_{1}s_{2}\ldots s_{N}}\defis \ket{s_{1}}_{1}\otimes\ket{s_{2}}_{2}\otimes \cdots \otimes \ket{s_{N}}_{N}.
\end{equation}
The dimension of $\hilbert_{\mathcal{L}}$ is thus given by $\dim \hilbert_{\mathcal{L}}=\prod_{n=1}^{N}q_{n}$, and the specification of an arbitrary state $\ket{\Psi}\in\hilbert_{\mathcal{L}}$ requires a value for each of the coefficients $c_{s_{1},s_{2},\ldots,s_{N}}$ corresponding to the element $\ket{s_{1}s_{2}\ldots s_{N}}$ of the basis. 

In a MPS, the coefficients $c_{s_{1},s_{2},\ldots,s_{N}}$ with respect to the chosen basis $\ket{s_1s_2\ldots s_N}$ are obtained as a product of matrices, hence the name. The variational parameters correspond to a set of $q_n$ complex $D_{n-1}\times D_n$ matrices $A^{s}(n)$ ($s=1,\ldots,q_n$) for every $n\in\mathcal{L}=\{1,\ldots,N\}$, where $D_0=D_N$. The objects $A(n)$ can also be interpreted as rank 3 tensors with entries $A^{s_{n}}_{\alpha,\beta}(n)$, where there is one physical index $s_{n}=1,\ldots,q_{n}$ and two virtual indices $\alpha=1,\ldots,D_{n-1}$ and $\beta=1,\ldots,D_{n}$. The integers $D_{n}$ are called the bond dimension or virtual dimension of the MPS. For a given lattice $\mathcal{L}$ with local Hilbert spaces $\hilbert_{n}$ and fixed bond dimensions $\{D_{n}\}_{n=1,\ldots,N}$, we can thus define the MPS parameter space $\mathbb{A}_{\mathrm{MPS}}$ as the complex Euclidean space
\begin{equation}
\mathbb{A}_{\mathrm{MPS}}=\bigoplus_{n=1}^{N}\mathbb{C}^{D_{n-1}\times q_{n}\times D_{n}},
\end{equation}
with
\begin{equation}
\dim \mathbb{A}_{\mathrm{MPS}}=\sum_{n=1}^{N} D_{n-1}q_{n} D_{n}.
\end{equation}
We can now define the MPS variational class.
\begin{definition}[Matrix product state] An MPS is defined as the holomorphic map
\begin{equation}
\Psi\colon \mathbb{A}_{\mathrm{MPS}}\to \hilbert_{\mathcal{L}}\colon A \mapsto \ket{\Psi[A]}
\end{equation}
where we misuse the functional notation $[A]$ for the dependence on a discrete set of objects $A=\{A(1),A(2),\ldots,A(N)\}=\{A(n)\}_{n=1,\ldots,N}$, and $\ket{\Psi[A]}$ is given by
\begin{equation}
\ket{\Psi[A]}\defis\sum_{s_{1}=1}^{q_{1}}\cdots \sum_{s_{N}=1}^{q_{N}} \tr\left[A^{s_{1}}(1)\cdots A^{s_{N}}(N) \right] \ket{s_{1}s_{2}\ldots s_{N}}.\label{eq:gmps:defmps}
\end{equation}
Despite the slightly confusing terminology, we refer both to the set of tensors $A\in\mathbb{A}_{\mathrm{MPS}}$ and to the corresponding physical state $\ket{\Psi[A]}\in \hilbert_{\mathcal{L}}$ as a MPS.
\end{definition}
Note that it is always useless to choose $D_{n}> q_{n} D_{n-1}$ or $D_{n-1} > q_{n} D_{n}$. For example, if $D_{n}> D_{n-1}q_{n}$ then define the $(D_{n-1}q_{n}\times D_{n})$-matrix $A_{(\alpha s),\beta}(n)$ from reordering and grouping the indices of the tensor $A(n)$. The rank of the matrix $A_{(\alpha s),\beta}(n)$ is limited by $D_{n-1}q_{n}$, and there exists a $D_{n-1}q_{n}\times D_{n-1}q_{n}$ matrix $B_{(\alpha s),\gamma}(n)$ and $D_{n-1}q_{n}\times D_{n}$ matrix $Q_{\gamma,\beta}$ such that $A_{(\alpha s),\beta}(n)=\sum_{\gamma=1}^{D_{n-1}q_{n}} B_{(\alpha s),\gamma} Q_{\gamma,\beta}$. Without loss of accuracy, we can redefine $A^{s}(n)\leftarrow B^{s}$, $A^{s}(n+1) \leftarrow Q A^{s}(n+1)$ and $D_{n} \leftarrow D_{n-1}q_{n}$. A similar proof holds for the case $D_{n-1}> q_{n} D_{n}$.

We can now define a variational set
\begin{equation}
\manifold{V}_{\mathrm{MPS}}=\left\{\ket{\Psi[A]}| \forall A\in \mathbb{A}_\mathrm{MPS}\right\}.
\end{equation}
The notation $\manifold{V}_{\mathrm{MPS}}$ will not be cluttered with explicit notation of the lattice $\mathcal{L}$ or the local Hilbert dimensions $\{q_{n}\}$. The bond dimensions $\{D_{n}\}_{n=1,\ldots,N}$ can be indicated explicitly as $\manifold{V}_{\mathrm{MPS}\{D_{n}\}}$ when confusion between different choices of $\{D_{n}\}_{n=1,\ldots,N}$ is possible. The variational set $\manifold{V}_\mathrm{MPS}$ contains contains rays of states, since it corresponds to the case where $\ket{\Psi[\lambda A]}=f(\lambda)\ket{\Psi[A]}$ for all $\lambda\in\mathbb{C}$, with $f(\lambda)=\lambda^{N}$. The set $\manifold{V}_{\mathrm{MPS}\{D_{n}\}}$ is definitely not a vector space, since for two states $\ket{\Psi[A_{1}]}, \ket{\Psi[A_{2}]}\in\manifold{V}_{\mathrm{MPS}\{D_{n}\}}$, the MPS representation of $\ket{\Psi[A_{1}]}+ \ket{\Psi[A_{2}]}$ requires in the most general case a set of bond dimensions $\{D_{n}'=2D_{n}\}_{n=1,\ldots,N}$. Put differently, in the most general case we obtain $\ket{\Psi[A_{1}]}+ \ket{\Psi[A_{2}]}=\ket{\Psi'[A']}\in\manifold{V}_{\mathrm{MPS}\{D_{n}'\}}$, where $(A')^{s}(n)$ is constructed as $(A')^{s}(n)=A_{1}^{s}(n) \oplus A_{2}^{s}(n)$, $\forall s=1,\ldots,q_{n}$, $\forall n=1,\ldots,N$. Unlike in the previous section, the map $\Psi$ is not injective and we need a more detailed study to investigate whether $\manifold{V}_{\mathrm{MPS}}$ can be given the structure of a complex manifold, which is the subject of the next subsection.

If a different set of bond dimensions $\{D_{n}'\}_{n=1,\ldots,N}$ satisfies $D_{n}'\leq D_{n}$, $\forall n=1,\ldots,N$, then $\manifold{V}_{\mathrm{MPS}\{D_{n}'\}}\subset \manifold{V}_{\mathrm{MPS}\{D_{n}\}}$. A MPS $\ket{\Psi'[A']}\in\manifold{V}_{\mathrm{MPS}\{D_{n}'\}}$ can be identified with a state $\ket{\Psi[A]}\in \manifold{V}_{\mathrm{MPS}\{D_{n}\}}$ by setting, $\forall n=1,\ldots,N$, $\forall s=1,\ldots,q_{n}$, $A^{s}_{\alpha,\beta}(n)=(A')^{s}_{\alpha,\beta}(n)$ for $\alpha=1,\ldots,D_{n-1}'$ and $\beta=1,\ldots,D_{n}'$, and $A^{s}_{\alpha,\beta}(n)=0$ for all other combinations of $\alpha$ and $\beta$. It will be shown throughout the remainder of this section that the subsets $\manifold{V}_{\mathrm{MPS}\{D_{n}'\}}$ of $\manifold{V}_{\mathrm{MPS}\{D_{n}\}}$ correspond to the singular regions where \textit{e.g.}\ the pullback metric $g$ becomes (more strongly) degenerate, which is made more precise in Subsection~\ref{ss:gmps:metric}. In order to define a variational manifold, these singular regions have to be removed by restricting to an open subset $\manifold{A}_{\mathrm{MPS}}\subset \mathbb{A}_{\mathrm{MPS}}$. The corresponding image under the map $\Psi$ defines a set $\varM_{\mathrm{MPS}}$ that can be given the structure of a complex manifold, which is the main result of the next subsection.

Before getting there, we first have to define some additional quantities, which naturally occur when evaluating expectation values of physical operators operators with respect to a MPS $\ket{\Psi[A]}$. Physical operators correspond to elements of the set $\End(\hilbert_{\mathcal{L}})$ of linear endomorphisms of $\hilbert_{\mathcal{L}}$. To the virtual bond dimensions $D_n$ of the MPS $\ket{\Psi[A]}$, we can associate ancillas or virtual systems whose pure states live in $\mathbb{C}^{D_n}$. The expectation value $\braket{\Psi[\overline{A}]|\operator{O}|\ket{\Psi[A}}$ involves both the MPS and its dual, which depends on the complex conjugate variable $\overline{A}=\{\overline{A}(n)\defis\overline{A(n)}\}_{n=1,\ldots,N}$ and involves a set of dual ancillas. Any physical operator in $\End(\hilbert_{\mathcal{L}})$ can be expressed as a linear combination of elementary product operators given by
\begin{equation}
\operator{O}=\bigotimes_{n=1}^{N} \operator{O}_{n},\label{eq:gmps:productoperator}
\end{equation}
with $\operator{O}_{n}$ a local operator acting non-trivially only on $\hilbert_{n}$. For such product operators, we now introduce the concept of a \emph{superoperator}.
\begin{definition}[Superoperator]
Let $\ket{\Psi[A]}$ be a MPS with virtual bond dimensions $\{D_n\}_{n=1,\ldots,N}$ and let the physical operator $\operator{O}$ be a product operator as defined in Eq.~(\ref{eq:gmps:productoperator}). For every local operator $\operator{O}_n$ acting non-trivially on site $n$ alone, we define a \emph{superoperator} $\voperator{E}_{O}(n)$ as a linear homomorphism from the tensor product ancilla space $\mathbb{C}^{D_{n}}\otimes \overline{\mathbb{C}^{D_{n}}}$ to $\mathbb{C}^{D_{n-1}}\otimes \overline{\mathbb{C}^{D_{n-1}}}$, according to the definition
\begin{equation}
\voperator{E}_{O_n}(n)\defis\sum_{s,s'=1}^{q_{n}} \braket{s|\operator{O}_n|s'} A^{s'}(n) \otimes \overline{A^{s}(n)}.
\end{equation}
Each superoperator $\voperator{E}_{O_n}(n)$ can thus be represented as a matrix of size $D_{n-1}^{2}\times D_{n}^{2}$.
\end{definition}
Using this definition and the property $\tr(A)\tr(B)=\tr(A\otimes B)=\tr(B\otimes A)$, we now obtain for the expectation value $\braket{\Psi[\overline{A}]|\operator{O}|\Psi[A]}$ of the product operator $\operator{O}$ [see Eq.~(\ref{eq:gmps:productoperator})]
\begin{equation}
\braket{\Psi[\overline{A}]|\operator{O}|\Psi[A]}=\tr\left[\voperator{E}_{O_{1}}(1)\voperator{E}_{O_{2}}(2)\cdots \voperator{E}_{O_{N}}(N)\right].\label{eq:gmps:expectationvalue}
\end{equation}
For the evaluation of the expectation value, the different superoperators associated to the subsequent sites have to be multiplied. Starting at the site $n$ corresponding to the smallest bond dimension $D_{\text{min}}=\min\{D_n\}$, Eq.~(\ref{eq:gmps:expectationvalue}) can be evaluated with a total computational cost that scales as $\order(D_{\text{min}}^2 D_{\text{max}}^{3})$ by exploiting the tensor product structure of $\voperator{E}_{O}(n)$, where $D_{\text{max}}=\max\{D_n\}$. For a lattice with open boundary conditions, we have $D_{\text{min}}=D_0=D_N=1$ and the total computational cost for evaluating expectation values scales as $\order(D_{\text{max}}^3)$. For a translation invariant state on a lattice with periodic boundary conditions, we expect $D_{\text{min}}=D_{\text{max}}=D_n=D$, and we obtain a more unfavorable scaling $\order(D^5)$.

To every superoperator $\voperator{E}_{O}(n)$ we can associate a map $\mathscr{E}^{(n)}_{O}\colon\End(\mathbb{C}^{D_{n}})\to \End(\mathbb{C}^{D_{n-1}})$ from virtual operators $x$ acting on the ancilla space $\mathbb{C}^{D_{n}}$ to virtual operators $\mathscr{E}^{(n)}_{O}(x)$ acting on the previous ancilla space $\mathbb{C}^{D_{n-1}}$ via the prescription
\begin{equation}
\mathscr{E}^{(n)}_{O}\colon\End(\mathbb{C}^{D_{n}})\to \End(\mathbb{C}^{D_{n-1}})\colon x\mapsto \mathscr{E}^{(n)}_{O}(x)\defis \sum_{s,s'=1}^{q_{n}}\braket{s|\operator{O}|s'} A^{s'}(n) x A^{s}(n)^{\dagger}.
\end{equation}
Analogously, a second map $\widetilde{\mathscr{E}}^{(n)}_{O}$ is defined as 
\begin{equation}
\widetilde{\mathscr{E}}^{(n)}_{O}\colon\End(\mathbb{C}^{D_{n-1}})\to \End(\mathbb{C}^{D_{n}})\colon x\mapsto \widetilde{\mathscr{E}}^{(n)}_{O}(x)\defis\sum_{s,s'=1}^{q_{n}}\braket{s|\operator{O}|s'} A^{s}(n)^{\dagger} x A^{s'}(n).
\end{equation}
Via the Choi-Jamio{\l}kowski isomorphism \cite{Sudarshan:1961aa,Jamiokowski:1972aa,Choi:1975aa,Arrighi:2004aa}, virtual operators $x$ in $\mathbb{L}(\mathbb{C}^{D_{n}})$ can be associated to vectors $\rket{x}$ in the ancilla product space $\mathbb{C}^{D_{n}}\otimes \overline{\mathbb{C}^{D_{n}}}$, for which we introduce a braket-style notation with round brackets. The relation between the maps $\mathscr{E}_{O}$, $\widetilde{\mathscr{E}}_{O}$ and the superoperator $\voperator{E}_{O}$ is given by $\voperator{E}_{O}(n)\rket{x}=\rket{\mathscr{E}^{(n)}_{O}(x)}$ and $\rbra{y}\voperator{E}_{O}(n)=\rbra{\widetilde{\mathscr{E}}^{(n)}_{O}(y)}$. Note that these maps only require multiplication of matrices in the original ancilla space, and can thus be implemented as operations with computational complexity $\order(D^{3})$. A particular role is played by the map $\mathscr{E}^{(n)}\defis \mathscr{E}^{(n)}_{\one}$, which is \emph{completely positive} and for which the matrices $A^{s}(n)$ are the Kraus operators.\cite{Nielsen:2004fk} This map appears at every site $n$ where $\operator{O}$ acts trivially (\textit{i.e.}\ $\operator{O}_{n}=\operator{\one}$). Its importance follows from the observation that many relevant operators only act non-trivally on a few sites (\textit{e.g.}\ local order parameters, correlation functions, \ldots). The corresponding superoperator $\voperator{E}(n)\defis\voperator{E}_{\one}(n)$ is reminiscent of the concept of transfer operators in statistical mechanics and is henceforth referred to as such. For the remainder of this paper, we define a generalized superoperator
\begin{equation}
\voperator{E}^{A}_{B}\defis\sum_{s=1}^{q} A^{s}\otimes\overline{B^{s}},\label{eq:gmps:defgensupop}
\end{equation}
so that $\voperator{E}(n)=\voperator{E}_{A(n)}^{A(n)}$. 

\begin{subequations}
\label{eq:gmps:defvirtualdensitymatrices}
For a lattice with open boundary conditions, the expectation value of product operators $\operator{O}$ can be computed with computational complexity $\order(D^{3})$. Most interesting operators (\textit{e.g.}\ local operators, short-range interaction terms in the Hamiltonian, correlation functions) can be written as a small sum of such product operators, so that the computation of their expectation values inherit this very favorable computational complexity. Since most operators are trivial ($\operator{O}_n=\operator{\one}$) on the majority of sites, we define the sets of virtual density matrices $l=\{l(n)\}_{n=0,\ldots,N}$ and $r=\{r(n)\}_{n=0,\ldots,N}$ for the auxiliary system via the recursive definitions ($\forall n=1,\ldots,N$)
\begin{align}
l(0)=1,&\qquad l(n)=\widetilde{\mathscr{E}}^{(n)}\left(l(n-1)\right);\\
r(N)=1,&\qquad r(n-1)=\mathscr{E}^{(n)}\left(r(n)\right);
\end{align}
all of which can be computed with computational complexity $\order(D^{3})$. The expectation value of a strictly local operator is then given by $\braket{\Psi[\overline{A}]|\operator{O}_{n}|\Psi[A]}=\rbraket{l(n-1)|\voperator{E}_{O_{n}}(n)|r(n)}$ and the normalization of the state is given by $\braket{\Psi[\overline{A}]|\Psi[A]}=l(N)=r(0)=\tr[l(n) r(n)]=\rbraket{l(n)|r(n)}$, $\forall n=0,\ldots,N$. The matrices $l(n)$ and $r(n)$ are Hermitian and positive semi-definite. Further conditions on $l(n)$ or $r(n)$ can be imposed by exploiting the freedom in the MPS representation, as discussed in the next subsection. We henceforth restrict to the case of MPS with open boundary conditions. The case of translation-invariant MPS with periodic boundary conditions will be discussed in the next section.
\end{subequations}

\subsection{The principal fiber bundle of matrix product states in affine Hilbert space}
\label{ss:gmps:affinefiberbundle}
We now study the properties of the MPS representation $\ket{\Psi[A)]}$ in Eq.~(\ref{eq:gmps:defmps}) and investigate the structure that can be given to the map $\Psi\colon\mathbb{A}_{\mathrm{MPS}}\to\hilbert_{\mathcal{L}}$. The redundancy in the representation of physical states $\ket{\Psi[A]}$ by a set of matrices or tensors $A=\{A(n)\}_{n=1,\ldots,N}\in \mathbb{A}_{\mathrm{MPS}}$, \textit{i.e.}\ the non-injectivity of the map $\Psi$, has been studied in great detail in Ref.~\onlinecite{2006quant.ph..8197P} and is here reviewed, since this is essential for the remainder of the paper.

We start by introducing a group $\mathsf{G}_{\mathrm{MPS}}$ of local gauge transformations that leave the physical state $\ket{\Psi[A]}$ encoded by the set of tensors $A$ invariant.
\begin{definition}[Gauge group of MPS]
The group $\mathsf{G}_{\mathrm{MPS}}$ of local gauge transformations is defined as the direct product
\begin{equation}
\mathsf{G}_{\mathrm{MPS}}\defis\prod_{n=1}^{N} \mathsf{GL}(D_{n};\mathbb{C})=\mathsf{GL}(D_1,\mathbb{C})\times \mathsf{GL}(D_2,\mathbb{C})\times \cdots \times \mathsf{GL}(D_n,\mathbb{C}).\label{eq:gmps:mpsgaugegroup}
\end{equation}
Note that $\mathsf{G}_{\mathrm{MPS}}$ is a complex Lie group, according to Lemma~\ref{lemma:var:productgroup}, with complex dimension given by
\begin{equation}
\dim \mathsf{G}_{\mathrm{MPS}} = \sum_{n=1}^{N} D_n^2.
\end{equation}
\end{definition}
In addition, we need to specify a group action $\Gamma:\mathbb{A}_{\mathrm{MPS}}\times \mathsf{G}_{\mathrm{MPS}}\to \mathbb{A}_{\mathrm{MPS}}$ to give meaning to the invariance of MPS under the action of the gauge group. Following the standard convention in the fiber bundle literature, we choose to work with a right action.
\begin{lemma}[Group action]
The map
\begin{multline}
\Gamma\colon\mathbb{A}_{\mathrm{MPS}}\times \mathsf{G}_{\mathrm{MPS}}\to \mathbb{A}_{\mathrm{MPS}}\colon\\
 (A,G)=(\{A(n)\}_{n=1,\ldots,N},\{G(n)\}_{n=1,\ldots,N})\mapsto \Gamma[A,G]=A^{[G]}=\{A^{[G]}(n)\}_{n=1,\ldots,N}
\end{multline}
where $A^{[G]}(n)$ is defined as
\begin{equation}
\forall s=1,\ldots,q_n:\qquad A^{[G]s}(n)= G(n-1)^{-1} A^s(n) G(n)
\end{equation}
with $G(0)=G(N)$, is a right group action of $\mathsf{G}_{\mathrm{MPS}}$ on $\mathbb{A}_{\mathrm{MPS}}$. In addition, the map $\Gamma$ is holomorphic.
\end{lemma}
\begin{proof}
The following two properties are trivially fulfilled:
\begin{itemize}
\item Identity: The identity element $\one_{\mathsf{G}_{\mathrm{MPS}}}=\{\one_{D_n}\}_{n=1,\ldots,N}$ with $\one_D$ the unit matrix of size $D\times D$, acts as $\Gamma[A,\one_{\mathsf{G}_{\mathrm{MPS}}}]=A$ for any $A\in\mathbb{A}_{\mathrm{MPS}}$
\item Associativity: $\Gamma[\Gamma[A,G_1],G_2]=\Gamma(A,G_1G_2)$ for any $A\in\mathbb{A}_{\mathrm{MPS}}$ and any $G_1,G_2 \in \mathsf{G}_{\mathrm{MPS}}$, where $G_1 G_2 = \{G_1(n)G_2(n)\}_{n=1,\ldots,N}$ is the standard group product in the product group $\mathsf{G}_{\mathrm{MPS}}$.
\end{itemize}
We can thus conclude that $\Gamma$ is a right group action. 

In addition, it is easy to prove from the definition that the map $\Gamma$ from the complex manifold $\mathbb{A}_{\mathrm{MPS}}\times \mathsf{G}_{\mathrm{MPS}}$ (see Lemma~\ref{lemma:var:productmanifold}) to the complex manifold $\mathbb{A}_{\mathrm{MPS}}$ is holomorphic. For any $G\in\mathsf{G}_{\mathrm{MPS}}$, we can define an open neighborhood containing $G$ in which we can use the matrix representation of the matrices $G(n)$ in $\mathbb{C}^{D_n\times D_n}$ as coordinates. Since the matrix entries of $G^{-1}(n)$ depend on these coordinates holomorphically according to the definition of a complex Lie group, and $\mathbb{A}_{\mathrm{MPS}}$ is endowed with the standard Euclidean coordinates, it is clear that $\Gamma$ is a holomorphic map. In particular, for any $G\in\mathsf{G}_{\mathrm{MPS}}$, the map $\mathbb{A}_{\mathrm{MPS}}\to \mathbb{A}_{\mathrm{MPS}}\colon A\mapsto A^{[G]}$ is biholomorphic.
\end{proof}
\begin{corollary}[Gauge invariance of MPS]
\begin{displaymath}
\forall G\in\mathsf{G}_{\mathrm{MPS}}, \forall A\in\mathbb{A}_{\mathrm{MPS}}: \ket{\Psi[ A^{[G]}]}=\ket{\Psi[A]}.
\end{displaymath}
\end{corollary}
Put differently, the parameter space $\mathbb{A}_{\mathrm{MPS}}$ is intersected by gauge orbits corresponding to the action of $\mathsf{G}_{\mathrm{MPS}}$, where all points on a gauge orbit correspond to equivalent representations for the same physical state $\ket{\Psi[A]}\in\hilbert_{\mathcal{L}}$.

The previous results were valid for any MPS with open or periodic boundary conditions. We now restrict to the case of open boundary conditions, and define a subset $\manifold{A}_{\mathrm{MPS}}$ of $\mathbb{A}_{\mathrm{MPS}}$.
\begin{definition}[Full-rank MPS] Consider a class of MPS with open boundary conditions ($D_0=D_N=1$). The subset $\manifold{A}_{\mathrm{MPS}}$ of \emph{full-rank MPS} with open boundary conditions is defined as
\begin{equation}
\manifold{A}_{\mathrm{MPS}}=\left\{A\in\mathbb{A}_{\mathrm{MPS}}| \forall n=0,\ldots,N: l(n)>0\ \text{and}\ r(n)>0 \right\},
\end{equation}
where the virtual density matrices of a MPS with open boundary conditions were defined in Eq.~(\ref{eq:gmps:defvirtualdensitymatrices}). Hence, all virtual density matrices should be strictly positive definite, \textit{i.e.} they should have full rank on top of being positive semi-definite. If we define$D_{n-1} q_n \times D_n$ matrices $V_{(\alpha s);\beta}(n)=A^s_{\alpha,\beta}(n)$, the positive definiteness of the left density matrices $l(n)$ requires that all matrices $V(n)$ have maximal rank, namely $\rank(V(n))=D_n$. Similarly, the $D_{n-1} \times q_n D_n$ matrices $W_{\alpha;(s\beta)}(n)=A^s_{\alpha,\beta}(n)$ should have maximal rank $D_{n-1}$ if the positive definiteness of the right density matrices $r(n)$ is given.

\end{definition}
\begin{lemma}
The subset $\manifold{A}_{\mathrm{MPS}}$ of full-rank MPS is a complex manifold with $\dim \manifold{A}_{\mathrm{MPS}}=\dim \mathbb{A}_{\mathrm{MPS}}$.
\end{lemma}
\begin{proof}
The density matrices $l(n)$ and $r(n)$ depend on the tensors $A=\{A(m)\}_{m=1,\ldots,N}$ continuously. The linear independence of the columns and rows of $l(n)$ and $r(n)$ will not be affected by sufficiently small perturbations \cite{Bhatia:1996zr,Kato:1995ys}, so that there must exist an open neighborhood around every $A\in\manifold{A}_{\mathrm{MPS}}$ where all virtual density matrices remain positive definite. Hence, $\manifold{A}_{\mathrm{MPS}}$ is an open subset of the complex Euclidean space $\mathbb{A}_{\mathrm{MPS}}$ and Lemma~\ref{lemma:var:opensubsetmanifold} can be applied.
\end{proof}

\begin{subequations}
Under the group action $A\leftarrow A^{[G]}$, we obtain the following transformation behavior for the left and right density matrices
\begin{align}
l(n)&\leftarrow l^{[G]}(n)=\lvert G(0)\rvert^{-2} G(n)^\dagger l(n) G(n),\\
r(n)&\leftarrow r^{[G]}(n)=\lvert G(N)\rvert^{2} G(n)^{-1} r(n) (G(n)^{-1})^{\dagger}.
\end{align}
and it is clear that the group action preserves the decomposition of $\mathbb{A}_{\mathrm{MPS}}$ into $\manifold{A}_{\mathrm{MPS}}$ and its complement. We now study the restriction of the group action to $\manifold{A}_{\mathrm{MPS}}$, which we still denote using $\Gamma$. The defining conditions of $\manifold{A}_{\mathrm{MPS}}$ are required to prove additional properties of the group action and to gain better insight into the structure of the gauge orbits. 
\end{subequations}
\begin{lemma}
For any $A\in\manifold{A}_{\mathrm{MPS}}$, the stabilizer subgroup $\mathsf{G}^{[A]}\subset\mathsf{G}_{\mathrm{MPS}}$ of transformations $G$ that leave $A$ invariant (\textit{i.e.}\ $\Gamma[A,G]=A$) is given by $\mathsf{G}^{[A]}=\{ c\one_{\mathsf{G}_{\mathrm{MPS}}}=\{c\one_{D_n}\}_{n=1,\ldots,N}|c\in\mathbb{C}_0\}\cong \mathsf{GL}(1,\mathbb{C})$, with $\mathbb{C}_0$ the set of all non-zero complex numbers.
\end{lemma}
\begin{proof}
A necessary condition for $A^{[G]}=A$ is that, for any $n=1,\ldots,N$
\begin{displaymath}
\sum_{s=1}^{q_n} A^s(n)^\dagger l(n-1) \left(A^{[G]}\right)^s(n) = l(n). 
\end{displaymath}
Since $G(0)=G(N)\in \mathsf{GL}(1,\mathbb{C})$, we have $G(0)=G(N)=c\in \mathbb{C}_{0}$. Applying the condition above for $n=1$ results in
\begin{displaymath}
\sum_{s=1}^{q_1} A^s(1)^\dagger l(0) \left(A^{[G]}\right)^s(1) = \frac{1}{c} l(1) G(1) = l(1).
\end{displaymath}
Since, for a full-rank MPS $A\in\manifold{A}_{\mathrm{MPS}}$, all $l(n)$ are positive definite and hence invertible, the equation above can be used to solve for $G(1)$ by left multiplication with $l(1)^{-1}$, resulting in $G(1)=c\one_{D_1}$. Continuing along these lines, we obtain $G(n)=c \one_{D_n}$, for all $n=1,\ldots,N$. It can indeed be checked that this particular choice $G=c\one_{\mathsf{G}_{\mathrm{MPS}}}\defis\{c\one_{D_n}\}_{n=1,\ldots,N}$ results in $\Gamma(A,G)=A$. Hence, we obtain $\mathsf{G}^{[A]}=\{c\one|c\in \mathbb{C}_0\}\cong\mathsf{GL}(1,\mathbb{C})$.
\end{proof}

\begin{definition}[Structure group of matrix product states] The structure group $\mathsf{S}_{\mathrm{MPS}}$ of MPS is defined as $\mathsf{G}_{\mathrm{MPS}}/\mathsf{GL}(1,\mathbb{C})$. Since $\mathsf{G}_{\mathrm{MPS}}$ is a direct product group containing $\mathsf{GL}(D_N,\mathbb{C})$ with $D_N=1$, we obtain
\begin{equation}
\mathsf{S}_{\mathrm{MPS}}\defis \mathsf{G}_{\mathrm{MPS}}/\mathsf{GL}(1,\mathbb{C})\cong \prod_{n=1}^{N-1} \mathsf{GL}(D_n,\mathbb{C}).
\label{eq:gmps:mpsstructuregroup}
\end{equation}
\end{definition}
The use of the name `structure group' will become obvious further down. $\mathsf{S}_{\mathrm{MPS}}$ is also a complex Lie group with $\dim \mathsf{S}_{\mathrm{MPS}}=\dim \mathsf{G}_{\mathrm{MPS}}-1$. The group action $\Gamma$ of $\mathsf{G}_{\mathrm{MPS}}$ also defines a group action for $\mathsf{S}_{\mathrm{MPS}}$ by setting $G(0)=G(N)=1$. Obviously, this group action is also holomorphic. We use the same symbol $\Gamma$ as the difference is clear from specifying the domain. 
\begin{corollary}
The group action $\Gamma:\manifold{A}_{\mathrm{MPS}}\times \mathsf{S}_{\mathrm{MPS}}\to \manifold{A}_{\mathrm{MPS}}$ is \emph{free}, \textit{i.e.} the stabilizer subgroup $\mathsf{S}^A$ of any $A\in\manifold{A}_{\mathrm{MPS}}$ is given by the trivial group $\{\one_{\mathsf{S}_{\mathrm{MPS}}}\}$ containing only the identity $\one_{\mathsf{S}_{\mathrm{MPS}}}=\{\one_{D_n}\}_{n=1,\ldots,N-1}$.
\end{corollary}

It can also be checked that elements $G=c\one_{\mathsf{G}_{\mathrm{MPS}}}$ leave any $A\in\mathbb{A}_{\mathrm{MPS}}$ invariant. Hence, this normal subgroup also corresponds to the \emph{kernel} of the group action. For elements in the complement of $\manifold{A}_{\mathrm{MPS}}$, the stabilizer subgroup will typically be larger so that even the action of the structure group $\mathsf{S}_{\mathrm{MPS}}$ on the whole space $\mathbb{A}_{\mathrm{MPS}}$ is not free. This illustrates the importance of restricting to the subset $\manifold{A}_{\mathrm{MPS}}$ of full-rank MPS. Note, however, that the proof above only requires that all left virtual density matrices $l=\{l(n)\}_{n=0,\ldots,N}$ are positive definite. A similar proof can be constructed by only using that all right virtual density matrices are positive definite. The simultaneous positive definiteness of both sets of density matrices, as imposed in the definition of $\manifold{A}_{\mathrm{MPS}}$, is used below.

Finally, we also need to show that the group action $\Gamma:\manifold{A}_{\mathrm{MPS}}\times \mathsf{S}_{\mathrm{MPS}}\to \manifold{A}_{\mathrm{MPS}}$ is \emph{proper}. A group action is proper if the map $\manifold{A}_{\mathrm{MPS}}\times\mathsf{S}_{\mathrm{MPS}}\to \manifold{A}_{\mathrm{MPS}}\times \manifold{A}_{\mathrm{MPS}}\colon (A,G)\mapsto (A,A^{[G]}=\Gamma[A,G])$ is proper, \textit{i.e.}\ if the preimage of compact subset of $\manifold{A}_{\mathrm{MPS}}\times \manifold{A}_{\mathrm{MPS}}$ corresponds to compact subsets $\manifold{A}_{\mathrm{MPS}}\times\mathsf{S}_{\mathrm{MPS}}$. We cite the following theorem without proof\cite{Kaup:1967fk}:
\begin{theorem}
A group action $\Gamma:\manifold{M}\times \mathsf{G}\to \manifold{M}$ of a real Lie group $\mathsf{G}$ on a complex manifold $\manifold{M}$ is proper if $\mathsf{G}$ is a closed subgroup of $\mathrm{Aut}(\manifold{M})$, the group of holomorphic automorphisms on $\manifold{M}$, and preserves a continuous distance on $\manifold{M}$.
\label{th:gmps:kaup}
\end{theorem}
A few remarks are in order. Any complex Lie group is also a real Lie group. Since our group action $\Gamma$ is a holomorphic map from the complex manifold $\manifold{A}_{\mathrm{MPS}}\times \mathsf{S}_{\mathrm{MPS}}$ to the complex manifold $\manifold{A}_{\mathrm{MPS}}$, we have that in particular the map $\Gamma^{[G]}:\manifold{A}_{\mathrm{MPS}}\to \manifold{A}_{\mathrm{MPS}}$ is a holomorphic automorphism for any $G\in \mathsf{S}_{\mathrm{MPS}}$. Hence, in order to be able to apply Theorem~\ref{th:gmps:kaup}, we only need a distance function on $\manifold{A}_{\mathrm{MPS}}$. We therefore define
\begin{equation}
D_{\mathrm{MPS}}[A_0,A_1]=\varmin_{A(t)} \int_{0}^{1} \sum_{n=1}^{N} \rbraket{l(n-1;t)| \mathbb{E}^{\dot{A}(n;t)}_{\dot{A}(n;t)}|r(n;t)}\, \rmd t\label{eq:gmps:distance}
\end{equation}
for any piecewise smooth path $A:[0,1]\to \manifold{A}_{\mathrm{MPS}}:t\mapsto A(t)$ with $A(0)=A_0$ and $A(1)=A_1$, and with $l(n;t)$ and $r(n;t)$ the density matrices defined by $A(t)$. The integrandum is inspired by, but not equivalent to, the pullback metric obtained in Subsection~\ref{ss:gmps:metric}. Whereas the pullback metric would result in a zero distance between different points $A_0$ and $A_1$ which are gauge equivalent (\textit{i.e.} $A_1=\Gamma[A_0,G]$ for some $G$), the integrandum used above is strictly positive for $A(t)\in\manifold{A}_{\mathrm{MPS}}$ and any two distinct points $A_0\neq A_1$ result in $D_{\mathrm{MPS}}[A_0,A_1])>0$.  Note that the positiveness of the distance depends on both sets of density matrices ($\{l(n)\}_{n=0,\ldots,N}$ and $\{r(n)\}_{n=0,\ldots,N}$) having full rank. In addition, by substituting $A(t)\leftarrow A^{[G]}(t)$, $l(n)\leftarrow l^{[G]}(n)=G(n)^{\dagger} l(n) G(n)$ and $r^{[G]}(n)\leftarrow G(n)^{-1} r(n) (G(n)^{\dagger})^{-1}$, we obtain $D_{\mathrm{MPS}}[A_0,A_1]=D_{\mathrm{MPS}}[A_0^{[G]},A_1^{[G]}]$ for any $G\in\mathsf{S}_{\mathrm{MPS}}$, so that the distance function $D$ is invariant under the action of $\mathsf{S}_{\mathrm{MPS}}$. We can thus conclude:
\begin{corollary}
The group action $\Gamma:\manifold{A}_{\mathrm{MPS}}\times \mathsf{S}_{\mathrm{MPS}}\to \manifold{A}_{\mathrm{MPS}}$ is proper.
\end{corollary}

We can now use one of the basic theorems from the fiber bundle literature \cite{Lee:2002uq,Duistermaat:2004fk}:
\begin{theorem}[Quotient manifold theorem]
If $\manifold{A}$ is a smooth manifold, $\mathsf{S}$ a Lie group and $\Gamma\colon\manifold{A}\times \mathsf{S}\to\manifold{A}$ a smooth, free and proper group action then
\begin{itemize}
\item The orbit space $\manifold{A}/\mathsf{S}$ is a smooth manifold with $\dim \manifold{A}/\mathsf{S}=\dim \manifold{A}-\dim \mathsf{S}$.
\item The natural projection $\pi\colon\manifold{A}\to \manifold{A}/\mathsf{S}$ is a smooth submersion
\end{itemize}
\label{th:gmps:quotientmanifold}
\end{theorem}
Correspondingly, $\pi\colon\manifold{A}\to \manifold{A}/\mathsf{S}$ is a principal fiber bundle with total space $\manifold{A}$, base space $\manifold{A}/\mathsf{S}$ and structure group $\mathsf{S}$. Since the complex manifold $\manifold{A}_{\mathrm{MPS}}$ is in the first place a smooth manifold, and the holomorphic group action $\Gamma$ is a smooth map, they can be used to conclude that the MPS orbit space $\manifold{A}_{\mathrm{MPS}}/\mathsf{S}_{\mathrm{MPS}}$ is a smooth manifold.

In addition, for complex manifolds and a complex Lie group, we also state the following theorem\cite{Huybrechts:2004vn}:
\begin{theorem}
Let $\Gamma\colon\manifold{A}\times\mathsf{S}\to \manifold{A}$ be a holomorphic, free and proper group action of a complex Lie group $\mathsf{S}$ to a complex manifold $\manifold{A}$. Then the orbit space $\manifold{A}/\mathsf{S}$ is a complex manifold and the quotient map $\pi\colon\manifold{A}\to\manifold{A}/\mathsf{S}$ is holomorphic.\label{th:gmps:complexquotientmanifold}
\end{theorem}
In addition, the natural projection map or quotient map $\pi\colon\manifold{A}\to\manifold{A}/\mathsf{S}$ has a \emph{universal property}:
\begin{lemma}
Any holomorphic map $\Psi\colon\manifold{A}\to \manifold{M}$ that is invariant under the action of $\mathsf{S}$ factorizes as $\Psi=\psi \circ \pi$, where $\psi$ is a holomorphic map from the orbit space $\manifold{A}/\mathsf{S}$ to $\varM$.\label{lemma:gmps:factorization}
\end{lemma}

According to Theorem~\ref{th:gmps:complexquotientmanifold}, the MPS orbit space $\mathcal{A}_{\mathrm{MPS}}/\mathsf{S}_{\mathrm{MPS}}$ is a complex manifold. Since the MPS representation $\Psi\colon\manifold{A}_{\mathrm{MPS}}\to\hilbert_{\manifold{L}}$ is invariant under the group action, Lemma~\ref{lemma:gmps:factorization} dictates that it has a natural restriction $\psi:\manifold{A}_{\mathrm{MPS}}/\mathsf{S}_{\mathrm{MPS}}\to \hilbert_{\manifold{L}}$, which is also holomorphic. By the transition from $\mathbb{A}_{\mathrm{MPS}}$ to the subset of full-rank $\manifold{A}_{\mathrm{MPS}}$, this restricted map is made injective, as we now show in what is considered the main result of this subsection:
\begin{theorem}
The variational class of MPS $\Psi\colon\manifold{A}_{\mathrm{MPS}}\to\varM_{\mathrm{MPS}}$ is a principal fiber bundle with structure group $\mathsf{S}_{\mathrm{MPS}}$, base manifold $\varM_{\mathrm{MPS}}$, total manifold $\manifold{A}_{\mathrm{MPS}}$ and bundle projection $\Psi$. The variational manifold $\varM_{\mathrm{MPS}}$ is a complex manifold that is biholomorphic to the orbit space $\manifold{A}_{\mathrm{MPS}}/\mathsf{S}_{\mathrm{MPS}}$ and thus has dimension
\begin{equation}
\dim \varM_{\mathrm{MPS}}=\dim \manifold{A}_{\mathrm{MPS}}-\dim \mathsf{S}_{\mathrm{MPS}}=\dim \mathbb{A}_{\mathrm{MPS}}-\dim \mathsf{G}_{\mathrm{MPS}}+1.
\end{equation}
\end{theorem}
\begin{proof}
Since $\varM_{\mathrm{MPS}}$ is defined as the image of $\psi\colon\manifold{A}_{\mathrm{MPS}}/\mathsf{S}_{\mathrm{MPS}}\to\hilbert_{\mathcal{L}}$ or, equivalently, the image of $\Psi\colon\manifold{A}_{\mathrm{MPS}}\to\hilbert_{\mathcal{L}}$ and we already know that $\psi$ is holomorphic, it suffices to show that $\psi\colon\manifold{A}_{\mathrm{MPS}}/\mathsf{S}_{\mathrm{MPS}}\to\hilbert_{\mathcal{L}}$ is injective. Put differently, we need to show that the preimage of any point $\ket{\Psi[A]}$ for the map $\Psi\colon\manifold{A}_{\mathrm{MPS}}\to \hilbert_{\mathcal{L}}$ corresponds precisely to the gauge orbit $\{A^{[G]}, G\in\mathsf{S}_{\mathrm{MPS}}\}$, which is also known as the fiber corresponding to the base point $\ket{\Psi[A]}$.

A recipe for finding the MPS representation of any state $\ket{\Psi}$ in the Hilbert space $\hilbert_{\mathcal{L}}$ of a one-dimensional lattice with open boundary conditions is given in Ref.~\onlinecite{2003PhRvL..91n7902V}. It is based on a series of Schmidt decompositions (singular value decompositions) of the coefficient matrix $c_{(s_1,\ldots,s_n),(s_{n+1},\ldots,s_N)}$ corresponding to bipartite cuts between any two sites $n$ and $n+1$. For a state $\ket{\Psi(A)}$ with $A\in\mathbb{A}_{\mathrm{MPS}}$, the Schmidt rank (number of nonzero singular values) corresponding to a cut between sites $n$ and $n+1$ is at most $D_n$ and the corresponding Schmidt coefficients (singular values) are given by the eigenvalues of $\sqrt{l(n)^{1/2} r(n) l(n)^{1/2}}$ or, equivalently, the eigenvalues $\sqrt{r(n)^{1/2} l(n) r(n)^{1/2}}$. By restricting to $A\in\manifold{A}_{\mathrm{MPS}}$, all matrices $l(n)$ and $r(n)$ are strictly positive definite. Correspondingly, all Schmidt coefficients are non-zero and the Schmidt rank corresponding to the cut between sites $n$ and $n+1$ is precisely $D_n$. Under these conditions, the Schmidt decomposition becomes unique, up to permutations and degeneracies in the Schmidt coefficients. This freedom corresponds precisely to transformations $G\in\mathsf{S}_{\mathrm{MPS}}$. This construction can also be found in Ref.~\onlinecite{2006quant.ph..8197P}. 
\end{proof}

\begin{figure}
\begin{center}
\includegraphics{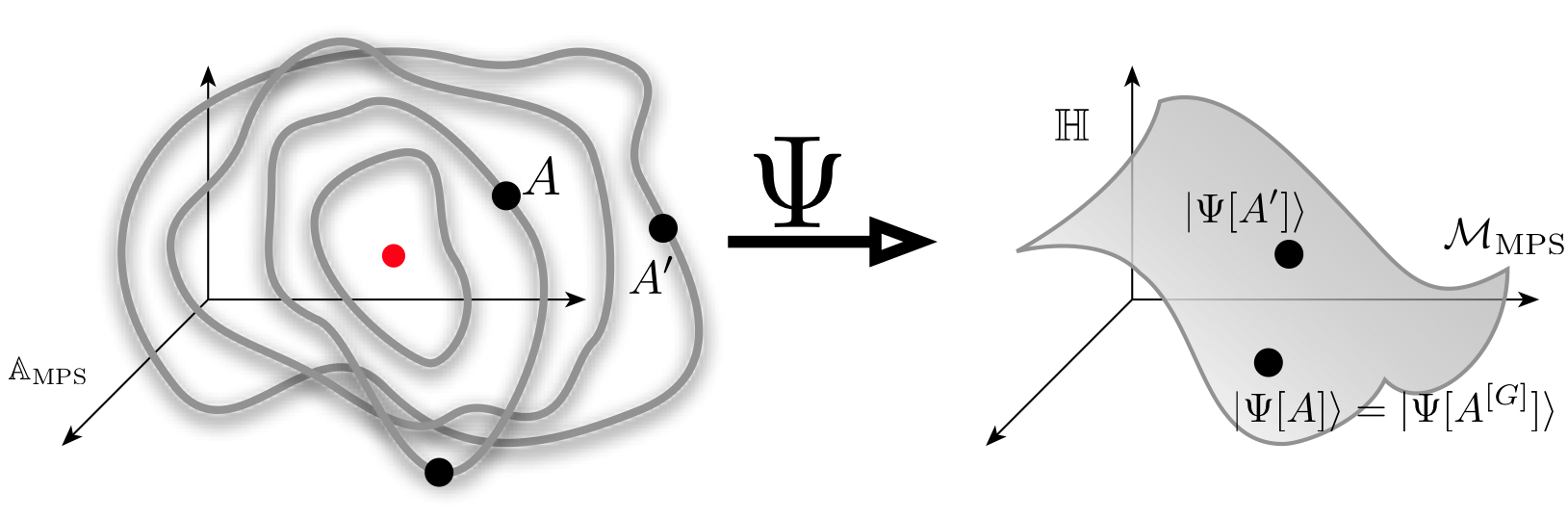}
\caption{Illustration of the principal fiber bundle interpretation of the MPS prescription. The closed lines in parameter space $\mathbb{A}_{\mathrm{MPS}}$ correspond to gauge orbits that are mapped to identical physical states in Hilbert space $\hilbert$. The dot in the middle corresponds to a MPS that does not have full rank. The gauge orbit looks fundamentally different and this point has to be excluded from the set $\manifold{A}_{\mathrm{MPS}}$ in order to define a principal fiber bundle.}
\label{fig:mps:fiberbundle}
\end{center} 
\end{figure}

A visualization of the interpretation of MPS as a principal fiber bundle is presented in FIG.~\ref{fig:mps:fiberbundle} 
For a normalized state $\braket{\Psi[\overline{A}]|\Psi[A]}=1$, the square of the Schmidt coefficients sums to one, since $\braket{\Psi[\overline{A}]|\Psi[A]}=\tr[ l(n)r(n) ]$. The \emph{entanglement spectrum} is then defined as spectrum of normalized squared Schmidt coefficients, or the negative logarithm thereof, depending on the convention. 

One can now try to associate a unique $A\in\manifold{A}_{\mathrm{MPS}}$ to every $\ket{\Psi[A]}\in\manifold{M}_{\mathrm{MPS}}$, or equivalently, to every point in $\manifold{A}_{\mathrm{MPS}}/\mathsf{S}_{\mathrm{MPS}}$. This corresponds to constructing a bundle section $\varphi\colon\varM_{\mathrm{MPS}}\to \manifold{A}_{\mathrm{MPS}}$, \textit{i.e.} a right inverse of $\Psi$: $\Psi\circ \varphi =\mathrm{Id}_{\varM_{\mathrm{MPS}}}$. While there is no easy way to explicitly specify a cross section, it is possible to try to characterize the points $A\in \varphi(\varM_{\mathrm{MPS}})$ by specifying a number of conditions that they should satisfy. Physically, these are called gauge fixing conditions and a MPS $A$ fulfilling these conditions is said to be in a \emph{canonical form}. Any MPS $A^{\prime}$ can be brought into a canonical form $A=\Gamma(A^{\prime},G)$ by acting with a transformation $G\in\mathsf{S}_{\mathrm{MPS}}$. In Ref.~\onlinecite{2006quant.ph..8197P}, a right-canonical form was constructed in two steps as:
\begin{itemize}
\item Firstly, the right orthonormalization condition is imposed (for all $n>1$)
\begin{equation}
\sum_{s=1}^{q_{n}}A^{s}(n) A^{s}(n)^{\dagger}=\mathscr{E}^{(n)}(\one_{D_{n}})=\one_{D_{n-1}}
\end{equation}
so that $r(n)=\one_{D_{n}}$ and the gauge freedom is reduced from $G(n)\in \mathsf{GL}(n;\mathbb{C})$ to $G(n)\in \mathsf{U}(n)$ ($\forall n=1,\ldots,N-1$);
\item Secondly, the left density matrices $l(n)$ are diagonalized using the remaining unitary gauge freedom. Clearly, $l(n)/\tr[l(n)]$ contains the entanglement spectrum corresponding to a bipartite cut between site $n-1$ and site $n$. In addition, we obtain $\tr[l(n)]=r(0)=\braket{\Psi[\overline{A}]|\Psi[A]}$ for any $n=1,\ldots,N$.
\end{itemize}
Alternatively, a left-canonical form can be defined. We refer to \citet{2006quant.ph..8197P,2011AnPhy.326...96S} for an efficient algorithm to obtain the canonical form starting from an arbitrary MPS. Strictly speaking, these gauge-fixing conditions do not identify a unique point within the gauge orbit even when the entanglement spectrum is non-degenerate, since there is still a freedom of choice in the phase of the eigenvectors that are used to diagonalize $l(n)$, for every $n=1,\ldots,N-1$. This residual gauge freedom corresponds to
\begin{displaymath}
\prod_{n=1}^{N-1}\underbrace{\mathsf{U}(1)\times \mathsf{U}(1)\times\cdots \times\mathsf{U}(1)}_{\text{$D_n$ times}}.
\end{displaymath}

Finally, we elaborate on the difference between the set of full-rank MPS and the set of injective MPS as defined in \citet{2006quant.ph..8197P}. For the set of full-rank MPS, the restricted map $\psi:\manifold{A}_{\mathrm{MPS}}/\mathsf{S}_{\mathrm{MPS}}\to \varM_{\mathrm{MPS}}$ becomes injective. Given the physical state $\ket{\Psi[A]}$ in any way, we can determine from it the unique fiber $\{A^{[G]},G\in\mathsf{S}_{\mathrm{MPS}}\}$ corresponding to this state. The set of injective MPS are smaller, as it corresponds to those states such that the unique fiber $\{A^{[G]},G\in\mathsf{S}_{\mathrm{MPS}}\}$ can be determined from local information about $\ket{\Psi[A]}$ only. More precisely, the fiber $\{A^{[G]},G\in\mathsf{S}_{\mathrm{MPS}}\}$ ---and thus also the state $\ket{\Psi[A]}$--- is completely determined by the set of all local density matrices of $\ell$ subsequent sites, where $\ell$ is called the \emph{injectivity length}. These states can then be obtained as unique ground states of local parent Hamiltonians, which is not necessarily the case for all full-rank MPS. In the next section on uniform MPS, we restrict to the smaller set of injective MPS, as they are the only ones that have a unambiguous thermodynamic limit. 

\subsection{Tangent bundles and the principal connection}
\label{ss:gmps:tangent}
Using the pushforward $\rmd \Psi$ of the bundle projection $\Psi$, we can define a bundle map $\rmd\Psi\colon T\manifold{A}_{\mathrm{MPS}}^+\to T\manifold{M}_{\mathrm{MPS}}^+$ between the holomorphic tangent bundle of $\manifold{A}_{\mathrm{MPS}}$ and the holomorphic tangent bundle of $\manifold{M}_{\mathrm{MPS}}$. At any point $A$ in $\manifold{A}_{\mathrm{MPS}}$, we have $T_A \manifold{A}_{\mathrm{MPS}}^+ \cong \mathbb{A}_{\mathrm{MPS}}$ by virtue of Lemma~\ref{lemma:var:opensubsetmanifold}. In addition, the holomorphic tangent space $T_{\ket{\Psi[A]}}\varM_{\mathrm{MPS}}^+\subset\hilbert_{\manifold{L}}$ is biholomorphic to a subspace of $\hilbert_{\manifold{L}}$, as denoted in Subsection~\ref{ss:var:hilbert}. It was illustrated in Ref.~\onlinecite{2012PhRvB..85c5130P,2012PhRvB..85j0408H} that this subspace defines a useful variational class to study the excited states of a Hamiltonian for which $\ket{\Psi[A]}$ is a good ground state approximation.

As in the previous section, we define a map $\ket{\Phi}\colon T\manifold{A}_{\mathrm{MPS}}^{+}\to \hilbert_{\mathcal{L}}\colon (B,A)\mapsto \ket{\Phi[B;A]}$ with
\begin{equation}
\begin{split}
\ket{\Phi[B;A]}&\defis\sum_{n=1}^{N}\sum_{i=1}^{D_{n-1} q_{n} D_{n}} B^{i}(n) \frac{\partial\ }{\partial A^{i}(n)} \ket{\Psi[A(n)]}\\
&=\sum_{n=1}^{N}\left(\sum_{s_{1}=1}^{q_{1}}\cdots\sum_{s_{n}=1}^{q_{n}}\cdots \sum_{s_{N}=1}^{q_{N}} \tr\left[A^{s_{1}}(1)\cdots B^{s_{n}}(n) \cdots A^{s_{N}}(N) \right] \ket{s_{1}s_{2}\cdots s_{N}}\right),\label{eq:gmps:defmpstangentvector}
\end{split}
\end{equation}
where $i$ is a collective index $i=(\alpha,s,\beta)$ that combines the physical index $s$ and the matrix indices $\alpha$ and $\beta$. A general tangent vector $\ket{\Phi[B;A]}$ is thus built from $N$ MPS, where one of the tensors $A(n)$ is replaced by $B(n)$. For the sake of brevity, we also introduce the notation $T_{\ket{\Psi[A]}}\varM_{\mathrm{MPS}}^{+}=\Tplane_{\mathrm{MPS}}^{[A]}$ and define the linear homomorphism $\rmd\Psi_A\defis\Phi^{[A]}$ as
\begin{equation}
\Phi^{[A]}\colon T_A\manifold{A}_{\mathrm{MPS}}^+\cong \mathbb{A}_{\mathrm{MPS}}\to \Tplane_{\mathrm{MPS}}^{[A]}\colon B\mapsto \ket{\Phi^{[A]}[B]}=\ket{\Phi[B;A]}.
\end{equation}
Often, we omit the explicit notation of the base point $A$ in the notation of the tangent space $\Tplane_{\mathrm{MPS}}$ and its vectors $\ket{\Phi[B]}$ if this is clear from the context.

Since for any $G\in\mathsf{S}_{\mathrm{MPS}}$, the map $\Gamma^{[G]}\colon\mathbb{A}_{\mathrm{MPS}}\to \mathbb{A}_{\mathrm{MPS}}\colon A\mapsto \Gamma^{[G]}(A)=\Gamma(A,G)=A^{[G]}$ describes a biholomorphism, we can define the pushforward biholomorphism of the tangent bundle to itself
\begin{equation}
\rmd\Gamma^{[G]}\colon T\mathbb{A}_{\mathrm{MPS}}\to T\mathbb{A}_{\mathrm{MPS}}\colon(A,B)\mapsto (A^{[G]},B^{[G]})
\end{equation}
where the group action on the tangent vectors $B$ is given by
\begin{equation}
\begin{split}B^{[G]}&=\rmd \Gamma^{[G]}_A(B)=\sum_{m=1}^N B^i(m) \frac{\partial \ }{\partial A^i(m)}A^{[G]}\\
&=\left\{B^{[G]}(n)=\sum_{m=1}^N B^i(m) \frac{\partial A^{[G]}(n)}{\partial A^i(m)}\right\}_{n=1,\ldots,N}
\end{split}
\end{equation}
which results in
\begin{equation}
\left(B^{[G]}\right)^s(n)=G(n-1)^{-1} B^s(n) G(n).
\end{equation}
We have thus defined a group action on the complex tangent bundle $T\manifold{A}^+_{\mathrm{MPS}}$, which is also a complex manifold in its own. 
\begin{definition}[Group action on tangent bundle]
Given the group action $\Gamma\colon\manifold{A}_{\mathrm{MPS}}\times\mathsf{S}_{\mathrm{MPS}}\to\manifold{A}_{\mathrm{MPS}}\colon (A,G)\mapsto \Gamma[A,G]=A^{[G]}$, a group action on the tangent bundle $T\manifold{A}^+_{\mathrm{MPS}}$ can be defined, which is denoted as $\rmd\Gamma$ and given by the prescription
\begin{equation}
\rmd \Gamma\colon T\manifold{A}^+_{\mathrm{MPS}}\times \mathsf{S}_{\mathrm{MPS}}\to T\manifold{A}^+_{\mathrm{MPS}}\colon ((A,B),G)\mapsto \rmd\Gamma^{[\Gamma]}[A;B]=(A^{[G]},B^{[G]}).
\end{equation}
\end{definition}
It can be proven that the group action $\rmd \Gamma$ of the complex Lie group $\mathsf{S}_{\mathrm{MPS}}$ on the complex manifold $T\manifold{A}^+_{\mathrm{MPS}}$ is also free, proper and holomorphic, and allows one to define a principal fiber bundle $T\manifold{A}^+_{\mathrm{MPS}}\to (T\manifold{A}^+_{\mathrm{MPS}})/\mathsf{S}_{\mathrm{MPS}}$, where $(T\manifold{A}^+_{\mathrm{MPS}})/\mathsf{S}_{\mathrm{MPS}}$ can also be given the structure of a complex manifold. In addition, it can easily be checked that the map $\rmd \Psi$ is invariant under the action of $\mathsf{S}_{\mathrm{MPS}}$. Indeed, explicit insertion in the definition shows $\ket{\Phi[B^{[G]};A^{[G]}]}=\ket{\Phi[B;A]}$ for any $G\in\mathsf{S}_{\mathrm{MPS}}$. According to Lemma~\ref{lemma:gmps:factorization}, the map $\rmd \Psi$ thus has a natural restriction to a map $(T\manifold{A}^+_{\mathrm{MPS}})/\mathsf{S}_{\mathrm{MPS}}\to T\varM^+_{\mathrm{MPS}}$, which is typically also denoted using the same symbol $\rmd \Psi$. 

Note, however, that $(T\manifold{A}^+_{\mathrm{MPS}})/\mathsf{S}_{\mathrm{MPS}}$ is not equal to $T(\manifold{A}_{\mathrm{MPS}}/\mathsf{S}_{\mathrm{MPS}})^+$ and that the restriction of $\rmd \Psi$ is not the same as $\rmd \psi$. This can easily be seen by counting dimensions. Whereas $\dim (T\manifold{A}^+_{\mathrm{MPS}})/\mathsf{S}_{\mathrm{MPS}}= 2 \dim \manifold{A}_{\mathrm{MPS}} - \dim \mathsf{S}_{\mathrm{MPS}}$, we have $\dim T(\manifold{A}_{\mathrm{MPS}}/\mathsf{S}_{\mathrm{MPS}})^+= 2 (\dim \manifold{A}_{\mathrm{MPS}} - \dim \mathsf{S}_{\mathrm{MPS}})$. One particular consequence is that, while $\rmd \psi$ would be injective, the restriction $\rmd \Psi\colon (T\manifold{A}^+_{\mathrm{MPS}})/\mathsf{S}_{\mathrm{MPS}}\to T\varM^+_{\mathrm{MPS}}$ is not. Indeed, after the multiplicative gauge freedom of $\mathsf{S}_{\mathrm{MPS}}$ is eliminated by \textit{e.g.}\ fixing a particular $A\in \manifold{A}_{\mathrm{MPS}}$, the linear homomorphism $\Phi^{[A]}$ has a non-trivial kernel $\mathbb{N}^{[A]}\subset \mathbb{A}_{\mathrm{MPS}}$ that contains vectors $\sum_{n=1}^N B^i \left.\frac{\partial\ }{\partial A^i(n)} \right|_A$ that are tangent to the fibers of the bundle $\Psi\colon\manifold{A}_{\mathrm{MPS}}\to \varM_{\mathrm{MPS}}$. The tangent vectors in the null space $\mathbb{N}^{[A]}\subset \mathbb{A}_{\mathrm{MPS}}$ were called zero modes in Ref.~\onlinecite{2012PhRvB..85j0408H}. In the fiber bundle literature, $\mathbb{N}^{[A]}\subset \mathbb{A}_{\mathrm{MPS}}$ is called the vertical subspace of $T_A\manifold{A}^+_{\mathrm{MPS}}$, and we can define a vertical subbundle $\mathrm{Ver}\ T\manifold{A}_{\mathrm{MPS}}^+$ with base manifold $\manifold{A}_{\mathrm{MPS}}$ and the fiber at base point $A$ given by $\mathbb{N}^{[A]}=\mathrm{Ver}\ T_A \manifold{A}^+_{\mathrm{MPS}}$. Naturally, we obtain $\ket{\Phi^{[A]}[B+B']}=\ket{\Phi^{[A]}[B]}$ for any $B'\in \mathbb{N}^{[A]}$. In order to associate a unique parameterization $B$ to every tangent vector of $\Tplane^{[A]}_{\mathrm{MPS}}$, we need to introduce an Ehresmann connection, for which there are a number of equivalent definitions. The Ehresmann connection defines at each point $A\in \manifold{A}_{\mathrm{MPS}}$ a horizontal subspace $\mathbb{B}^{[A]}\defis\mathrm{Hor}\ T_A \manifold{A}_{\mathrm{MPS}}^+$ such that
\begin{equation}
T_A \manifold{A}_{\mathrm{MPS}}^+\cong \mathbb{A}_{\mathrm{MPS}}=\mathrm{Ver}\ T_A \manifold{A}_{\mathrm{MPS}}^+\oplus \mathrm{Hor}\ T_A \manifold{A}^+_{\mathrm{MPS}}=\mathbb{N}^{[A]}\oplus \mathbb{B}^{[A]}.
\end{equation}
Every tangent vector $\ket{\Phi[B]}\in \mathbb{T}_{\mathrm{MPS}}$ then has a unique representation $B\in\mathbb{B}^{[A]}$, which is called the horizontal lift of $\ket{\Phi[B]}$. The Ehresmann connection can also be introduced as a one-form $\nu$ that takes value in the vertical subspace and acts like a projection on it, \textit{i.e.}\ $\nu_A$ is a map from $T_A\manifold{A}^+_{\mathrm{MPS}}\cong\mathbb{A}_{\mathrm{MPS}}$ to $\mathbb{N}^{[A]}$ such that $\nu_A[B]=B$ for any $B\in\mathbb{N}^{[A]}$. The horizontal subspace is then defined as
\begin{equation}
\mathbb{B}^{[A]}=\mathrm{ker}\ \nu_A=\left\{ B\in \mathbb{A}_{\mathrm{MPS}}| \nu_A[B]=0\right\}.
\end{equation}

Since the fiber bundle $\Psi\colon\manifold{A}_{\mathrm{MPS}}\to\varM_{\mathrm{MPS}}$ is a principal $\mathsf{G}$-bundle with structure group $\mathsf{S}_{\mathrm{MPS}}$, the vertical subspace $\mathbb{N}^{[A]}$ is isomorphic to the Lie algebra $\mathfrak{s}_{\mathrm{MPS}}$ of the structure group. We obtain
\begin{equation}
\mathfrak{s}_{\mathrm{MPS}}=\bigoplus_{n=1}^{N-1} \mathfrak{gl}(D_n,\mathbb{C})\cong \bigoplus_{n=1}^{N-1} \mathbb{C}^{D_n\times D_n}.
\end{equation}
We can define a map $\mathscr{N}\colon\mathfrak{s}_{\mathrm{MPS}}\to \mathrm{Ver}\ T\manifold{A}^+_{\mathrm{MPS}}\colon x=\{x(n)\}_{n=1,\ldots,N-1} \mapsto \mathscr{N}[x]$ with
\begin{equation}
\mathscr{N}_A[x]=\sum_{n=1}^{N} \mathscr{N}_A^i[x](n)\frac{\partial\ }{\partial A^i(n)}\in\mathbb{N}^{[A]}\label{eq:gmps:defzeromodes}
\end{equation}
and components given by
\begin{equation}
\mathscr{N}_A^s[x](n)=A^s(n) x(n)-x(n-1) A^s(n).
\end{equation}
The vector field $\mathscr{N}[x]$ acts as an infinitesimal generator for the group action of $G=\exp(x)$. Hence, any vertical tangent vector $B\in\mathbb{N}^{[A]}$ corresponds precisely to one element $x\in \mathfrak{s}_{\mathrm{MPS}}$ via $B=\mathscr{N}_A[x]$. We can then define the horizontal subspace as the kernel of a principal connection, which is defined as a Lie-algebra valued one-form $\omega$ such that $\omega_A\left[\mathscr{N}_A[x]\right] = x$ for any $x\in\mathfrak{s}_{\mathrm{MPS}}$ and any $A\in\manifold{A}_{\mathrm{MPS}}$. In addition, a principal connection has to transform \emph{equivariantly}, which requires that
\begin{equation}
\omega_{A^{[G]}}\left[B^{[G]}\right]= \mathrm{Ad}_{G^{-1}}\left[\omega_A(B)\right]
\end{equation}
for any $B\in T_A\manifold{A}_{\mathrm{MPS}}^+\cong\mathbb{A}_{\mathrm{MPS}}$, where for any $x\in\mathfrak{s}_{\mathrm{MPS}}$ and any $G\in\mathsf{S}_{\mathrm{MPS}}$, we have introduced the adjoint map $\mathrm{Ad}_{G}$ using the prescription
\begin{equation}
\mathrm{Ad}_{G}\left[x\right]=\{ G(n) x(n) G(n)^{-1}\}_{n=1,\ldots,N-1}.
\end{equation}
The principal connection defines an Ehresmann connection as $\nu_{A}[B]=\mathscr{N}_A[\omega_A[B]]$ which is called a principal Ehresmann connection. The equivarience property of the principal connection implies that the horizontal lift to the horizontal subspace $\mathbb{B}^{[A]}=\ker \omega_A$ can be interpreted as a \emph{section} (a continuous right inverse) of the restricted map $\rmd\Psi\colon(T\manifold{A}_{\mathrm{MPS}})/\mathsf{S}_{\mathrm{MPS}}\to T\varM_{\mathrm{MPS}}$. This restricted map can be given the structure of a fiber bundle and is called the \emph{bundle of principal connections}\cite{Kobayaschi:1957fk}, since there is a one-to-one mapping between sections of $\rmd\Psi\colon(T\manifold{A}_{\mathrm{MPS}})/\mathsf{S}_{\mathrm{MPS}}\to T\varM_{\mathrm{MPS}}$ and principal connections.

It can be checked that the following definition of $\omega^{(\mathrm{L})}$ satisfies the two conditions required in order to be a principal connection:
\begin{equation}
\begin{split}
\omega^{(\mathrm{L})}_A(B)&=\left\{l(n)^{-1} \sum_{m=1}^{n} \widetilde{\mathscr{E}}^{(n)}\left[\cdots \widetilde{\mathscr{E}}^{(m+1)}\left[\sum_s A^s(m)^\dagger l(m-1) B^s(m)\right]\cdots\right]\right\}_{n=1,\ldots,N-1}\\
&=\Bigg\{ l(n)^{-1} \sum_{m=1}^{n} \sum_{\{s_k\}} A^{s_n}(n)^{\dagger}\cdots A^{s_{m+1}}(m+1)^{\dagger}A^{s_m}(m)^{\dagger}l(m-1) \\
&\qquad\qquad\qquad\qquad\qquad\qquad\qquad \times B^{s_m}(m) A^{s_{m+1}}(m+1)\cdots A^{s}(n)\Bigg\}_{n=1,\ldots,N-1}
\end{split}
\end{equation}
To prove the equivariant transformation behavior of $\omega^{(\mathrm{L})}$, one should use that for $A\leftarrow A^{[G]}$, $l\leftarrow l^{[G]}$ with $l^{[G]}(n)= G(n)^\dagger l(n) G(n)$. The horizontal subspace is defined by these $B$ such that $\omega^{(\mathrm{L})}_A(B)=0$. Using the positivity of the virtual density matrices $l(n)$ defined in Eq.~(\ref{eq:gmps:defvirtualdensitymatrices}) and of the maps $\widetilde{\mathscr{E}}$, we obtain that tangent vectors $B\in\mathbb{B}^{[A]}$ satisfy
\begin{equation}
\forall n=1,\ldots,N-1:\quad A^s(n)^\dagger l(n-1) B^s(n)=0 \quad \Leftrightarrow \quad \rbra{l(n-1)}\voperator{E}^B_A=0.\label{eq:gmps:leftgaugeB}
\end{equation}

We can summarize this construction in a language that is more familiar to physicists, and in particular to the DMRG community. The representation $\Phi$ of MPS tangent vectors is invariant under the multiplicative group action $\ket{\Phi[B;A]}=\ket{\Phi[B^{[G]};A^{[G]}]}$. Having fixed this gauge freedom by selecting a fixed representation $A$ for the base point $\ket{\Psi[A]}$, an additional additive gauge invariance remains, since $\ket{\Phi[B;A]}=\ket{\Phi[B+\mathscr{N}_A[x];A]}$, where $\mathscr{N}_A[x]$ was defined in Eq.~(\ref{eq:gmps:defzeromodes}). Because of the linearity of $\Phi$ with respect to $B$, this boils down to the statement that $\ket{\Phi[\mathscr{N}_A[x];A]}=0$, which can easily be checked by explicitly substituting the definition of $\mathscr{N}_A[x]$ in the definition of $\Phi$ and noticing that all terms cancel. At a fixed value of $A$, we thus need a set of gauge fixing conditions to  link a physical vector $\ket{\Phi[B;A]}$ to a unique representation $B$, which can then also be called a canonical representation. In particular, if $B$ satisfies the \emph{left gauge fixing conditions} of Eq.~(\ref{eq:gmps:leftgaugeB}), it can be said to be in the left canonical form. The gauge fixing of the tangent vectors boils down to a vector space decomposition and is therefore much simpler than the gauge fixing of the original MPS. Unlike the left canonical form for the representation $A$ of the MPS $\ket{\Psi[A]}$, there is no residual gauge freedom left by the gauge fixing conditions for $B$ and the left canonical form is unique.

\begin{figure}
\begin{center}
\includegraphics[width=\textwidth]{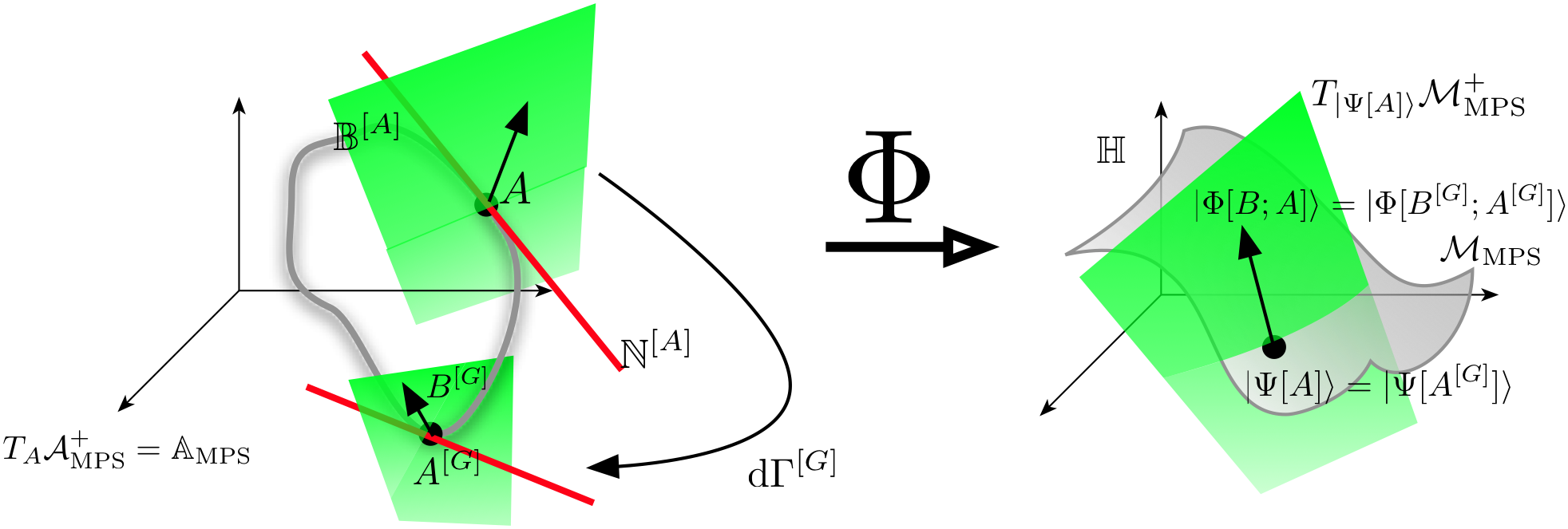}
\caption{Illustration of the tangent map $\Phi$ at base point $A$. In parameter space $\mathbb{A}_{\mathrm{MPS}}$, there is a vertical subspace $\mathbb{N}^{[A]}$ of vectors tangent to the gauge orbits (red line). A unique parameterization of vectors $\ket{\Phi[B;A]}$ in the MPS tangent space $T_{\ket{\Psi[A]}}\mathcal{M}_{\mathrm{MPS}}$ requires the definition of a complementary horizontal subspace $\mathbb{B}^{[A]}$ (green plane). If this horizontal subspace is defined as the kernel of a principal bundle connection, then it transforms equivariantly according to the adjoint representation.}
\label{fig:mps:fiberbundletangent}
\end{center} 
\end{figure}

We can also define an alternative yet equally valid principal connection $\omega^{(\mathrm{R})}$ via the prescription
\begin{equation}
\begin{split}\omega^{(\mathrm{R})}_A(B)&=\left\{-\sum_{m=n+1}^{N} \mathscr{E}^{(n)}\left[\cdots \mathscr{E}^{(m-1)}\left[\sum_s B^s(m) r(m) A^s(m)^\dagger\right]\cdots\right]r(n)^{-1}\right\}_{n=1,\ldots,N-1}\\
&=\Bigg\{- \sum_{m=n+1}^{N} \sum_{\{s_k\}} A^{s_n}(n)^{\dagger}\cdots A^{s_{m-1}}(m-1)^{\dagger}B^{s_m}(m)^{\dagger}r(m) \\
&\qquad\qquad\qquad\qquad\qquad \times A^{s_m}(m)^{\dagger} A^{s_{m-1}}(m-1)^\dagger\cdots A^{s}(n)^\dagger r(n)^{-1}\Bigg\}_{n=1,\ldots,N-1}
\end{split}
\end{equation}
The horizontal subspace $\mathbb{B}^{[A]}$ defined as $\mathrm{ker}\ \omega^{(\mathrm{R})}_A$ contains vectors $B$ that satisfy
\begin{equation}
\forall n=2,\ldots,N:\quad B^s(n) r(n) A^s(n)^\dagger=0 \quad \Leftrightarrow \quad \voperator{E}^B_A\rket{r(n)}=0,\label{eq:gmps:rightgaugeB}
\end{equation}
which are henceforth referred to as the \emph{right gauge-fixing conditions}. The vectors $B$ satisfying these conditions are said to be in the right-canonical form.

Finally, we can directly construct the transformation $\mathscr{N}_A[x]\in\mathbb{N}^{[A]}$ that transforms a general element $\tilde{B}\in\mathbb{A}_{\mathrm{MPS}}$ into a new representative $B=\tilde{B}+\mathscr{N}_A[x]\in \mathbb{B}^{[A]}$ satisfying \textit{e.g.}\ the left gauge fixing conditions in Eq.~(\ref{eq:gmps:leftgaugeB}). Starting for $n=1$, we obtain $B^s(1)=\tilde{B}^s(1)+A^s(1)x(1)$. Imposing
\begin{displaymath}
A^s(1)^\dagger l(0) B^s(1)=0
\end{displaymath}
results in
\begin{displaymath}
x(1)=-l(1)^{-1} \sum_{s=1}^{q_1} A^s(1)^\dagger l(0) \tilde{B}^s(1).
\end{displaymath}
If we have now imposed the gauge fixing conditions for all $n=1,\ldots,m-1$, which has fixed the values of $x(1)$ to $x(m-1)$ completely, we obtain $B^s(m)=\tilde{B}^s(m)-x(m-1) A^s(m) + A^s(m) x(m)$, where $x(m)$ can be used to impose the gauge fixing condition at $n=m$, resulting in ($\forall m=1,\ldots,N-1$)
\begin{displaymath}
x(m)=-l(m)^{-1} \sum_{s=1}^{q_m} A^s(m)^\dagger l(m-1) \left(\tilde{B}^s(m)-x(m-1) A^s(m)\right).
\end{displaymath}
This result is of course trivial, as it corresponds to $x=-\omega^{(\mathrm{L})}_A[\tilde{B}]$. Note that for $n=N$, we cannot impose the left gauge fixing condition. Similarly, we cannot impose the right gauge fixing condition for $n=1$. Since the overlap between the MPS $\ket{\Psi[A]}$ and one if its tangent vectors $\ket{\Phi[B]}$ is given by
\begin{equation}
\braket{\Psi[\overline{A}]|\Phi[B]}=\sum_{n=1}^{N} \rbraket{l(n-1)|\voperator{E}^{B(n)}_{A(n)}|r(n)},
\end{equation}
the right hand side is a gauge-invariant expression. By imposing the left gauge fixing conditions, it reduced to $\braket{\Psi[\overline{A}]|\Phi[B]}=\sum_{s=1}^{q_N} A^s(N)^\dagger l(N-1) B^s(N) r(N)$ with $r(N)=1$, so that imposing the left gauge fixing condition for $n=N$ would result in restricting to the subspace of tangent vectors $B$ such that $\braket{\Psi[\overline{A}]|\Phi[B]}=0$. This restriction becomes valid in the next section, where we discuss the manifold of MPS in the projective Hilbert space $P(\hilbert_{\mathcal{L}})$.

\subsection{Principal fiber bundle of matrix product states in projective Hilbert space}
\label{ss:gmps:projectivefiberbundle}
In projective Hilbert space $P(\hilbert)$, we can define a class of projective MPS.
\begin{definition}[Projective matrix product states]
The variational class of projective MPS corresponds to the map
\begin{equation}
\tilde{\Psi}\colon\manifold{A}\to P(\hilbert)\colon A\mapsto \tilde{\Psi}(A)\defis [\ket{\Psi[A]}]
\end{equation}
with $[\ket{\Psi}]$ the ray of states containing $\ket{\Psi}$. The corresponding variational set is given by
\begin{equation}
\tilde{\varM}_{\mathrm{MPS}}\defis \{\tilde{\Psi}(A), A\in\manifold{A}_{\mathrm{MPS}}\}.
\end{equation}
\end{definition}
For obvious reasons, we have again restricted to the submanifold of full-rank MPS $\manifold{A}_{\mathrm{MPS}}$. For any $A\in \manifold{A}_{\mathrm{MPS}}$, we also have $\lambda A\defis \{\lambda A(n)\}_{n=1,\ldots,N}\in \manifold{A}_{\mathrm{MPS}}$ for any $\lambda \in \mathbb{C}$. Since $\ket{\Psi[\lambda A]}=\lambda^N \ket{\Psi[A]}$, we have $\tilde{\Psi}(\lambda A)=\tilde{\Psi}(A)$ and we can define a larger structure group for which $\tilde{\Psi}$ is invariant.
\begin{definition}[Structure group of projective matrix product states] The structure group $\tilde{\mathsf{S}}_{\mathrm{MPS}}$ of projective MPS is defined as the product group
\begin{equation}
\tilde{\mathsf{S}}_{\mathrm{MPS}}\defis \mathsf{GL}(1,\mathbb{C})\times \mathsf{S}_{\mathrm{MPS}},
\end{equation}
where $\mathsf{GL}(1,\mathbb{C})$ corresponds to the group of normalization and phase changes. According to Lemma~\ref{lemma:var:productgroup}, this is a complex Lie group with $\dim \tilde{\mathsf{S}}_{\mathrm{MPS}}=\dim \mathsf{S}_{\mathrm{MPS}}+1$. The right action of $\tilde{\mathsf{S}}_{\mathrm{MPS}}$ on $\manifold{A}_{\mathrm{MPS}}$ is given by the map
\begin{multline}
\tilde{\Gamma}\colon\manifold{A}_{\mathrm{MPS}}\times \tilde{\mathsf{S}}_{\mathrm{MPS}}\to \manifold{A}_{\mathrm{MPS}}\colon\\
\qquad (A;(\lambda,G))\mapsto \tilde{\Gamma}(A,(\lambda,G))=\lambda \Gamma(A,G)=\lambda A^{[G]}=A^{[(\lambda,G)]}.
\end{multline}
\end{definition}
It can easily be checked that the additional $\mathsf{GL}(1,\mathbb{C})$ group does not change the properties of the group action, \textit{i.e.}\ just as $\Gamma$ is the new action $\tilde{\Gamma}$ a holomorphic, free and proper group action. Hence, we can reiterate all the results from Subsection~\ref{ss:gmps:affinefiberbundle} in order to obtain
\begin{theorem}
The variational class of projective MPS $\tilde{\Psi}\colon\manifold{A}_{\mathrm{MPS}}\to\tilde{\varM}_{\mathrm{MPS}}$ is a principal fiber bundle with structure group $\tilde{\mathsf{S}}_{\mathrm{MPS}}$, base manifold $\tilde{\varM}_{\mathrm{MPS}}$, total manifold $\manifold{A}_{\mathrm{MPS}}$ and bundle projection $\tilde{\Psi}$. The variational manifold $\tilde{\varM}_{\mathrm{MPS}}$ is a complex manifold embedded in $P(\hilbert)$ that is biholomorphic to the orbit space $\manifold{A}_{\mathrm{MPS}}/\tilde{\mathsf{S}}_{\mathrm{MPS}}$ and thus has dimension
\begin{equation}
\dim \tilde{\varM}_{\mathrm{MPS}}=\dim \manifold{A}_{\mathrm{MPS}}-\dim \tilde{\mathsf{S}}_{\mathrm{MPS}}=\dim \varM_{\mathrm{MPS}}-1.
\end{equation}
\end{theorem}
The proof is obtained as a straightforward generalization of the methods and results used in Subsection~\ref{ss:gmps:affinefiberbundle}. 

Correspondingly, we now study the tangent map $\rmd \tilde{\Psi}\colon T\manifold{A}^+_{\mathrm{MPS}}\to T\tilde{\varM}^+_{\mathrm{MPS}}$.  Since the new structure group $\tilde{\mathsf{S}}_{\mathrm{MPS}}$ is one dimension larger then $\mathsf{S}_{\mathrm{MPS}}$, the corresponding Lie-algebra $\tilde{\mathfrak{s}}_{\mathrm{MPS}}$ has also gained an additional dimension corresponding to
\begin{equation}
\tilde{\mathfrak{s}}_{\mathrm{MPS}}=\mathfrak{gl}(1,\mathbb{C})\oplus \mathfrak{s}_{\mathrm{MPS}}
\end{equation}
with $\mathfrak{gl}(1,\mathbb{C})\cong\mathbb{C}$. To each element $(\alpha,x)\in\tilde{\mathfrak{s}}_{\mathrm{MPS}}$ corresponds a vertical vector field $\tilde{\mathscr{N}}[(\alpha,x)]$, for which the components of $\tilde{\mathscr{N}}_A[(\alpha,x)]\in\tilde{\mathbb{N}}^{[A]}$ at base point $A$ are given by
\begin{equation}
\tilde{\mathscr{N}}_A^s[(\alpha,x)](n)=A^s(n) x(n)-x(n-1)A^s(n) + \alpha A^s(n).
\end{equation}
Note that $\ket{\Phi[B+\tilde{\mathscr{N}}_A[(\alpha,x)]]}=\ket{\Phi[B]}+\alpha N \ket{\Psi[A]}$, which is not a contradiction since $T_{[\ket{\Psi[A]}]}\tilde{\varM}^+_{\mathrm{MPS}}\cong T_{\ket{\Psi}} \varM^+_{\mathrm{MPS}}/\sim$, where for any two vectors $\ket{\Phi[B_1]},\ket{\Phi[B_2]}\in T_{\ket{\Psi[A]}} \varM^+_{\mathrm{MPS}}$, $\ket{\Phi[B_1]}\sim\ket{\Phi[B_2]}$ if $\ket{\Phi[B_1]}-\ket{\Phi[B_2]}=\beta \ket{\Psi[A]}$ for some $\beta\in\mathbb{C}$. We can easily construct two principal connections $\tilde{\omega}^{(\mathrm{L})}$ and $\tilde{\omega}^{(\mathrm{R})}$ using the prescriptions
\begin{equation}
\tilde{\omega}^{(\mathrm{L})}_A(B)=\left(\frac{1}{N}\frac{\braket{\Psi[\overline{A}]|\Phi[B]}}{\braket{\Psi[\overline{A}]|\Psi[A]}},\frac{\omega^{(\mathrm{L})}_A[B]}{\braket{\Psi[\overline{A}]|\Psi[A]}}\right)
\end{equation}
and
\begin{equation}
\tilde{\omega}^{(\mathrm{R})}_A(B)=\left(\frac{1}{N}\frac{\braket{\Psi[\overline{A}]|\Phi[B]}}{\braket{\Psi[\overline{A}]|\Psi[A]}},\frac{\omega^{(\mathrm{R})}_A[B]}{\braket{\Psi[\overline{A}]|\Psi[A]}}\right)
\end{equation}
It can easily be checked that $\tilde{\omega}^{(\mathrm{L,R})}_A(\tilde{\mathscr{N}}_A[(\alpha,x)])=(\alpha,x)$ and 
\begin{equation}
\tilde{\omega}^{(\mathrm{L,R})}_{A^{[(\lambda,G)]}}(B^{[(\lambda,G)]})=\widetilde{\mathrm{Ad}}_{(\lambda,G)^{-1}}\left[\tilde{\omega}^{(\mathrm{L,R})}_A(B)\right]
\end{equation}
where the components of $B^{[(\lambda,G)]}$ are given by $(B^{[(\lambda,G)]})^s(n)=\lambda G(n-1)^{-1} B^s G(n)$, where $(\lambda,G)^{-1}=(\lambda^{-1},G^{-1})$ and for any $(\alpha,x)\in\tilde{\mathfrak{s}}_{\mathrm{MPS}}$ we have
\begin{equation}
\widetilde{\mathrm{Ad}}_{(\lambda,G)}\left[(\alpha,x)\right]=\left(\alpha,\mathrm{Ad}_{G}\left[x\right]\right).
\end{equation}
The horizontal subspace $\tilde{\mathbb{B}}^{[A]}$ is now defined by the vectors $B$ satisfying $\tilde{\omega}^{(\mathrm{L})}_A(B)=0$ or $\tilde{\omega}^{(\mathrm{R})}_A(B)=0$. Hence, for both choices of the connection, the horizontal bundle $\tilde{\mathbb{B}}^{[A]}$ contains these vectors $B$ for which $\braket{\Psi[\overline{A}]|\Phi[B]}=0$. We thus obtain 
\begin{equation}
\tilde{\mathbb{T}}^{[A]}_{\mathrm{MPS}}\defis T_{[\ket{\Psi[A]}]}\tilde{\varM}^+_{\mathrm{MPS}}\cong \Tplane_{\mathrm{MPS}}^{[A]\perp},
\end{equation}
with $\Tplane_{\mathrm{MPS}}^{[A]\perp}$ the orthogonal complement of $\ket{\Psi[A]}$ in $\Tplane_{\mathrm{MPS}}^{[A]}$. A full characterization of the vectors $B\in\tilde{\mathbb{B}}^{[A]}$ is given by
\begin{equation}
\tilde{\omega}^{(\mathrm{L})}_A(B)=0\quad \Leftrightarrow\quad\forall n=1,\ldots,N: A^s(n)^\dagger l(n-1) B^s(n)=0\label{eq:gmps:leftgaugeBp}
\end{equation}
or
\begin{equation}
\tilde{\omega}^{(\mathrm{R})}_A(B)=0\quad \Leftrightarrow\quad\forall n=1,\ldots,N: B^s(n)r(n) A^s(n)^\dagger=0\label{eq:gmps:rightgaugeBp}.
\end{equation}

We can now also reinterpret the construction from Subsection~\ref{ss:var:projective} of the previous section in terms of a fiber bundle structure. The vector field $v^i(\vz)\partial_i$ corresponds to the vertical subbundle and generates the group action of $\mathsf{GL}(1,\mathbb{C})$, the group of normalization and phase changes. The condition in Eq.~(\ref{eq:var:condtangentvectorsparameterspace}) follows from defining a principal Ehresmann connection in order to define a unique horizontal lift of tangent vectors in $T_{[\ket{\Psi[\vz]}}\tilde{\varM}^+$ to the horizontal subspace of $(T_{\vz} \mathbb{C}^m)^+\cong\mathbb{C}^m$. 

\subsection{Pullback metric and efficient parametrization}
\label{ss:gmps:metric}
We have finally arrived at the point where we can induce the natural metric of $\hilbert$ or $P(\hilbert)$ onto $\varM_{\mathrm{MPS}}$ or $\tilde{\varM}_{\mathrm{MPS}}$ respectively, transforming these manifolds into K\"{a}hler manifolds, followed by a pullback of this metric to $\manifold{A}_{\mathrm{MPS}}$. Throughout this section, we discard the notation of $[A]$ in all quantities depending on it. 

The entries of the metric $g$ are implicitly defined by the following $N^2$ terms:
\begin{equation}
\begin{split}
\braket{\Phi[\overline{B}]|\Phi[B']}&=\sum_{n,n'=1}^{N} \overline{B}(n)^{\overline{\imath}} g_{(n,\overline{\imath});(n',j)} {B'}(n')^j\\
&=\sum_{n<n'=1}^{N}\rbraket{l(n-1)|\voperator{E}^{A(n)}_{B(n)} \left(\prod_{m=n+1}^{n'-1} \voperator{E}^{A(m)}_{A(m)}\right) \voperator{E}^{B'(n')}_{A(n')}|r(n')}\\
&\qquad+\sum_{n'<n=1}^{N}\rbraket{l(n'-1)|\voperator{E}^{B'(n')}_{A(n')} \left(\prod_{m=n'+1}^{n-1} \voperator{E}^{A(m)}_{A(m)}\right) \voperator{E}^{A(n)}_{B(n)}|r(n)}\\
&\qquad\qquad+\sum_{n=1}^{N}\rbraket{l(n-1)|\voperator{E}^{B'(n)}_{B(n)}|r(n)}.
\label{eq:gmps:overlap}
\end{split}
\end{equation}
where the definitions in Eq.~(\ref{eq:gmps:defgensupop}) and Eq.~(\ref{eq:gmps:defvirtualdensitymatrices}) were used and we use a summation convention with respect to the collected indices $i=(\alpha,s,\beta)$, but not with respect to the site index $n$. The metric is thus a complicated matrix of size $\dim \mathbb{A}_{\mathrm{MPS}}\times\dim \mathbb{A}_{\mathrm{MPS}}$, that couples all variations $B(n)$ and $B'(n')$ at different sites $n$ and $n'$. Straightforwardly, it is degenerate, since any $B\in\mathbb{N}$ results in
\begin{displaymath}
\sum_{n'=1}^{N} g_{(n,\overline{\imath});(n',j)} B^j(n') =0, \quad \forall n=1,\ldots,N, \forall \imath=1,\ldots D_{n-1}q_n D_n.
\end{displaymath}
It is, however, positive definite, when restricted to the horizontal subspace $\mathbb{B}$, so that we can define a pseudo-inverse metric such that
\begin{equation}
\sum_{n'=1}^{N} g^{(n,i);(n',\overline{\jmath})}g_{(n',\overline{\jmath});(n'',k)}=\left(P_{\mathbb{B}}\right)^{(n,i)}_{\quad(n'',k)}\label{eq:gmps:defpseudoinversemetric}
\end{equation}
where $P_{\mathbb{B}}$ is a projector onto the horizontal subspace $\mathbb{B}$ such that $\mathrm{ker}\  P_{\mathbb{B}}=\mathbb{N}$. Using a principal connection $\omega$ we can define a projector $P_{\mathbb{N}}$ to the vertical subspace with the horizontal subspace $\mathbb{B}$ as kernel, by defining
\begin{equation}
\sum_{n'=1}^{N}\left(P_{\mathbb{N}}\right)^{(n,i)}_{\quad(n',j)} B^j(n')=\mathscr{N}_A\left[\omega_A[B]\right]=\nu_A[B],
\end{equation}
so that $P_{\mathbb{N}}$ corresponds to the matrix representation of the principal Ehresmann connection $\nu$ associated to $\omega$. The projector $P_{\mathbb{B}}$ for the horizontal subspace can then be written as
\begin{equation}
P_{\mathbb{B}}=\one_{\mathbb{A}_{\mathrm{MPS}}}-P_{\mathbb{N}}.
\end{equation}
However, due to the complicated structure of the metric, it seems like an impossible task to explicitly compute a pseudo-inverse satisfying Eq.~(\ref{eq:gmps:defpseudoinversemetric}). This problem could be solved by using an iterative implementation, but the evaluation of $\braket{\Phi[\overline{B}]|\Phi[B']}$ would still scale as $\order(N^{2})$, which is also very unfavorable. However, for the horizontal subspace $\mathbb{B}$ defined by either $\omega^{(\mathrm{L})}$ or $\omega^{(\mathrm{R})}$, the vectors $B\in\mathbb{B}$ satisfy either the left or right gauge fixing conditions Eq.~(\ref{eq:gmps:leftgaugeB}) or Eq.~(\ref{eq:gmps:rightgaugeB}) respectively, and in both cases the overlap $\braket{\Phi[\overline{B}]|\Phi[B']}$ simplifies significantly to
\begin{equation}
\braket{\Phi[\overline{B}]|\Phi[B']}=\sum_{n=1}^{N}\rbraket{l(n-1)|\voperator{E}^{B'(n)}_{B(n)}|r(n)}\label{eq:gmps:overlapsimp}
\end{equation}
and all non-local terms that couple $B(n)$ with $B'(n')$ at different sites $n=n'$ are eliminated. Below, we construct a representation of tangent vectors $B\in\mathbb{B}$ in such a way that the metric is equal to the identity.

We first consider the modifications that arise when working with the projective manifold $\tilde{\varM}_{\mathrm{MPS}}\subset P(\hilbert)$ instead. We can implicitly define the pullback $\tilde{g}$ of the Fubini-Study metric to $\manifold{A}_{\mathrm{MPS}}$ as
\begin{equation}
\sum_{n,n'=1}^{N} \overline{B}(n)^{\overline{\imath}} \tilde{g}_{(n,\overline{\imath});(n',j)}{B'}(n')^j=\frac{\braket{\Phi[\overline{B}]|\Phi[B]}}{\braket{\Psi[\overline{A}]|\Psi[A]}}-\frac{\braket{\Phi[\overline{B}]|\Psi[A]}\braket{\Psi[\overline{A}]|\Phi[B]}}{\braket{\Psi[\overline{A}]|\Psi[A]}^2},\label{eq:gmps:metricP}
\end{equation}
where all quantities in the right hand side have been explicitly defined above. When working with vectors $B\in\tilde{\mathbb{B}}$, where the horizontal subspace has been defined by either $\tilde{\omega}^{(\mathrm{L})}$ or $\tilde{\omega}^{(\mathrm{R})}$, the right hand side essentially reduces to Eq.~(\ref{eq:gmps:overlapsimp}), up to the overall factor $\braket{\Psi[\overline{A}]|\Psi[A]}^{-1}$, which can be set to $1$ by properly normalizing the MPS $\ket{\Psi[A]}$. A pseudo-inverse metric is then completely characterized by
\begin{equation}
\sum_{n'=1}^{N} \tilde{g}^{(n,i);(n',\overline{\jmath})}\tilde{g}_{(n',\overline{\jmath});(n'',k)}=\left(P_{\tilde{\mathbb{B}}}\right)^{(n,i)}_{\quad(n'',k)}.\label{eq:gmps:defpseudoinversemetricP}
\end{equation}

The simplification of the metric $g$ in Eq.~(\ref{eq:gmps:overlapsimp}) is only useful if we have an efficient algorithm for imposing the left of right gauge fixing conditions in Eq.~(\ref{eq:gmps:leftgaugeB}) and Eq.~(\ref{eq:gmps:rightgaugeB}) respectively, or  Eq.~(\ref{eq:gmps:leftgaugeBp}) and Eq.~(\ref{eq:gmps:rightgaugeBp}) in case of the projective metric $\tilde{g}$. Even better is an efficient algorithm to construct gauge-fixed representations $B$  in the horizontal subspace $\mathbb{B}$ or $\tilde{\mathbb{B}}$ from the smaller number of truly independent degrees of freedom. Let us start with the projective case. It is indeed  easy to find a linear parameterization\footnote{It is important that this representation is linear in order to preserve the vector space structure of the tangent space.}  $B=\tilde{\mathscr{B}}[X]$ depending on a set $X=\{X(n)\}_{n=1,\ldots,N}$ of complex $(q_{n} D_{n-1}-D_{n})\times D_{n}$ matrices $X(n)$, where $\tilde{\mathscr{B}}[X](n)$ depends only on $X(n)$ (locality), so that $\tilde{\mathscr{B}}[X]$ satisfies the left gauge fixing conditions [Eq.~(\ref{eq:gmps:leftgaugeBp})] for all $n=1,\ldots,N$. In fact, we can even further simplify the metric and convert it into the unit matrix. We thereto define the set $L=\{L(n)\}_{n=1,\ldots,N}$ of $D_{n}\times q_{n}D_{n-1}$ matrices $L(n)$ as
\begin{equation}
[L(n)]_{\alpha;(s \beta)}= [A^{s}(n)^{\dagger}l(n-1)^{1/2}]_{\alpha,\beta}\label{eq:gmps:repdefL}
\end{equation}
and then construct a set $V_L=\{V_L(n)\}_{n=1,\ldots,N}$ of $q_{n} D_{n-1} \times (q_{n}D_{n-1}-D_n)$ matrices $V_{L}(n)$ so that $V_{L}(n)$ contains an orthonormal basis for the null space of $L(n)$, \textit{i.e.} $L(n) V_L(n)=0$ and $V_{L}(n)^\dagger V_{L}(n)=\one_{q_{n} D_{n-1}-D_{n}}$, for all $n=1,\ldots,N$. Setting $[V^{s}_{L}(n)]_{\alpha,\beta}=[V_{L}(n)]_{(s \alpha);\beta}$, we then define the representation $\tilde{\mathscr{B}}[X]$ as
\begin{equation}
\tilde{\mathscr{B}}^{s}[X](n)= l(n-1)^{-1/2} V_L^s(n) X(n)  r(n)^{-1/2}
\end{equation}
in order to obtain $\sum_{s=1}^{q_{n}} A^s(n)^\dagger l(n-1) \tilde{\mathscr{B}}[X](n)=0$, $\forall n=1,\ldots,N$, and
\begin{equation}
\braket{\Phi[\overline{\tilde{\mathscr{B}}}[\overline{X}]]|\Phi[\tilde{\mathscr{B}}[Y]]}=\sum_{n=1}^{N} \mathrm{tr}\left[ X(n)^{\dagger} Y(n)\right].
\end{equation}
Hence, when expressed in terms of the matrices $X$, the covariant and contravariant components of a vector are identical. An alternative representation $B=\tilde{\mathscr{B}}'[X']$ in terms of a set $X'=\{X'(n)\}_{n=1,\ldots,N}$ of complex $D_{n-1}\times(q_n D_n-D_{n-1})$ matrices $X'(n)$ can be constructed, so that $B=\tilde{\mathscr{B}}'[X']$ lies within the horizontal subspace defined by $\tilde{\omega}^{(\mathrm{R})}$, \textit{i.e.} the matrices $B^s(n)$ satisfy the right gauge fixing conditions [Eq.~(\ref{eq:gmps:rightgaugeBp})].
Define hereto the $q_{n}D_{n}\times D_{n-1}$ matrices $R(n)$ as
\begin{equation}
[R(n)]_{(\alpha,s);\beta}= [ r(n)^{1/2}A^{s}(n)^{\dagger}]_{\alpha,\beta}\label{eq:gmps:repdefR}
\end{equation}
and then construct a $ (q_{n} D_{n}-D_{n-1})\times q_{n} D_{n}$ matrix $V_{R}(n)$ so that $V_{R}(n)^{\dagger}$ contains an orthonormal basis for the null space of $R(n)^{\dagger}$, \textit{i.e.}\ $V_{R}(n) R(n)=0$ and $V_{R}(n) V_{R}(n)^{\dagger}=\one_{q_{n} D_{n}-D_{n-1}}$, for all $n=1,\ldots,N$. Setting $[V^{s}_{R}(n)]_{\alpha,\beta}=[V_{R}(n)]_{\alpha;(\beta,s)}$, the representation $\tilde{\mathscr{B}}'[X]$ is defined using the prescription
\begin{equation}
(\tilde{\mathscr{B}}')^{s}[X](n)= l(n-1)^{-1/2} X(n) V_{R}^{s}(n)  r(n)^{-1/2}.
\end{equation}

Finally, when using the pullback $g$ of the affine Hilbert metric, we can only impose \textit{e.g.}\ the left-gauge fixing conditions [Eq.~(\ref{eq:gmps:leftgaugeB})] for $n=1,\ldots,N-1$. The non-zero value of $\sum_{s=1}^{q_N} A^s(N)^\dagger l(N-1) B^{s}(N)$ determines the overlap $\braket{\Psi[\overline{A}]|\Phi[B]}$. We can then use a representation $B=\mathscr{B}[(\alpha,X)]$ depending on a set $X=\{X(n)\}_{n=1,\ldots,N}$ of complex $(q_{n} D_{n-1}-D_{n})\times D_{n}$ matrices $X(n)$ and a complex scalar $\alpha$, where
\begin{equation}
\mathscr{B}^s[(\alpha,X)](n)=\left\{\begin{array}{ll}\tilde{\mathscr{B}}^s[X](n),&n<N,\\
\displaystyle\tilde{\mathscr{B}}^s[X](N)+\frac{\alpha}{\braket{\Psi[\overline{A}]|\Psi[A]}^{1/2}} A^s(N),&n=N.\end{array}\right.
\end{equation}
We thus have $\braket{\Psi[\overline{A}]|\Psi[\mathscr{B}[(\alpha,X)]}=\alpha\braket{\Psi[\overline{A}]|\Psi[A]}^{-1/2}$, and
\begin{equation}
\braket{\Phi[\overline{\mathscr{B}}[(\overline{\alpha},\overline{X})]]|\Phi[\mathscr{B}[(\beta,Y)]]}=\overline{\alpha}\beta+\sum_{n=1}^{N} \mathrm{tr}\left[ X(n)^{\dagger} Y(n)\right].
\end{equation}

Note that we can now also understand why we took the effort of reducing the parameter space from the simple vector space $\mathbb{A}_{\mathrm{MPS}}$ to the more complicated manifold $\manifold{A}_{\mathrm{MPS}}$. If at a point $A\in\mathbb{A}_{\mathrm{MPS}}$ one of the virtual density matrices $l(n)$ or $r(n)$ do not have a full rank, the metric $g$ or $\tilde{g}$ has additional zero eigenvalues, \textit{i.e.} its reduction to the horizontal subspace is still degenerate. Hence, these rank-decifit points really correspond to the singular points of the set $\manifold{V}_{\mathrm{MPS}}$. The open set $\varM_{\mathrm{MPS}}$ is obtained by precisely removing these singular points from the set $\manifold{V}_{\mathrm{MPS}}$. Given the definition of a metric and its (pseudo)-inverse, we can then go on to define a Levi-Civita connection and a Riemann curvature tensor, as was sketched in the previous section. We have now reached the point where it is important to emphasize the difference between the Levi-Civita connection, which is an intrinsic property of the vector bundle $T\varM^+_{\mathrm{MPS}}$ and follows from the geometry induced by its embedding in Hilbert space, and the principal connection, which is defined for the bundle map $\rmd\Psi\colon T\manifold{A}^+_{\mathrm{MPS}}\to T\manifold{M}^+_{\mathrm{MPS}}$, and allows one to lift quantities in $T\manifold{M}^+_{\mathrm{MPS}}$ to $T\manifold{A}^+_{\mathrm{MPS}}$.  In particular, the principal connections defined above allow to lift the Levi-Civita connection from $T\manifold{M}^+_{\mathrm{MPS}}$, and thus to define it in terms of the coordinates of $T\manifold{A}^+_{\mathrm{MPS}}$.

Generic MPS have the property that the map $\Psi$ is multilinear, \textit{i.e.}\ it is linear in each parameter. Every line parallel to one of the coordinate axes in $\mathbb{A}_{\mathrm{MPS}}$ is mapped to a straight line, and thus to a geodesic, in the affine Hilbert space $\hilbert$. Hence, the variational manifold $\varM_{\mathrm{MPS}}$ is a ruled surface, where the tangent vectors along the rules constitute an (over)complete basis, since there is a rule corresponding to every coordinate axis in $\mathbb{A}_{\mathrm{MPS}}$. This fact was used in Ref.~\onlinecite{Sidles:2009fk} to study the curvature properties of variational manifolds defined by multilinear maps in a general setting, and it was shown that the sectional curvature of the resulting manifolds is always negative. We therefore restrict our discussion of the Levi-Civita connection and the Riemann curvature tensor to uniform MPS, which are defined in the next section, and for which the defining map does not have the multi-linearity property.

\section{Geometry of uniform matrix product states}
\label{s:umps}
Because many interesting physical systems are of macroscopic size, one is often interested in the bulk properties of these systems, far away from any physical boundary. In addition, the main interest is often in systems which are translation invariant. These requirements vote in favor of systems with periodic boundary conditions, where there are no boundary effects ---only finite-size effects with nice scaling behavior--- and translation invariance can easily be reproduced. On a lattice with periodic boundary conditions (where translation invariance of the models dictates a site-independent $q_{n}=q$), a translation invariant subclass of MPS can be obtained by choosing the bond dimensions $D_{n}=D$ site-independent and using a translation invariant representation, \textit{i.e.}\ by choosing the matrices to be site-independent: $A_{n}^{s}=A^{s}$, $\forall s=1,\dots,q$, $\forall n=1,\ldots,N$. The resulting variational class is called the class of \emph{uniform matrix product states} $\manifold{V}_{\mathrm{uMPS}(D)}\subset \manifold{V}_{\mathrm{MPS}\{D_n=D\}}$. Note that a general (translation non-invariant) gauge transformation will ruin the translation invariance of the representation. The only allowed gauge transformation in $\mathbb{A}_{\mathrm{uMPS}(D)}$ is a global transformation with site-indepedent matrices $G(n)=G$, $\forall n=1,\ldots,N$, where $G\in\mathsf{G}_{\mathrm{uMPS}(D)}\equiv \mathsf{GL}(D;\mathbb{C})$. Vice versa, a translation invariant MPS might only have a representation as a uMPS after a suitable site-dependent gauge transform. In addition, some translation invariant MPS do not allow for a translation invariant representation without enlarging the bond dimension\cite{2006quant.ph..8197P}. Thus, $\manifold{V}_{\mathrm{uMPS}(D)}$ does not contain all translation invariant states of $\manifold{V}_{\mathrm{MPS}\{D_{n}=D\}}$. 

\subsection{Uniform MPS as principal fiber bundles}
\label{ss:umps:fiberbundle}
Throughout this section, we are dealing with uniform MPS, for which we do not invent new symbols. However, the functional dependence on a single tensor $A$ rather than a collection of site-dependent tensors $A=\{A(n)\}$ is denoted by using round brackets $(\ )$ instead of square brackets $[\ ]$.
\begin{definition}[Uniform MPS]
Let a uniform MPS (uMPS) on a lattice $\manifold{L}=\{1,\ldots,N\}$ with a local site dimension $q$ and periodic boundary conditions be defined as a holomorphic map
\begin{equation}
\Psi\colon\mathbb{A}_{\mathrm{uMPS}(D)}\to \hilbert_{\manifold{L}}\colon A \mapsto \ket{\Psi(A)}\defis\sum_{s_{1}=1}^{q}\cdots \sum_{s_{N}=1}^{q} \tr\left[A^{s_{1}}\cdots A^{s_{N}} \right] \ket{s_{1}s_{2}\ldots s_{N}},\label{eq:umps:defumps}
\end{equation}
with $\mathbb{A}_{\mathrm{uMPS}(D)}\equiv\mathbb{C}^{D\times q \times D}$. Henceforth, we discard of the explicit notation of $(D)$ in $\mathbb{A}_{\mathrm{uMPS}}$. The image of the map $\Psi$ is defined as the variational set
\begin{equation}
\manifold{V}_{\mathrm{uMPS}}=\{\ket{\Psi(A)}, \forall A \in \mathbb{A}_{\mathrm{MPS}}\}\subset\hilbert_{\manifold{L}}.
\end{equation}
\end{definition}

In the evaluation of physical expectation values, an important role is played by the transfer matrix $\voperator{E}=\voperator{E}^A_A=\sum_{s=1}^{q} A^{s}\otimes\overline{A^{s}}$ and its associated completely positive maps $\mathscr{E}$ and $\widetilde{\mathscr{E}}$, all of which are now site-independent. The transfer matrix $\voperator{E}$ has $D^2$ eigenvalues $z^{(k)}$ ($k=1,\ldots,D^2$) and corresponding left and right eigenvectors which we denote as $\rbra{l^{(k)}}$ and $\rket{r^{(k)}}\in \mathbb{C}^{D}\otimes\overline{\mathbb{C}}^{D}$. They correspond to linear operators $l^{(k)},r^{(k)}\in\End(\mathbb{C}^{D})$, with $\mathbb{C}^{D}$ being the ancilla space, and are related to the associated maps by $\mathscr{E}(r^{(k)})=z^{(k)} r^{(k)}$ and $\widetilde{\mathscr{E}}(l^{(k)})=z^{(k)} l^{(k)}$. 

As mentioned in the introduction, the physical state $\ket{\Psi(A)}$ is unchanged under the the right group action $\Gamma\colon\mathbb{A}_{\mathrm{uMPS}}\times\group{G}_{\mathrm{uMPS}}\to \mathbb{A}_{\mathrm{uMPS}}\colon\Gamma(A,G)=\Gamma^{(G)}(A)=A^{(G)}$ defined by
\begin{equation}
\left(A^{(G)}\right)^s=G^{-1} A^s G,
\end{equation}
where $\group{G}_{\mathrm{uMPS}}\cong\mathsf{GL}(D,\mathbb{C})$. It can easily be seen that the center subgroup $\mathsf{GL}(1,\mathbb{C})=\{ c \one_D|\forall c\in\mathbb{C}_0\}$ is within the stabilizer subgroup $\mathsf{G}^{(A)}$ for any $A\in\mathbb{A}_{\mathrm{uMPS}}$. Hence, we can define a quotient group
\begin{equation}
\mathsf{S}_{\mathrm{uMPS}}=\mathsf{G}_{\mathrm{uMPS}}/\mathsf{GL}(1,\mathbb{C})\cong \mathsf{PGL}(D,\mathbb{C})
\end{equation}
with $\mathsf{PGL}(D,\mathbb{C})$ the projective linear group. We denote elements of $\mathsf{S}_{\mathrm{uMPS}}$ as $[G]$, with $[G]=G\mathsf{GL}(1,\mathbb{C})=\{c G|c\in\mathbb{C}_0\}$ the coset of the normal subgroup $\mathsf{GL}(1,\mathbb{C})\subset\mathsf{GL}(D,\mathbb{C})$. As in the previous section, we have to restrict to some open subset $\manifold{A}_{\mathrm{uMPS}}\subset \mathbb{A}_{\mathrm{MPS}}$ in order to ensure that the group action
\begin{equation}
\Gamma\colon \manifold{A}_{\mathrm{uMPS}}\times \mathsf{S}_{\mathrm{uMPS}} \to \manifold{A}_{\mathrm{uMPS}}\colon (A,[G]) \mapsto \Gamma(A,[G])=A^{(G)}\qquad (\text{for any}\ G\in[G])
\end{equation}
is free and proper. Naturally, the definition above is independent of the choice of $G$ within the coset $[G]$. In order to characterize the subset $\manifold{A}_{\mathrm{uMPS}}$, we recall some known decomposition results for uniform MPS \cite{1992CMaPh.144..443F}.

As in the previous section, we start with tensors $A\in\mathbb{A}_{\mathrm{uMPS}}$ such that the $qD\times D$ matrix $V_{(\alpha s),\beta}=A^s_{\alpha,\beta}$ and the $D\times qD$ matrix $W_{\alpha,(s\beta)}=A^s_{\alpha,\beta}$ have maximal rank, \textit{i.e.} rank $D$. If this were not the case, the uMPS $\ket{\Psi(A)}$ could be written as a uMPS $\ket{\tilde{\Psi}(\tilde{A})}$ with lower bond dimension $\tilde{D}<D$. However, this restriction is in itself insufficient for the case of periodic boundary conditions. Under the full rank condition, the $D\times D$ matrices $A^{s}$ have a \emph{block decomposition} into $J\geq 1$ blocks as
\begin{equation}
A^{s}=\begin{bmatrix} \lambda_{1} A^{s}_{1}& 0 & \cdots & 0\\ 0 & \lambda_{2} A^{s}_{2} & \cdots & 0\\
\vdots & \vdots & \ddots & \vdots\\
0 & 0 & \ldots & \lambda_{J} A^{s}_{J}
\end{bmatrix},\label{eq:umps:blockdecomposition}
\end{equation}
where $A^{s}_{j}$ are matrices of size $D_{j}\times D_{j}$, $\forall j=1,\ldots,J$, with $\sum_{j=1}^{J}D_{j}=D$. The corresponding matrices $(V_j)_{(\alpha s),\beta}=(A_j)^s_{\alpha,\beta}$ and matrix $(W_j)_{\alpha,(s\beta)}=(A_j)^s_{\alpha,\beta}$ have rank $D_j$. The coefficients $\lambda_{j}$ satisfy $0< \lambda_{j}\leq 1$ and are chosen such that the corresponding transfer operators $\voperator{E}_{j}=\sum_{s=1}^{q} A^{s}_{j}\otimes \overline{A}^{s}_{j}$ have $1$ as eigenvalue with largest absolute value. The block decomposition is constructed such that this eigenvalue is non-degenerate for each block. The remaining gauge invariance within the blocks can be used to bring the blocks $A_j$ into a specific format, such as \textit{e.g.} the right canonical form of Subsection~\ref{ss:gmps:affinefiberbundle}:
\begin{itemize}
\item $\sum_{s=1}^{q} A_{j}^{s} {A_{j}^{s}}^{\dagger}=\mathscr{E}_{j}(\one_{D_{j}})=\one_{D_{j}}$,
\item $\sum_{s=1}^{q} {A_{j}^{s}}^{\dagger} l_{j} A_{j}^{s}=\widetilde{\mathscr{E}}_{j}(l_{j})= l_{j}$ where $l_{j}$ is a diagonal matrix with strictly positive eigenvalues. 
\end{itemize}

The block decomposition states that  the uMPS $\ket{\Psi(A)}$ can be written as a superposition 
\begin{equation}
\ket{\Psi(A)}=\sum_{j=1}^{J} \lambda_{j}^{N}\ket{\Psi_{j}(A_{j})},\label{eq:umps:stateblockdecomposition}
\end{equation}
where $\ket{\Psi_{j}(A_{j})}\in \varM_{\text{uMPS}(D_{j})}$ is a uMPS with lower bond dimension $D_{j}<D$. Let us now relate this block decomposition of $A\in\mathbb{A}_{\mathrm{uMPS}}$ to the stabilizer subgroup $\mathsf{G}^{(A)}$. As mentioned before, the stabilizer subgroup certainly contains the subgroup $\mathsf{GL}(1,\mathbb{C})$. If $A$ has $J>1$ blocks in its block decomposition, it is immediately clear that there exist additional gauge transformations $G=\bigoplus_{j=1}^J c_j \one_{D_j}$ with $c_j\in\mathbb{C}_0$ for all $j=1,\ldots,J$ within the stabilizer subgroup $\mathsf{G}^{(A)}$. Let us now restrict restrict to those uMPS $A$ which have a single block within the block decomposition. Up to a normalization, the spectral radius $\rho(\voperator{E})=1$ and the transfer operator has a unique eigenvalue $z^{(1)}=1$ with corresponding left and right eigenvectors $l^{(1)}$ and $r^{(1)}$ that have full rank. It can be shown that $\voperator{E}$ then has $P$ eigenvalues $z^{(p)}$ ($p=1,\ldots,P$) that are evenly distributed on the unit circle, \textit{i.e.} $z^{(p)}=\exp(\rmi 2\pi (p-1)/P)=z_{\ast}^{p-1}$ with $z_{\ast}=\exp(\rmi 2\pi/P)$. A further decomposition, called the \emph{periodic decomposition}, is possible:
\begin{equation}
A^{s}=\begin{bmatrix} 0 & A^{s}_{1}& 0 & \cdots & 0\\ 0 & 0 & A^{s}_{2} & \cdots & 0\\
\vdots & \vdots & \vdots & \ddots & \vdots\\
 A^{s}_{P} & 0 & 0 & \cdots & 0
\end{bmatrix},\label{eq:umps:periodicdecomposition}
\end{equation}
where $A^{s}_{p}$ is a matrix of size $D_{p-1}\times D_{p}$ with $D_{0}=D_{P}$ and $\sum_{p=1}^{P} D_{p}=D$. The eigenstates $l^{(p)}$ and $r^{(p)}$ corresponding to the eigenvalues $z^{(p)}$ ($p\in\mathbb{Z}_{P}$) of unit magnitude correspondingly decompose into
\begin{subequations}
\begin{align}
l^{(p)}&=\begin{bmatrix} \big(z^{(p)}\big)^{P-1} l_{1} & 0 & 0 & \cdots & 0\\ 0 & \big(z^{(p)}\big)^{P-2}l_{2} & 0 & \cdots & 0\\
\vdots & \vdots & \vdots & \ddots & \vdots\\
0 & 0 & 0 & \cdots & l_{P}
\end{bmatrix},\\
r^{(p)}&=\begin{bmatrix} r_{P} & 0 & 0 & \cdots & 0\\ 0 & \big(z^{(p)}\big)^{P-1}r_{1} & 0 & \cdots & 0\\
\vdots & \vdots & \vdots & \ddots & \vdots\\
0 & 0 & 0 & \cdots & \big(z^{(p)}\big) r_{P-1}
\end{bmatrix},
\end{align}
\end{subequations}
where $\sum_{s=1}^{q} {A^{s}_{p}}^{\dagger} l_{p} A^{s}_{p}= l_{p+1}\in \mathbb{C}^{D_{p}\times D_{p}}$ and $\sum_{s=1}^{q} A^{s}_{p} r_{p} {A^{s}_{p}}^{\dagger}= r_{p-1}\in \mathbb{C}^{D_{p-1}\times D_{p-1}}$ ($\forall p\in\mathbb{Z}_{p}$). If $P$ is a factor of $N$, the state $\ket{\Psi(A)}$ can be written as
\begin{equation}
\ket{\Psi(A)}=\sum_{p\in\mathbb{Z}_{p}} \operator{T}^{p} \ket{\tilde{\Psi}[\tilde{A}]},\label{eq:umps:stateperiodicdecomposition}
\end{equation}
where $\ket{\tilde{\Psi}[\tilde{A}]}\in\varM_{\text{MPS}\{\tilde{D}_{n}\}}$ is a non-uniform MPS with $\tilde{A}^{s}(n)=A^{s}_{n\mod P}$ and $\tilde{D}_{n}=D_{n\mod P}$. The state $\ket{\tilde{\Psi}[\tilde{A}]}$ is thus $P$-periodic ($\operator{T}^{P}\ket{\tilde{\Psi}[\tilde{A}]}=\ket{\tilde{\Psi}[\tilde{A}]}$) and $\ket{\Psi(A)}$ is a translation invariant superposition of $\ket{\tilde{\Psi}[\tilde{A}]}$ and its shifted versions. If $P$ is not a factor of $N$, $\ket{\Psi(A)}=0$.

The set of \emph{injective uMPS} $\manifold{A}_{\mathrm{uMPS}}\subset\mathbb{A}_{\mathrm{uMPS}}$ is given as those uMPS $A$ for which $z^{(1)}=1$ is the only eigenvalue with modulus $1$, \textit{i.e.} $P=1$ in the periodic decomposition of $A$. All other eigenvalues are then situated within the unit circle. This condition is also required in order to obtain an unambiguous thermodynamic limit in Subsection~\ref{ss:umps:thermodynamiclimit}. The nomenclature (\textit{i.e.}\ injective) is clarified at the end of this subsection. A gauge-invariant definition for the subset $\manifold{A}_{\mathrm{uMPS}}\subset\mathbb{A}_{\mathrm{uMPS}}$ of injective uMPS that does depend on a particular canonical form or a particular normalization such as $\rho(\voperator{E})=1$ can also be given.
\begin{lemma}
The set of injective MPS \emph{injective uMPS} $\manifold{A}_{\mathrm{uMPS}}$ defined as
\begin{equation}
\manifold{A}_{\mathrm{uMPS}}=\left\{A\in\mathbb{A}_{\mathrm{uMPS}}\mid \rho\big(\voperator{E}-z^{(1)} \rket{r^{(1)}}\rbra{l^{(1)}}\big)<\rho(\voperator{E})\ \text{and}\ \rank(l^{(1)})=\rank(r^{(1)})=D\right\}
\end{equation}
where $z^{(1)}$ is any eigenvalue of largest magnitude of the transfer matrix $\mathbb{E}$ associated to $A$ (it is unique for $A\in\manifold{A}_{\mathrm{uMPS}}$), and $\rket{r^{(1)}}$ and $\rbra{l^{(1)}}$ are the corresponding right and left eigenvectors, 
is a complex manifold with $\dim \manifold{A}_{\mathrm{uMPS}}=\dim \mathbb{A}_{\mathrm{uMPS}}=q D^2$.
\end{lemma}
\begin{proof}
According to Lemma~\ref{lemma:var:opensubsetmanifold}, it is sufficient to prove that $\manifold{A}_{\mathrm{uMPS}}$ is an open subset of $\mathbb{A}_{\mathrm{uMPS}}$. Since $z^{(1)}$ is a non-degenerate eigenvalue, its value and associated eigenvectors change continuously under arbitrary small perturbations, so that the non-degeneracy and the full rank condition on the eigenvectors are also preserved\cite{Kato:1995ys}. Hence, for any $A\in\manifold{A}_{\mathrm{uMPS}}$, there exist an open neighborhood of $A$ contained in $\manifold{A}_{\mathrm{uMPS}}$.
Alternatively, one can argue that the complement $\manifold{A}_{\mathrm{uMPS}}$ in $\mathbb{A}_{\mathrm{uMPS}}$ satisfies conditions which make it a closed set.
\end{proof}

Having characterized the complex manifold $\manifold{A}_{\mathrm{uMPS}}$, we can now study the restriction of the group action to $\manifold{A}_{\mathrm{uMPS}}$. 
\begin{lemma}
For any $A\in\manifold{A}_{\mathrm{uMPS}}$, the stabilizer subgroup $\mathsf{G}^{(A)}$ with respect to the group action of $\mathsf{G}_{\mathrm{uMPS}}$ is exactly equal to the center group $\mathsf{GL}(1,\mathbb{C})$.
\end{lemma}
\begin{proof}
For any $A\in\manifold{A}_{\mathrm{uMPS}}$, the corresponding transfer operator $\mathbb{E}$ has a non-degenerate eigenvalue $z^{(1)}$ for which the corresponding eigenvectors $l^{(1)}=l$ and $r^{(1)}=r$ have full rank. Since $G\in\mathsf{G}^{(A)}$ implies that $A^{(G)}=A$, we obtain
\begin{displaymath}
\voperator{E}\rket{Gr}=z^{(1)}\rket{Gr},
\end{displaymath}
so that $Gr$ is proportional to the $r$. Hence, the only elements $G$ of the stability group satisfying this relation are $G=c\one_D\in\mathsf{GL}(1,\mathbb{C})$.
\end{proof}
Hence, by trading the full gauge group $\mathsf{G}_{\mathrm{uMPS}}$ for the quotient group $\mathsf{S}_{\mathrm{uMPS}}=\mathsf{PGL}(\mathbb{C},D)$, we obtain the following result:
\begin{corollary}
The group action $\Gamma:\manifold{A}_{\mathrm{uMPS}}\times \mathsf{S}_{\mathrm{uMPS}}\to \manifold{A}_{\mathrm{uMPS}}$ is free. 
\end{corollary}

Note that the proof above only requires a non-degenerate eigenvalue with corresponding full rank eigenvectors. Hence, the group action $\Gamma:\manifold{A}_{\mathrm{uMPS}}\times \mathsf{S}_{\mathrm{uMPS}}\to \manifold{A}_{\mathrm{uMPS}}$ is free for all uMPS $A$ which have a single block in the block decomposition ($J=1$), even if this block corresponds to $P>1$ in the periodic decomposition. However, for $P>1$, there do exist non-trivial translation non-invariant gauge transformations (\textit{i.e.} with site-dependent $G(n)$) that have no effect on the uMPS. In addition, the restriction to uMPS $A$ with $P=1$ is required below in order to have a natural notion of injectivity.

\begin{lemma}
The group action $\Gamma:\manifold{A}_{\mathrm{uMPS}}\times \mathsf{S}_{\mathrm{uMPS}}\to \manifold{A}_{\mathrm{uMPS}}$ is proper. 
\end{lemma}
\begin{proof}
As in the previous section, we use Theorem~\ref{th:gmps:kaup} to prove this result. The group action $\Gamma:\manifold{A}_{\mathrm{uMPS}}\times\mathsf{S}_{\mathrm{uMPS}}\to \manifold{A}_{\mathrm{uMPS}}$ naturally fulfills the conditions of Theorem~\ref{th:gmps:kaup} if we can define a distance function that is invariant under the action of $\mathsf{S}_{\mathrm{uMPS}}$. The following distance function meets this requirement:
\begin{equation}
D_{\mathrm{uMPS}}(A_0,A_1)=\varmin_{A(t)} \int_{0}^{1} \rbraket{l(t)|\mathbb{E}^{\dot{A}(t)}_{\dot{A}(t)}|r(t)}\, \rmd t
\end{equation}
where $A:[0,1]\to \manifold{A}_{\mathrm{uMPS}}:t\mapsto A(t)$ is a piecewise smooth path with $A(0)=A_0$ and $A(1)=A_1$, and $l(t)$ and $r(t)$ are a left and right eigenvector corresponding to the largest eigenvalue of $\mathbb{E}^{A(t)}_{A(t)}$ and are normalized as $\tr[ l(t) r(t) ] =1$. Given the conditions on the subset $\manifold{A}_{\mathrm{uMPS}}$ of injective uMPS, the integrand is strictly positive for nonzero $\dot{A}(t)$, so that any two distinct uMPS $A_0$ and $A_1$ result in $D_{\mathrm{uMPS}}(A_0,A_1)>0$. It is straightforwardly checked that $D_{\mathrm{uMPS}}(A_0,A_1)=D_{\mathrm{uMPS}}(A^{(G)}_0,A^{(G)}_1)$ for any $[G]\in\mathsf{S}_{\mathrm{uMPS}}$. The distance function $D_{\mathrm{uMPS}}$ can only be related to a (rescaled version of) $D_{\mathrm{MPS}}$ in the thermodynamic limit, as is explained in the next subsection.
\end{proof}

We now have all the necessary ingredients to prove the main theorem of this section:
\begin{theorem}
The variational class of injective uMPS $\Psi:\manifold{A}_{\mathrm{uMPS}}\to \varM_{\mathrm{uMPS}}$ can be given the structure of a principal fiber bundle with structure group $\mathsf{S}_{\mathrm{uMPS}}$, base manifold $\varM_{\mathrm{uMPS}}$, total manifold $\manifold{A}_{\mathrm{uMPS}}$ and bundle projection $\Psi$. The variational manifold $\varM_{\mathrm{uMPS}}$ is a complex submanifold of $\hilbert$ that is biholomorphic to the orbit space $\manifold{A}_{\mathrm{uMPS}}/\mathsf{S}_{\mathrm{uMPS}}$ and thus has dimension $\dim \varM_{\mathrm{MPS}}= q D^2 - D^2 +1$. 
\end{theorem}
\begin{proof}
As in the previous section, this theorem depends strongly on the general results stated in Theorems~\ref{th:gmps:quotientmanifold} and \ref{th:gmps:complexquotientmanifold} and Lemma~\ref{lemma:gmps:factorization}. These theorems allow to conclude that $\manifold{A}_{\mathrm{uMPS}}/\mathsf{S}_{\mathrm{uMPS}}$ is a complex manifold and that $\Psi:\manifold{A}_{\mathrm{uMPS}}\to \varM_{\mathrm{uMPS}}$ factorizes as $\psi\circ \pi$ with $\pi:\manifold{A}_{\mathrm{uMPS}}\to \manifold{A}_{\mathrm{uMPS}}/\mathsf{S}_{\mathrm{uMPS}}$ the natural holomorphic projection. In order to make any statements about the nature of $\varM_{\mathrm{uMPS}}$, it remains to be proven that $\psi: \manifold{A}_{\mathrm{uMPS}}/\mathsf{S}_{\mathrm{uMPS}}\to \varM_{\mathrm{uMPS}}$ is a biholomorphism. The holomorphic map $\psi$ is surjective by definition. Hence, it remains to be proven that it is injective, \textit{i.e.} that the uMPS representation of a state $\ket{\Psi(A)}$ is unique up to the action of the gauge group for any $A\in\manifold{A}_{\mathrm{uMPS}}$. It was proven in Ref.~\onlinecite{1992CMaPh.144..443F,2006quant.ph..8197P} that the defining properties of $\manifold{A}_{\mathrm{uMPS}}$ are sufficient to obtain a one-to-one correspondence between $A$ and $\ket{\Psi(A)}$, hence justifying the name injective MPS. In fact, the result is even stronger, since $A$ is completely defined (up to gauge transformations) by all reduced density matrices $\rho_{\ell}=\tr_{\manifold{L}\setminus \ell} \ket{\Psi(A)}\bra{\Psi(\overline{A})}$ for any contiguous block of sites $\ell$ that is longer than a certain minimal length $\ell_0$, called the injectivity length.
\end{proof}

As in the previous section on generic MPS, we can similarly define a manifold $\tilde{\varM}_{\mathrm{uMPS}}\subset P(\hilbert)$ and define it as the base space of a fiber bundle with bundle space $\manifold{A}_{\mathrm{uMPS}}$ and an enlarged structure group $\tilde{\group{S}}_{\mathrm{uMPS}}$ that also includes norm and phase changes. We will explicitly do so when discussing the thermodynamic limit in the next subsection, where this becomes the most natural structure to look at.

\subsection{Thermodynamic limit}
\label{ss:umps:thermodynamiclimit}
Despite the nice properties of systems with periodic boundary conditions, the increased computational complexity of evaluating expectation values with respect to MPS with periodic boundary conditions has hindered their applicability. This increased computational complexity is caused by the fact that correlations between two points can travel along two different ways on the circle. In contrast, systems with open boundary conditions can have strong boundary effects (Friedel oscillations) that extend deeply into the bulk, especially for (near)-critical systems. However, for very large systems ---which are finite-size restrictions of translation invariant Hamiltonians in the thermodynamic limit--- we still expect the matrices of the MPS approximation of the ground state to become site-independent when sufficiently far from the boundaries. By exploiting the translation invariance in either a MPS with periodic boundary conditions or in the bulk of a MPS with open boundary conditions, we can directly define a uniform MPS representation in the thermodynamic limit. The computational disadvantages of the MPS with periodic boundary conditions disappear, since observables with compact support cannot distinguish between open or periodic boundary conditions. On the other hand, boundary effects are also undetectable by operators that live deep in the bulk. We can therefore discard them and restrict to the translation invariant bulk of a system with open boundary conditions.

A quantitative verification of these statements requires the definition of the class of uMPS $\ket{\Psi(A)}$ in the thermodynamic limit. Starting with a system on the lattice $\mathcal{L}=\{-N,-N+1,\ldots,N-1,N\}$, this limit is formally obtained as
\begin{equation}
\ket{\Psi(A)}=\lim_{N\to \infty } \sum_{\{s_{n}\}=1}^{q} \tr\left[Q\prod_{n=-N}^{+N} A^{s_{n}}\right] \ket{\{s_{n}\}}=\sum_{\{s_{n}\}=1}^{q} \tr\left[Q\prod_{n\in\mathbb{Z}} A^{s_{n}}\right] \ket{\{s_{n}\}}.\label{eq:umps:definf}
\end{equation}
The $D\times D$ matrix $Q$ is a boundary matrix that allows one to interpolate between periodic boundary conditions ($Q=\one_D$, $\text{rank}(Q)=D$) and open boundary conditions ($\text{rank}(Q)=1$). The purpose of this section is to show in a constructive fashion that the subset of elements $A\in\mathbb{A}_{\mathrm{uMPS}}$ that have a well-defined thermodynamic limit ---in the sense that normalized expectation values of local operators do not depend on the boundary conditions encoded by $Q$--- corresponds exactly to the set $\manifold{A}_{\mathrm{uMPS}}$. Note that the limit in the definition above is only valid at a physical level, and a more rigorous treatment would inevitably require the introduction of $C^\ast$-algebras. Such a rigorous study of uMPS $\ket{\Psi(A)}$ has been done by Fannes, Nachtergaele and Werner in Ref.~\onlinecite{1992CMaPh.144..443F}, who refer to these states as \emph{finitely correlated states}. The class of finitely correlated states is even more general, and the subclass that corresponds to the uMPS are the so-called \emph{purely generated} finitely correlated states. When the transfer matrix $\voperator{E}$ has a unique eigenvalue $1$ (\textit{i.e.}\ $J=1$ in the block decomposition), the state is called \emph{ergodic}, and when this is also the only eigenvalue with modulus $1$ (\textit{i.e.}\ $P=1$ in the periodic decomposition), the state is called a \emph{pure} finitely correlated state. Unlike the ordinary MPS $\ket{\Psi[A]}$, which is linear in each of its arguments $A(n)$ separately, the uMPS $\ket{\Psi(A)}$ is highly non-linear in its argument. It took some major breakthroughs before an algorithm was constructed that allowed the variational optimization of the uMPS ansatz \cite{2007PhRvL..98g0201V}.

The norm of the state in Eq.~(\ref{eq:umps:definf}) is given by $\braket{\Psi(A)|\Psi(A)}=\lim_{N\to \infty} \tr[(Q\otimes \overline{Q})\voperator{E}^{2N+1}]$. We can always rescale $A$ such that the spectral radius $\rho(\voperator{E})=1$ and $\ket{\Psi(A)}$ becomes normalizable to some finite value. However, by only looking at normalized expectation values and defining everything in terms of the limit $N\to\infty$, this is not really necessary. Let $z^{(k)}$ for $k=1,\ldots,K$ be the eigenvalues with largest magnitude, so that $\lvert z^{(k)}\rvert=\rho(\voperator{E})$, and denote the corresponding left and right eigenvectors as $\rbra{l^{(k)}}$ and $\rket{r^{(k)}}$, which are normalized as $\rbraket{l^{(k)}|r^{(k)}}=1$ so that $\voperator{S}^{(k)}=\rket{r^{(k)}}\rbra{l^{(k)}}$ is a projector onto the corresponding eigenspace. At least one eigenvalue is positive, and we label it with $k=1$, so that $z^{(1)}=\rho(\voperator{E})$. For the normalization of the state, the only terms that survive the limit are given by
\begin{equation}
\braket{\Psi(\overline{A})|\Psi(A)}=\lim_{N\to\infty}\sum_{k=1}^{K} (z^{(k)})^{2N+1} \rbraket{l^{(k)}|Q\otimes\overline{Q}|r^{(k)}}.
\end{equation}
If a product operator $\operator{O}$ has non-trivial support only on the sites $\{-M,-M+1,\ldots,+M\}$ with $M$ some constant, then the correlations acting along the other side of the circle have to travel over an ``infinite distance'' in the limit $N\to \infty$. Here too, only the terms corresponding to the eigenvalues $z^{(k)}$ ($k=1,\ldots,K$) survive, resulting in
\begin{multline}
\braket{\Psi(\overline{A})|\operator{O}|\Psi(A)}=
\lim_{N\to\infty}\sum_{k=1}^{K} (z^{(k)})^{2N-2M} \rbraket{l^{(k)}|Q\otimes\overline{Q}|r^{(k)}}\\
\times\rbraket{l^{(k)}| \voperator{E}_{O(-M)}\voperator{E}_{O(-M+1)}\cdots \voperator{E}_{O(M)}|r^{(k)}}.
\end{multline}
We can thus define the normalized expectation value as
\begin{equation}
O(\overline{A},A)\defis\frac{\braket{\Psi(\overline{A})|\operator{O}|\Psi(A)}}{\braket{\Psi(\overline{A})|\Psi(A)}}=\lim_{N\to\infty} O_N
\end{equation}
where $\{O_N\}$ is an infinite sequence with entries
\begin{equation}
\begin{split}
O_N&=\frac{\sum_{k=1}^{K} (z^{(k)})^{2N-2M} \rbraket{l^{(k)}|Q\otimes\overline{Q}|r^{(k)}}
\rbraket{l^{(k)}| \voperator{E}_{O(-M)}\voperator{E}_{O(-M+1)}\cdots \voperator{E}_{O(M)}|r^{(k)}}}{\sum_{k=1}^{K} (z^{(k)})^{2N+1} \rbraket{l^{(k)}|Q\otimes\overline{Q}|r^{(k)}}}\\
&=\left[\sum_{k=1}^{K} \left(\frac{z^{(k)}}{z^{(1)}}\right)^{2N-2M} \rbraket{l^{(k)}|Q\otimes\overline{Q}|r^{(k)}}\frac{\rbraket{l^{(k)}| \voperator{E}_{O(-M)}\voperator{E}_{O(-M+1)}\cdots \voperator{E}_{O(M)}|r^{(k)}}}{(z^{(1)})^{2M+1}}\right]\\
&\qquad\qquad \times\left[\sum_{k=1}^{K} \left(\frac{z^{(k)}}{z^{(1)}}\right)^{2N+1} \rbraket{l^{(k)}|Q\otimes\overline{Q}|r^{(k)}}\right]^{-1}.
\end{split}
\end{equation}
The entries $O_N$ can be computed using the iterative construction for MPS with open boundary conditions, resulting in a computational complexity $\order(K D^{3})$. Since the number of eigenvalues $K$ with magnitude $\rho(\voperator{E})$ is typically much smaller than $D^{2}$, the increased computational complexity of periodic boundary conditions disappears, provided that we can efficiently determine the $K$ eigenvalues and their corresponding eigenvectors using an iterative eigensolver. However, unless all $K$ eigenvalues $z^{(k)}$ equal a unique positive value or
\begin{displaymath}
\rbraket{l^{(k)}| \voperator{E}_{O(-M)}\voperator{E}_{O(-M+1)}\cdots \voperator{E}_{O(M)}|r^{(k)}}=0
\end{displaymath}
for all eigenvalues that don't, the sequence $\{O_N\}$ is an alternating sequence that does not converge so that the limit in the definition of $O(\overline{A},A)$ does not exist. If all eigenvalues $z^{(k)}=z^{(1)}=\rho(\voperator{E})$ for $k=1,\ldots,K$, the thermodynamic limit is well-defined, but the expectation value still depends on the boundary matrix $Q$. This corresponds to a block decomposition of $A$ into $J\geq K$ blocks, where $\lambda_j=\rho(\voperator{E})$ for $j=1,\ldots,K$ and $\lambda_j<\rho(\voperator{E})$ for $j=K+1,\ldots,J$, and where in addition the individual blocks $A_j$ have a periodic decomposition with $P=1$ for $j=1,\ldots,K$. The MPS $\ket{\Psi(A)}$ is then a superposition of uniform MPS with smaller bond dimensions as in Eq.~(\ref{eq:umps:stateblockdecomposition}), and $Q$ determines the superposition coefficients. Note that the blocks with $j=K+1,\ldots,J$, for which $\lambda_j<\rho(\voperator{E})$ are irrelevant in the thermodynamic limit. If such blocks are present, the state $\ket{\Psi(A)}$ can in the thermodynamic limit be represented by an equivalent uMPS $\ket{\tilde{\Psi}(\tilde{A})}$ with smaller bond dimension $\tilde{D}=\sum_{j=1}^{K} D_j$. Only when $K=1$ and the transfer operator $\voperator{E}$ has a unique eigenvalue with magnitude equal to the spectral radius do we obtain a normalized expectation value $\braket{\Psi(\overline{A})|\operator{O}|\Psi(A)}/\braket{\Psi(\overline{A})|\Psi(A)}$ that is completely independent of $Q$. We then obtain
\begin{equation}
O(\overline{A},A)=\frac{\rbraket{l^{(1)}| \voperator{E}_{O(-M)}\voperator{E}_{O(-M+1)}\cdots \voperator{E}_{O(M)}|r^{(1)}}}{(z^{(1)})^{2M+1}}\label{eq:umps:normalizedexpectationvalue}
\end{equation}
In addition, we do assume that there are no other blocks $A_j$ with $\lambda_j<\rho(\voperator{E})$ ($j=2,\ldots,J$) present in the block decomposition (\textit{i.e.}\ we assume $J=1$), so that $l^{(1)}=l$ and $r^{(1)}=r$ have full rank. The resulting subset of uMPS $A$ with these properties corresponds precisely to the set $\manifold{A}_{\mathrm{uMPS}}$. While there is no difference between open and periodic boundary conditions for normalized expectation values of local operators if we restrict to $A\in\manifold{A}_{\mathrm{uMPS}}$, we do assume we are working with open boundary conditions and denote $Q=\bm{v}_\text{R} \bm{v}_\text{L}^\dagger$ in the definition of the uMPS and its tangent vectors (see next subsection) accordingly. This assumption simplifies the computation of expectation values and allows for certain generalizations that \textit{e.g.} enable the description of topologically non-trivial excitations in Ref.~\onlinecite{2012PhRvB..85j0408H}. We always assume that the transfer matrix $\voperator{E}$ has a non-degenerate eigenvalue $1$ as single eigenvalue on the unit circle, and no physical expectation value will ever depend on on the left and right boundary vectors $\bm{v}_\text{L}$ and $\bm{v}_\text{R}$.

Note that the central object in the discussion above was the normalized expectation value of (local) observables. This quantity can be defined in the whole of Hilbert space $\hilbert_{\manifold{L}}$ and is independent of the norm and phase of a given state. Hence, the normalized expectation value of a local operator can unambiguously be defined for the elements $[\ket{\Psi}]\in P(\hilbert_{\manifold{L}})$. In the thermodynamic limit, it is thus natural to interpret the uMPS representation as the principal fiber bundle $\tilde{\Psi}\colon \manifold{A}_{\mathrm{uMPS}}\to \tilde{\varM}_{\mathrm{uMPS}}$ given by
\begin{equation}
\tilde{\Psi}\colon \manifold{A}_{\mathrm{uMPS}}\to \tilde{\varM}_{\mathrm{uMPS}} \subset P(\hilbert_{\manifold{L}})\colon A\mapsto [\ket{\Psi(A)}],
\end{equation}
which is injective up to the action of a structure group or gauge group
\begin{equation}
\tilde{S}_{\mathrm{uMPS}}=\mathsf{GL}(1,\mathbb{C})\times \mathsf{S}_{\mathrm{uMPS}}\cong \mathsf{GL}(1,\mathbb{C})\times \mathsf{PGL}(D,\mathbb{C})
\end{equation}
with right group action
\begin{equation}
\tilde{\Gamma} \colon \manifold{A}_{\mathrm{uMPS}}\times \tilde{S}_{\mathrm{uMPS}}\to \manifold{A}_{\mathrm{uMPS}}\colon \big(A,(\lambda,[G])\big)\mapsto \lambda A^{(G)}.
\end{equation}
It is clear that we can repeat all the proofs of the last subsection, if we also define a normalized distance function $\tilde{D}_{\mathrm{uMPS}}$ that is invariant under the action of $\tilde{\mathsf{S}}_{\mathrm{uMPS}}$. Firstly, we observe that $D_{\mathrm{uMPS}}$ defined in the previous subsection matches with $D_{\mathrm{MPS}}$ defined in Eq.~\eqref{eq:gmps:distance} in the setting of a translation-invariant thermodynamic limit, up to an overall diverging factor $\lim_{N\to \infty} N=\lvert \mathbb{Z}\rvert$. Such a diverging factor will be encountered often in the remainder of this section. To also obtain invariance under norm and phase changes, we modify the distance function $D_{\mathrm{uMPS}}$ to
\begin{equation}
\tilde{D}_{\mathrm{uMPS}}(A_0,A_1)=\varmin_{A(t)} \int_{0}^{1} \frac{\rbraket{l(t)|\mathbb{E}^{\dot{A}(t)}_{\dot{A}(t)}|r(t)}}{z^{(1)}(t)}\, \rmd t
\end{equation}
with $z^{(1)}(t)=\rho(\mathbb{E}^{A(t)}_{A(t)})$. With this definition at hand, it is easy to prove that $\tilde{\Gamma}\colon \manifold{A}_{\mathrm{uMPS}}\times \tilde{S}_{\mathrm{uMPS}}\to \manifold{A}_{\mathrm{uMPS}}$ is free and proper and the principal fiber bundle construction still holds.

We now introduce some notations that are used throughout the remainder of this section. We have already introduced the left and right eigenvectors $l\defis l^{(1)}$ and $r\defis r^{(1)}$ corresponding to the unique eigenvalue $z^{(1)}=\rho(\voperator{E})$, which are normalized such that $\rbraket{l|r}=1$. All other eigenvalues $z^{(k)}$, $k>1$ satisfy the strict inequality $\lvert z^{(k)}\rvert<\rho(\voperator{E})$. We also define $\voperator{S}^{(1)}=\voperator{S}=\rket{r}\rbra{l}$ as a projector onto the eigenspace of eigenvalue $z^{(1)}$, and its complement $\voperator{Q}=\voperator{\one}-\voperator{S}$. Since physical states are now living within the projective space $\tilde{\varM}_{\mathrm{uMPS}}$, we can and often will use a point $A$ on the gauge orbit for which $\rho(\voperator{E})=1$. Put differently, we often `renormalize' $A\leftarrow A/\sqrt{\rho(\voperator{E})}$ such that $\rho(\voperator{E})=1$ for reasons of simplicity, \textit{i.e.}\ this allows one to eliminate the denominator in the normalized expectation value of operators [Eq.~(\ref{eq:umps:normalizedexpectationvalue})]. Given a set of local operators $\operator{O}^\alpha$, we can use these definitions to compute the 2-point connected correlation function as 
\begin{equation}
\begin{split}
\Gamma^{(\alpha,\beta)}(n)&=\rbraket{l|\voperator{E}_{O^{\alpha}}\voperator{E}^{n-1}\voperator{E}_{O^{\beta}}|r}-\rbraket{l|\voperator{E}_{O^{\alpha}}|r}\rbraket{l|\voperator{E}_{O^{\beta}}|r}\\
&=\rbraket{l|\voperator{E}_{O^{\alpha}}\voperator{Q}\big(\voperator{Q}\voperator{E}\voperator{Q}\big)^{n-1}\voperator{Q}\voperator{E}_{O^{\beta}}|r}.
\end{split}
\end{equation}
where we have used
\begin{displaymath}
\voperator{E}^{n}-\voperator{S}=\voperator{Q}\voperator{E}^n\voperator{Q}=\voperator{Q}\big(\voperator{Q}\voperator{E}\voperator{Q}\big)^n\voperator{Q}.
\end{displaymath}
The correlation length $\xi$ is then determined by the largest eigenvalue of $\voperator{Q}\voperator{E}\voperator{Q}$ as
\begin{equation}
\xi=-\frac{1}{\log \left[\rho(\voperator{Q}\voperator{E}\voperator{Q})\right]}.
\end{equation}
Under the given assumption, $\rho(\voperator{Q}\voperator{E}\voperator{Q})<1$ and the correlation length $\xi$ is finite. Hence, all pure uMPS are exponentially clustering. The correlation length is determined by $\rho(\voperator{Q}\voperator{E}\voperator{Q})$, which is equal to the eigenvalue of the transfer matrix $\voperator{E}$ that is second largest in absolute value.

Finally, we compute the normalized overlap between two uMPS as
\begin{equation}
F(\overline{A},A^{\prime})=\frac{\lvert\braket{\Psi(\overline{A})|\Psi(A^{\prime})}\rvert}{\braket{\Psi(\overline{A})|\Psi(A)}^{1/2}\braket{\Psi(\overline{A}^{\prime})|\Psi(A^{\prime})}^{1/2}}\sim \lim_{N\to\infty} \left[\frac{\rho(\voperator{E}^{A^{\prime}}_{A})}{\rho(\voperator{E}^A_A)^{1/2} \rho(\voperator{E}^{A^{\prime}}_{A^{\prime}})^{1/2}}\right]^{2N+1},
\end{equation}
We can thus define $d(\overline{A},A^{\prime})=\rho(\voperator{E}^{A^{\prime}}_{A})\rho(\voperator{E}^A_A)^{-1/2} \rho(\voperator{E}^{A^{\prime}}_{A^{\prime}})^{-1/2}$ as the fidelity per site between the two uMPS $\ket{\Psi(A)}$ and $\ket{\Psi(A^{\prime})}$. It can easily be shown that this definition is compatible with $d(\overline{A},A^{\prime})\leq 1$. In addition, under the given conditions for injective uMPS $A,A^{\prime}\in\manifold{A}_{\mathrm{uMPS}}$, it was proven in Ref.~\onlinecite{2008PhRvL.100p7202P} that  $d(A,A^{\prime})=1$ implies that the two states are equivalent, such that there exists $(\lambda,[G])\in\tilde{\mathsf{S}}_{\mathrm{uMPS}}$ for which $A^{\prime}=\Gamma(A,(\lambda,[G]))=\lambda A^{(G)}$. Alternatively, $d(\overline{A},A^{\prime})<1$ corresponds to gauge-inequivalent states and implies $F(A,A^{\prime})=0$ due to the orthogonality catastrophe\cite{Anderson:1967aa}: any two inequivalent injective uMPS are automatically orthogonal in the thermodynamic limit. In conclusion, the fidelity (per site) is a useful tool to check whether to uMPS $A,A^{\prime}\in\manifold{A}_{\mathrm{uMPS}}$ belong to the same gauge orbit.

\subsection{Tangent bundles and the principal Ehresmann connection}
\label{ss:umps:tangent}
As in the previous section, we can now define a bundle map $\rmd\tilde{\Psi}$ between the holomorphic tangent bundles $T\manifold{A}_{\mathrm{uMPS}}$ and $T\tilde{\varM}_{\mathrm{uMPS}}$. At any point in $A\in\manifold{A}_{\mathrm{uMPS}}$ we have $T_A\manifold{A}_{\mathrm{uMPS}}\cong \mathbb{A}_{\mathrm{uMPS}}\cong \mathbb{C}^{D\times d \times D}$. While the projective Hilbert space $P(\hilbert_{\manifold{L}})$ is the most natural choice in the thermodynamic limit, it is conceptually simpler to work with the affine Hilbert space $\hilbert_{\manifold{L}}$. We return to the projective setting at the end of this subsection. For now, we also define the bundle map $\rmd\Psi:T\manifold{A}_{\mathrm{uMPS}}\to T\varM_{\mathrm{uMPS}}$. As in the previous section, we denote $T_{\ket{\Psi(A)}}\varM_{\mathrm{uMPS}}=\mathbb{T}^{(A)}_{\mathrm{uMPS}}$ and we introduce a map $\Phi\colon T\manifold{A}_{\mathrm{uMPS}}\to \mathbb{T}^{(A)}_{\mathrm{uMPS}}\colon (B,A)\mapsto \ket{\Phi(B;A)}=\ket{\Phi^{(A)}(B)}$ using the prescription
\begin{equation}
\begin{split}
\ket{\Phi(B;A)}=\ket{\Phi^{(A)}(B)}&=B^{i}\frac{\partial\ }{\partial A^{i}} \ket{\Psi(A)}\\
&=\sum_{n\in\mathbb{Z}}\sum_{\{s_{n}\}=1}^{q} \bm{v}_{\mathrm{L}}^{\dagger}\left[\left(\prod_{m<n} A^{s_{m}}\right) B^{s_{n}} \left(\prod_{m'>n} A^{s_{m'}}\right)\right]\bm{v}_{\mathrm{R}} \ket{\{s_{n}\}}.
\end{split}
\label{eq:umps:defumpstangent}
\end{equation}

Note that the tangent space $\mathbb{T}^{(A)}_{\mathrm{uMPS}}$ contains only translation-invariant states. However, we can interpret a uMPS $A\in\manifold{A}_{\mathrm{uMPS}}$ as a special point within the class of generic MPS $\mathbb{A}_{\mathrm{MPS}\{D_n=D\}}$ by identifying $\{A(n)=A\}_{n\in\mathbb{Z}}\in\mathbb{A}_{\mathrm{MPS}(D)}\defis \mathbb{A}_{\mathrm{MPS}\{D_n=D\}}=\prod_{n\in\mathbb{Z}}\mathbb{C}^{D\times q\times D}$, such that $\ket{\Psi[\{A(n)=A\}_{n\in\mathbb{Z}}]}=\ket{\Psi(A)}$, where we heavily overload the notation $\Psi$. Clearly, we have $\varM_{\mathrm{uMPS}}\subset \varM_{\mathrm{MPS}(D)}\defis \varM_{\mathrm{MPS}\{D_n=D\}}$. Henceforth, we discard the notation of $(D)$ or $\{D_n=D\}$ in the definition of the MPS spaces. In principle, we should now take proper care to define the subset of injective MPS $\manifold{A}_{\mathrm{MPS}}$ in the thermodynamic limit, since we have only done this for a finite lattice with open boundary conditions in the previous section and for uniform MPS in the current section. However, we can expect that $\manifold{A}_{\mathrm{MPS}}$ is an open set which includes the points $\{A(n)=A\}_{n\in\mathbb{Z}}$ for any $A\in\manifold{A}_{\mathrm{uMPS}}$. Using the identification $A\mapsto\{A(n)=A\}_{n\in\mathbb{Z}}$, we can embed $\manifold{A}_{\mathrm{uMPS}}$ as a subset in $\manifold{A}_{\mathrm{MPS}}$. In addition, with $\manifold{A}_{\mathrm{MPS}}$ being an open set, it contains a neighborhood around every point $\{A(n)=A\}_{n\in\mathbb{Z}}$ and we expect $T_{\{A(n)=A\}_{n\in\mathbb{Z}}}\manifold{A}_{\mathrm{MPS}}\cong \mathbb{A}_{\mathrm{MPS}}$. We denote $T_{\ket{\Psi[\{A(n)=A\}_{n\in\mathbb{Z}}]}}\varM_{\mathrm{MPS}}$ as $\mathbb{T}_{\mathrm{MPS}}^{(A)}$. For any $A\in\manifold{A}_{\mathrm{uMPS}}$, we can then define a map
\begin{equation}
\Phi^{(A)}\colon T_{\{A(n)=A\}_{n\in\mathbb{Z}}}\manifold{A}_{\mathrm{MPS}}\cong \mathbb{A}_{\mathrm{MPS}}\mapsto \mathbb{T}_{\mathrm{MPS}}^{(A)}\colon (A,B=\{B(n)\}_{n\in\mathbb{Z}})\mapsto \ket{\Phi^{(A)}[B]}
\end{equation}
with the prescription
\begin{equation}
\ket{\Phi^{(A)}[B]}=\sum_{n\in\mathbb{Z}}\sum_{\{s_{n}\}=1}^{q} \bm{v}_{\mathrm{L}}^{\dagger}\left[\left(\prod_{m<n} A^{s_{m}}\right) B^{s_{n}}(n) \left(\prod_{m'>n} A^{s_{m'}}\right)\right]\bm{v}_{\mathrm{R}} \ket{\{s_{n}\}}.
\end{equation}
As before, the notation of $\Phi$ is heavily overloaded and the difference is indicated by the fact whether the argument is contained in round or square brackets. Evidently, we also have $\mathbb{T}^{(A)}_{\mathrm{uMPS}}\subset \mathbb{T}^{(A)}_{\mathrm{MPS}}$.

As a final observation, we recall that in many problems with translation invariance, Hilbert space can be decomposed into different momentum sectors. Here too, we can introduce a new set of definitions, by writing
\begin{equation}
T_{\{A(n)=A\}_{n\in\mathbb{Z}}}\manifold{A}_{\mathrm{MPS}}\cong \mathbb{A}_{\mathrm{MPS}}\cong \bigoplus_{n\in\mathbb{Z}} \mathbb{C}^{D\times q\times D}=\int_{p\in[-\pi,\pi)}^{\oplus} \mathbb{A}_p
\end{equation}
where $\mathbb{A}_{p}\cong\mathbb{A}_{\mathrm{uMPS}}\cong\mathbb{C}^{D\times q\times D}$ and we identify $B\in\mathbb{A}_{p}$ with $\{B(n)=B\rme^{\rmi p n}\}_{n\in\mathbb{Z}}\in\mathbb{A}_{\mathrm{MPS}}$. Analogously, we also make a momentum space decomposition of the tangent space as
\begin{equation}
\mathbb{T}^{(A)}_{\mathrm{MPS}}=\int_{p\in[-\pi,\pi)}^{\oplus}\mathbb{T}^{(A)}_p.
\end{equation}
For every $A\in\manifold{A}_{\mathrm{uMPS}}$, we then define a final map $\Phi^{(A)}_p:\mathbb{A}_p\to \mathbb{T}^{(A)}_{p}:B\mapsto \ket{\Phi^{(A)}_p(B)}$ using the prescription
\begin{equation}
\ket{\Phi^{(A)}_{p}(B)}=\sum_{n\in\mathbb{Z}}\rme^{\rmi p n}\sum_{\{s_{n}\}=1}^{q} \bm{v}_{\mathrm{L}}^{\dagger}\left[\left(\prod_{m<n} A^{s_{m}}\right) B^{s_{n}} \left(\prod_{m'>n} A^{s_{m'}}\right)\right]\bm{v}_{\mathrm{R}} \ket{\{s_{n}\}}.\label{eq:umps:deftangentk}
\end{equation}
Note that $\ket{\Phi_{0}^{(A)}(B)}=\ket{\Phi^{(A)}(B)}$ and thus $\mathbb{T}^{(A)}_{0}=\mathbb{T}^{(A)}_{\mathrm{uMPS}}$, the set of translation-invariant tangent vectors. These notations are used interchangeably.

The different representations $\Phi$ of tangent vectors also have a large representation redundancy. For the principal fiber bundle $\Psi:\manifold{A}_{\mathrm{MPS}}\to\varM_{\mathrm{MPS}}$, the structure group is given by $\group{S}_{\mathrm{MPS}}=\group{G}_{\mathrm{MPS}}/\mathsf{GL}(1,\mathbb{C})$ with $\group{G}_{\mathrm{MPS}}\cong\prod_{n\in\mathbb{Z}}\mathsf{GL}(D,\mathbb{C})$ and the normal subgroup $\mathsf{GL}(1,\mathbb{C})\subset\group{G}_{\mathrm{MPS}}$ given by $\mathsf{GL}(1,\mathbb{C})=\{ \{G(n)=c \one_D\}_{n\in\mathbb{Z}}|c\in\mathbb{C}_{0}\}$. However, since we restrict to uniform elements $A\in\manifold{A}_{\mathrm{uMPS}}\subset \manifold{A}_{\mathrm{MPS}}$ using the embedding discussed above, we should only consider the invariance of $\ket{\Phi^{(A)}[B]}$ under the translation-invariant action of the group $\mathsf{S}_{\mathrm{uMPS}}\cong \mathsf{PGL}(D,\mathbb{C})\cong \mathsf{GL}(D,\mathbb{C})/\mathsf{GL}(1,\mathbb{C})$. We can also embed $\mathsf{S}_{\mathrm{uMPS}}$ as a subgroup of $\mathsf{S}_{\mathrm{MPS}}$ by identifying $[G]\in\mathsf{S}_{\mathrm{uMPS}}$ with $[\{G(n)=G\}_{n\in\mathbb{Z}}]\in\group{S}_{\mathrm{MPS}}$, which is independent of the chosen element from the coset $[G]$. We then obtain, for any $[G]\in\mathsf{S}_{\mathrm{uMPS}}$, that $\ket{\Phi^{(A^{(G)})}[B^{(G)}]}=\ket{\Phi^{(A)}[B]}$ where the group action $B^{(G)}=\{B^{(G)}(n)\}_{n\in\mathbb{Z}}$ is defined by
\begin{equation}
\left(B^{(G)}\right)^s(n)=G^{-1} B^s(n) G
\end{equation}
and is independent of the chosen element from the coset $[G]$. At a fixed point $\{A(n)=A\}_{n\in\mathbb{Z}}$, the linear homomorphism $\Phi^{(A)}:\mathbb{A}_{\mathrm{MPS}}\to\mathbb{T}^{(A)}_{\mathrm{MPS}}$ has a null space $\mathbb{N}^{(A)}$ that is isomorphic to the full (translation non-invariant) group algebra $\mathfrak{s}_{\mathrm{MPS}}$, given by
\begin{equation}
\mathfrak{s}_{\mathrm{MPS}}=\left\{x=\{x(n)\}_{n\in\mathbb{Z}}\in \mathfrak{g}_{\mathrm{MPS}}\cong \bigoplus_{n\in\mathbb{Z}} \mathfrak{gl}(D,\mathbb{C})\cong \bigoplus_{n\in\mathbb{Z}} \mathbb{C}^{D\times D}\,\right\vert\ \left.\,\sum_{n\in\mathbb{Z}} \tr\left[x(n)\right] = 0\right\}
\end{equation}
It is now easier to decompose the group algebra into the different momentum sectors $\mathfrak{s}_{\mathrm{MPS}}=\int_{p\in[-\pi,\pi)}^{\oplus} \mathfrak{s}_p$ with
\begin{equation}
\mathfrak{s}_p=\begin{cases} \mathfrak{gl}(D,\mathbb{C})=\mathbb{C}^{D\times D},&p\neq 0,\\
\mathfrak{pgl}(D,\mathbb{C})\cong \mathfrak{sl}(D,\mathbb{C})=\{x\in\mathbb{C}^{D\times D}| \tr[x]=0\},&p=0.\end{cases}
\end{equation}
An element $x\in\mathfrak{s}_p$ is identified with $\{x(n)=x\rme^{\rmi p n}\}_{n\in\mathbb{Z}}\in\mathfrak{s}_{\mathrm{MPS}}$. Every $\mathfrak{s}_p$ is isomorphic to the null space $\mathbb{N}_p^{(A)}\subset \mathbb{A}_p$ of the map $\Phi^{(A)}_p$ using the isomorphism
\begin{equation}
\mathscr{N}^{(A)}_p:\mathfrak{s}_p\to \mathbb{N}^{(A)}_p\subset \mathbb{A}_p: x \mapsto \mathscr{N}^{(A)}_p(x)
\end{equation}
with $(\mathscr{N}^{(A)}_p)^s(x)=  A^s x - \rme^{-\rmi p} x A^s$. Hence, in momentum space, the different spaces completely decouple and we can simply work with finite-dimensional vector spaces. For the given prescription of $\mathscr{N}^{(A)}_p(x)$, the embedding of $\mathfrak{s}_p$ in $\mathfrak{s}_{\mathrm{MPS}}$ is compatible with the embedding of $\mathbb{A}_p$ in $\mathbb{A}_{\mathrm{MPS}}$. Note also that $\mathscr{N}^{(A)}_{p=0}(\one_D)=0$, which is in accordance with the restriction of $\mathfrak{s}_{p=0}$ to the set of traceless matrices.

Having defined the vertical subspaces $\mathbb{N}^{(A)}_p$ of $\mathbb{A}_p$, we should now introduce a principal connection $\omega^{(A)}_p:\mathbb{A}_{p}\to \mathfrak{s}_p$ that defines a complementary horizontal subspace $\mathbb{B}^{(A)}_p=\ker \omega^{(A)}$ such that $\mathbb{A}_{p}=\mathbb{B}^{(A)}_p\oplus\mathbb{N}^{(A)}_p$. The principal connection should act as $\omega^{(A)}_p\left(\mathscr{N}^{(A)}_p(x)\right)=x$ for any $x\in\mathfrak{s}_p$ and transform equivariantly as $\omega^{(A^{(G)})}_p\left(B^{(G)}\right)=\mathrm{Ad}_{G^{-1}}\left(\omega^{(A)}_p(B)\right)$ for any $B\in\mathbb{A}_p$ and any $[G]\in\mathsf{S}_{\mathrm{uMPS}}$. We first consider the case $p\neq 0$. The virtual operators $z^{(1)}\one-\rme^{\pm\rmi p}\voperator{E}$ have no zero eigenvalues and can thus be inverted, since $\voperator{E}$ has a unique eigenvalue of magnitude $z^{(1)}$, namely $z^{(1)}$ itself. Let $\mathscr{F}_p$ and $\tilde{\mathscr{F}}_p$ denote the inverse of the corresponding maps $z^{(1)} \mathrm{Id}-\rme^{+\rmi p} \mathscr{E}$ and $z^{(1)} \mathrm{Id}-\rme^{-\rmi p} \tilde{\mathscr{E}}$. Two valid choices for a principal connection are given by
\begin{align}
\omega^{(A,\mathrm{L})}_p(B)&=l^{-1} \tilde{\mathscr{F}}_p\left(\sum_{s=1}^q (A^s)^\dagger l B^s\right),&\omega^{(A,\mathrm{R})}_p(B)&=-\rme^{+\rmi p}\mathscr{F}_p\left(\sum_{s=1}^q B^s r (A^s)^\dagger \right) r^{-1}.\label{eq:umps:connection}
\end{align}
Equivarience can be checked by noting that $l^{(G)}=G^\dagger l G$ and $r^{(G)}=G^{-1} r (G^{-1})^\dagger$. Note that, while the transformation of $l$ and $r$ depends on the chosen element $G$ from the coset $[G]\in\mathsf{S}_{\mathrm{uMPS}}$, the principal connections $\omega^{(A,\mathrm{L})}$ and $\omega^{(A,\mathrm{R})}$ do not, nor does the normalization condition $\tr[l r]=\rbraket{l|r}=1$. The $p=0$ case requires a special treatment because the maps $z^{(1)}\mathrm{Id}-\mathscr{E}$ and $z^{(1)}\mathrm{Id}-\tilde{\mathscr{E}}$ corresponding to the right and left action of the operator $z^{(1)}\one-\voperator{E}$ are not invertible. There is a unique eigenvalue zero with corresponding projector $\voperator{S}=\rket{r}\rbra{l}$. However, for $B=\mathscr{N}^{(A)}_{p=0}(x)$, it can easily be checked that $\rbraket{l|\voperator{E}^{B}_{A}|r}=0$. Hence, we can define a pseudo-inverse $(z^{(1)}\one-\voperator{E})^{(-1)}$ with the property that
\begin{equation}
(z^{(1)}\one-\voperator{E})^{(-1)}(z^{(1)}\one-\voperator{E})=(z^{(1)}\one-\voperator{E})(z^{(1)}\one-\voperator{E})^{(-1)}=\one-\voperator{S}=\voperator{Q}.\label{eq:umps:pseudop0}
\end{equation}
The maps associated to the right and left action of $(z^{(1)}\one-\voperator{E})^{(-1)}$ are denoted as $\mathscr{F}_{p=0}$ and $\tilde{\mathscr{F}}_{p=0}$. With these definitions, the principal connections defined in Eq.~(\ref{eq:umps:connection}) are also valid for $p=0$. Vectors $B\in\mathbb{B}^{(A)}_p$ defined by $\mathbb{B}^{(A)}_p=\ker \omega_p^{(A,\mathrm{L})}$ satisfy
\begin{equation}
\begin{cases}
\rbra{l}\voperator{E}^B_A=0\Leftrightarrow\sum_{s=1}^q (A^s)^\dagger l B^s=0 ,&p\neq 0,\\
\rbra{l}\voperator{E}^B_A\voperator{Q}=0 \Leftrightarrow\sum_{s=1}^q (A^s)^\dagger l B^s=  l \tr\left[\sum_{s=1}^q (A^s)^\dagger l B^s r\right],&p=0,
\end{cases}
\end{equation}
which can be called the \emph{left gauge-fixing conditions}. Vectors $B\in\mathbb{B}^{(A)}_p$ satisfy these conditions and can therefore be said to be in the left-canonical form. Similarly, vectors $B\in\mathbb{B}^{(A)}_p$ defined by $\mathbb{B}^{(A)}_p=\ker \omega_p^{(A,\mathrm{R})}$ are in the right canonical form and satisfy the \emph{right gauge-fixing conditions}
\begin{equation}
\begin{cases}
\voperator{E}^B_A\rket{r}=0\Leftrightarrow \sum_{s=1}^q B^s r (A^s)^\dagger=0,&p\neq 0,\\
\voperator{Q}\voperator{E}^B_A\rket{r}=0 \Leftrightarrow \sum_{s=1}^q B^s r (A^s)^\dagger=r \tr\left[l \sum_{s=1}^q B^s r (A^s)^\dagger\right] ,&p=0.
\end{cases}
\end{equation}

Finally, we can return to the projective case. We can run through the same steps in order to define an enlarged tangent space $T_{[\ket{\Psi(A)}}\tilde{\varM}_{\mathrm{MPS}}$ with the uMPS $[\ket{\Psi(A)}]\in\tilde{\varM}_{\mathrm{uMPS}}$ at its base. We can write $\tilde{\mathbb{T}}_{\mathrm{MPS}}^{(A)}\defis T_{[\ket{\Psi(A)}}\tilde{\varM}_{\mathrm{MPS}}\cong T_{\ket{\Psi(A)}}\varM_{\mathrm{MPS}}/\sim$, where two tangent vectors $\ket{\Phi^{(A)}[B_1]}$ and $\ket{\Phi^{(A)}[B_2]}$ are equivalent if there exists some $\alpha\in\mathbb{C}$ such that $\ket{\Phi^{(A)}[B_1]}-\ket{\Phi^{(A)}[B_2]}=\alpha \ket{\Psi(A)}$. If we now make a momentum decomposition $\tilde{\mathbb{T}}_{\mathrm{MPS}}^{(A)}=\int_{p\in[-\pi,\pi)}^{\oplus} \tilde{\mathbb{T}}_p^{(A)}$, it is easily obtained that $\tilde{\mathbb{T}}_p^{(A)}\cong \mathbb{T}^{(A)}_p$ for any $p\neq 0$, whereas $\tilde{\mathbb{T}}_0^{(A)}=\tilde{\mathbb{T}}^{(A)}_{\mathrm{uMPS}}=T_{[\ket{\Psi(A)}]}\tilde{\varM}_{\mathrm{uMPS}}\cong \mathbb{T}^{(A)}_0/\sim\cong \mathbb{T}^{(A)\perp}_0$, where we have chosen to represent the different inequivalent vectors living in the quotient space by the unique representative that is orthogonal to $\ket{\Psi(A)}$. In accordance, we can now define a symmetry group $\tilde{\mathsf{S}}_{\mathrm{MPS}}\defis \mathsf{GL}(1,\mathbb{C})\times \mathsf{S}_{\mathrm{MPS}}$ and corresponding group algebra $\tilde{\mathfrak{s}}_{\mathrm{MPS}}=\int_{p\in[-\pi,\pi)}^{\oplus} \tilde{\mathfrak{s}}_p$. Since the additional $\mathsf{GL}(1,\mathbb{C})$ symmetry operations of norm and phase changes are translation-invariant operations, we obtain $\tilde{\mathfrak{s}}_p\cong \mathfrak{s}_p$ for any $p\neq 0$, whereas $\tilde{\mathfrak{s}}_0\cong\mathfrak{gl}(1,\mathbb{C})\oplus\mathfrak{s}_0 \cong\mathbb{C}\oplus\mathfrak{s}_0$. We can define a vertical subspace $\tilde{\mathbb{N}}^{(A)}_p\subset \mathbb{A}_{p}\cong \mathbb{A}_{\mathrm{uMPS}}$ that is isomorphic to $\tilde{\mathfrak{s}}_p$ using the map
\begin{equation}
\tilde{\mathscr{N}}_p^{(A)}\colon \tilde{\mathfrak{s}}_p \to \tilde{\mathbb{N}}^{(A)}_p\colon \begin{cases} x \mapsto \tilde{\mathscr{N}}_p^{(A)}(x)=\mathscr{N}_p^{(A)}(x),&p\neq 0,\\
(\alpha,x)\mapsto \tilde{\mathscr{N}}_0^{(A)}(\alpha;x)= \mathscr{N}_0^{(A)}(x)+\alpha A,&p=0.\end{cases}
\end{equation}
A complementary horizontal subspace $\tilde{\mathbb{B}}_p^{(A)}$ is defined as $\tilde{\mathbb{B}}_p^{(A)}=\ker \tilde{\omega}^{(A)}_p$, where the principal connection $\tilde{\omega}^{(A)}_p\colon\mathbb{A}_p\to \tilde{\mathfrak{s}}_p$ can be defined as
\begin{equation}
\tilde{\omega}^{(A)}_p\colon\mathbb{A}_p\to \tilde{\mathfrak{s}}_p\colon\begin{cases}
B\mapsto \tilde{\omega}^{(A)}_p(B)=\omega^{(A)}_p(B),&p\neq 0,\\
B\mapsto \tilde{\omega}^{(A)}_0(B)=\left(\frac{\rbraket{l|\voperator{E}^B_A|r}}{z^{(1)}},\omega^{(A)}_0(B)\right),&p=0,
\end{cases}
\end{equation}
where $\omega_p^{(A)}$ is a valid principal connection in the affine case, such as $\omega_p^{(A,\mathrm{L})}$ or $\omega_p^{(A,\mathrm{R})}$ defined in Eq.~(\eqref{eq:umps:connection}). The corresponding connections in the projective case are denoted as $\tilde{\omega}_p^{(A,\mathrm{L})}$ and $\tilde{\omega}_p^{(A,\mathrm{R})}$ respectively. Vectors $B\in\tilde{\mathbb{B}}^{(A)}_p=\ker \omega_p^{(A,\mathrm{L})}$ now satisfy the left gauge-fixing conditions
\begin{equation}
\rbra{l}\voperator{E}^{B}_A=0\label{eq:umps:leftgaugeuB}
\end{equation}
for any choice of $p$, both $p\neq 0$ and $p=0$. Analogously, for $B\in\tilde{\mathbb{B}}^{(A)}_p=\ker \omega_p^{(A,\mathrm{R})}$ we obtain the right gauge-fixing conditions
\begin{equation}
\voperator{E}^{B}_A\rket{r}=0\label{eq:umps:rightgaugeuB}
\end{equation}
for any $p$ including $p=0$. In conclusion, we now compute the overlap between any tangent vector $\ket{\Phi^{(A)}_p(B)}$ and the original uMPS $\ket{\Psi(A)}$. Henceforth, we always work in the projective setting but with representatives $\ket{\Psi(A)}\in\varM_{\mathrm{MPS}}$ instead of the equivalence classes $[\ket{\Psi(A)}]\in\tilde{\varM}_{\mathrm{MPS}}$. However, we choose these representatives such that $\braket{\Psi(A)|\Psi(A)}=1$ whenever we are computing with $\ket{\Psi(A)}$ or $\ket{\Psi^{(B)}}$. Put differently, we always act with an element of the $\mathsf{GL}(1,\mathbb{C})$ subgroup of $\tilde{\mathsf{S}}_{\mathrm{uMPS}}$ so as to have a representation $A\in\manifold{A}_{\mathrm{uMPS}}$ in the fiber corresponding to $[\ket{\Psi(A)}]\in\tilde{\varM}_{\mathrm{uMPS}}$ for which $z^{(1)}=\rho(\voperator{E})=1$. Correspondingly, we also obtain $\voperator{Q}\voperator{E}\voperator{Q}=\voperator{E}-\voperator{S}$. Based on the definition of the subset $\manifold{A}_{\mathrm{uMPS}}$ of injective MPS, we then have $\rho(\voperator{Q}\voperator{E}\voperator{Q})<1$ for the chosen representations $A$. Under this condition, the overlap between any uMPS $\ket{\Psi(A)}$ and $\ket{\Phi^{(A)}_p(B)}$ is given by
\begin{equation}
\braket{\Psi(\overline{A})|\Phi_{p}^{(A)}(B)}=\sum_{n\in\mathbb{Z}} \rme^{\rmi p n}\rbraket{l|\voperator{E}^{B}_{A}|r}= 2\pi \delta(p) \rbraket{l|\voperator{E}^{B}_{A}|r}\label{eq:umps:tangentoverlapwithumps}
\end{equation}
so that all states $\ket{\Phi_{p}(B)}$ with $p\neq 0$ are automatically orthogonal to $\ket{\Psi(A)}$ due to the orthogonality of the different momentum sectors. At momentum $p=0$, we observe that for $B\in\tilde{\mathbb{B}}^{(A)}_0$ defined by either the left or right principal connection $\tilde{\omega}^{(A)}_{0}$, we also obtain $\braket{\Psi(\overline{A})|\Phi_{p}^{(A)}(B)}=0$, in correspondence with our expectation $\tilde{\mathbb{T}}^{(A)}_0\cong \mathbb{T}^{(A)\perp}_0$. For the choice $B=A$, which is in $\mathbb{B}^{(A)}_0$ but not in $\tilde{B}^{(A)}_0$, we obtain $\braket{\Psi(\overline{A})|\Phi_{p}^{(A)}(B)}=2\pi\delta(0)=\lvert\mathbb{Z}\rvert$, where the cardinality $\lvert\mathbb{Z}\rvert$ represents the diverging number of lattice sites ($\mathcal{L}=\mathbb{Z}$). This explains why we work within the projective setting throughout the remainder of this section.

\subsection{Pullback metric and efficient parametrization}
\label{ss:umps:metric}
Once again, we have come to the point where we can induce the Fubini-Study metric from $P(\hilbert)$ onto $\tilde{\varM}_{\mathrm{uMPS}}$ in order to transform it into a K\"{a}hler manifold, and then define a pullback metric $\tilde{g}$ on $\manifold{A}_{\mathrm{uMPS}}$. As just mentioned, we now exclusively treat the projective setting. Nevertheless, we do also compute the pullback $g$ of the natural metric of affine Hilbert space, since it features in the computation of $\tilde{g}$ and it provides further justification for the restriction to the projective setting. 

The pullback metric $g(p,p';\overline{A},A)$ is implicitly defined by $\braket{\Phi_{p}^{(A)}(\overline{B})|\Phi^{(A)}_{p'}(B')}$. Henceforth, we discard again the explicit notation of the base point $(A)$ at which we are working, since this is fixed throughout the remainder of this subsection. We have to be very careful with the infinite sums over the positions $n\in \mathbb{Z}$ and $n'\in\mathbb{Z}$ of $B$ and $B'$. When a diverging result is obtained, it is easily possible to make errors by miscounting. Only when the result is guaranteed to be finite can we freely use index substitutions. We therefore replace every occurrence of $\voperator{E}^{n}$ by a `regularized' operator $\voperator{Q}\voperator{E}^{n}\voperator{Q}=\voperator{E}^{n}\voperator{Q}=\voperator{Q}\voperator{E}^{n}=\voperator{E}^{n}-\voperator{S}=\voperator{Q}(\voperator{Q}\voperator{E}\voperator{Q})^{n}\voperator{Q}$ with $\rho(\voperator{Q}\voperator{E}\voperator{Q})<1$ and a `singular' part $\voperator{S}=\rket{r}\rbra{l}$. The reason of this notation becomes clear if we now evaluate $\braket{\Phi_{p}(\overline{B})|\Phi_{p'}(B')}$ as
\begin{equation*}
\begin{split}
\langle\Phi_{p}(\overline{B})&\mid\Phi_{p'}(B')\rangle= \overline{B}^{\overline{\imath}} g_{\overline{\imath},j}(p,p') {B'}^j\\
\qquad=& \sum_{n=-\infty}^{+\infty}\sum_{n'=-\infty}^{+\infty}\rme^{+\rmi p' n' - \rmi p n}\left[\theta(n=n')\rbraket{l|\voperator{E}^{B'}_{B}|r}\right.\\
&\qquad\qquad\qquad\qquad\qquad\left.+ \theta(n'>n) \rbraket{l|\voperator{E}^{A}_{B} (\voperator{E})^{n'-n-1}\voperator{E}^{B'}_{A}|r}+ \theta(n'<n) \rbraket{l|\voperator{E}^{B'}_{A} (\voperator{E})^{n-n'-1}E^{A}_{B}|r}\right]\\
=&\sum_{n_{0}=-\infty}^{+\infty}\rme^{\rmi (p'-p)n_{0}}\sum_{\Delta n=-\infty}^{+\infty}\rme^{\rmi p \Delta n}\left[\theta(\Delta n=0)\rbraket{l|\voperator{E}^{B'}_{B}|r}\right.\\
&\qquad\qquad\qquad\qquad\qquad\left.+\theta(\Delta n > 0) \rbraket{l|\voperator{E}^{A}_{B} \voperator{Q}\voperator{E}^{\Delta n-1}\voperator{Q}\voperator{E}^{B'}_{A}|r}+\theta(\Delta n<0) \rbraket{l|\voperator{E}^{B'}_{A} \voperator{Q}\voperator{E}^{-\Delta n-1}\voperator{Q}\voperator{E}^{A}_{B}|r}\right]\\
&+\rbraket{l|\voperator{E}^{A}_{B}|r}\rbraket{l|\voperator{E}^{B'}_{A}|r} \sum_{n=-\infty}^{+\infty}\sum_{n'=-\infty}^{n-1}\rme^{\rmi p' n' - \rmi p n}+\rbraket{l|\voperator{E}^{B'}_{A}|r}\rbraket{l|\voperator{E}^{A}_{B}|r} \sum_{n=-\infty}^{+\infty}\sum_{n'=n+1}^{+\infty}\rme^{\rmi p' n' - \rmi p n}.
\end{split}
\end{equation*}
In this calculation, we have introduced a 'discrete' Heaviside function $\theta$ taking a logical expression as argument and resulting $1$ if the argument is true and zero otherwise. By using the well known result for the geometric series of an operator with spectral radius smaller than one, we obtain
\begin{equation}
\sum_{n=0}^{+\infty}\voperator{Q}\voperator{E}^{n}\voperator{Q}=\sum_{n=0}^{+\infty}\voperator{Q}(\voperator{Q}\voperator{E}\voperator{Q})^{n}\voperator{Q}=\voperator{Q}(\voperator{\one}-\voperator{Q}\voperator{E}\voperator{Q})^{-1}\voperator{Q}
\end{equation}
and thus
\begin{equation}
\begin{split}
\braket{\Phi_{p}(\overline{B})|\Phi_{p'}(B')}=& \overline{B}^{\overline{\imath}} g_{\overline{\imath},j}(p,p') {B'}^j=2\pi \delta(p-p')\overline{B}^{\overline{\imath}} g_{\overline{\imath},j}(p) {B'}^j\\
=&2\pi\delta(p'-p)\left[\rbraket{l|\voperator{E}^{B'}_{B}|r}+\rbraket{l|\voperator{E}^{A}_{B} \voperator{Q} (\voperator{\one}-\rme^{\rmi p}\voperator{Q}\voperator{E}\voperator{Q})^{-1}\voperator{Q}\voperator{E}^{B'}_{A}|r}\right.\\
&\qquad\left.+\rbraket{l|\voperator{E}^{B'}_{A} \voperator{Q}(\voperator{\one}-\rme^{-\rmi p}\voperator{Q}\voperator{E}\voperator{Q})^{-1}\voperator{Q}\voperator{E}^{A}_{B}|r}
-\rbraket{l|\voperator{E}^{B'}_{A}|r}\rbraket{l|\voperator{E}^{A}_{B}|r}\right]\\
&+\left[2\pi \delta(p)\right]^2 \rbraket{l|\voperator{E}^{B'}_{A}|r}\rbraket{l|\voperator{E}^{A}_{B}|r}
\end{split}
\label{eq:umps:phipoverlap}
\end{equation}
As expected, momentum eigenstates cannot be normalized to unity in an infinitely large system, but rather satisfy a $\delta$ normalization. However, for momentum $p=0$, we have an additional diverging contribution which is much stronger. By using Eq.~(\ref{eq:umps:tangentoverlapwithumps}), it can be traced back to the diverging overlap of $\ket{\Phi_p(B)}$ with $\ket{\Psi(A)}$. 
Let us now analyze the origin of the different terms in the expression above. The regular part $\voperator{Q}\voperator{E}\voperator{Q}$ produces a finite contribution inside the square brackets where $B$ and $B'$ cannot be separated into different factors. We therefore also refer to these terms as the \emph{connected contribution}. For $p=0$, the product $\voperator{Q}(\voperator{\one}-\rme^{\pm \rmi p}\voperator{Q}\voperator{E}\voperator{Q})^{-1}\voperator{Q}$ can be interpreted as the pseudo-inverse $(\voperator{\one}-\voperator{E})^{(-1)}$ of the singular superoperator $\voperator{\one}-\voperator{E}$, which we have already encountered in the previous subsection and was defined in Eq.~(\ref{eq:umps:pseudop0}). We now extend this definition and henceforth define $(\voperator{\one}-\rme^{\pm \rmi p}\voperator{E})^{(-1)}\defis\voperator{Q}(\voperator{\one}-\rme^{\pm \rmi p}\voperator{Q}\voperator{E}\voperator{Q})^{-1}\voperator{Q}$, so that $(\voperator{\one}-\voperator{E})^{(-1)} (\voperator{\one}-\voperator{E})=(\voperator{\one}-\rme^{\pm \rmi p}\voperator{E})(\voperator{\one}-\rme^{\pm \rmi p}\voperator{E})^{(-1)}=\voperator{Q}=\voperator{\one}-\rket{r}\rbra{l}$. Only for momentum $p=0$ does $(\voperator{\one}-\rme^{\pm \rmi p}\voperator{E})^{(-1)}$ denote a true pseudo-inverse. The singular part $\voperator{S}$ produces a finite contribution in the square brackets for any momentum, and the doubly diverging contribution at momentum $p=0$. In these terms, $B$ and $B'$ appear in two separate factors, and they are henceforth referred to as the \emph{disconnected contribution}. Since the doubly diverging term results from the non-zero overlap with the original uMPS, it disappears for tangent vectors in $\Tplane_0^{(A)\perp}$.

Clearly, this hints that we should work in the projective setting. Since the pullback of the Fubini-Study metric is implicitly defined by
\begin{equation}
\overline{B}^{\overline{\imath}} \tilde{g}_{\overline{\imath},j}(p,p') {B'}^j=\frac{\braket{\Phi_{p}(\overline{B})|\Phi_{p'}(B')}}{\braket{\Psi(\overline{A})|\Psi(A)}}-\frac{\braket{\Phi_p(\overline{B})|\Psi(A)}\braket{\Psi(\overline{A})|\Phi_{p'}(B')}}{\braket{\Psi(\overline{A})|\Psi(A)}^2},
\end{equation}
where we use the convention to choose $A$ such that $\braket{\Psi(\overline{A})|\Psi(A)}=1$, we obtain
\begin{equation}
\begin{split}
\overline{B}^{\overline{\imath}} \tilde{g}_{\overline{\imath},j}(p,p') {B'}^j=&2\pi \delta(p-p')\overline{B}^{\overline{\imath}} \tilde{g}_{\overline{\imath},j}(p) {B'}^j\\
=&2\pi\delta(p'-p)\left[\rbraket{l|\voperator{E}^{B'}_{B}|r}+\rbraket{l|\voperator{E}^{A}_{B} (\voperator{\one}-\rme^{\rmi p}\voperator{E})^{(-1)}\voperator{E}^{B'}_{A}|r}\right.\\
&\qquad\left.+\rbraket{l|\voperator{E}^{B'}_{A} (\voperator{\one}-\rme^{-\rmi p}\voperator{E})^{(-1)}\voperator{Q}\voperator{E}^{A}_{B}|r}
-\rbraket{l|\voperator{E}^{B'}_{A}|r}\rbraket{l|\voperator{E}^{A}_{B}|r}\right].
\end{split}
\label{eq:umps:metric}
\end{equation}
Up to the unavoidable diverging $\delta$ normalization, we now obtain a strictly finite contribution that is henceforth denoted as $\tilde{g}_{\overline{\imath},j}(p)$, so that $\tilde{g}_{\overline{\imath},j}(p,p')=2\pi\delta(p-p') \tilde{g}_{\overline{\imath},j}(p)$. The doubly diverging contribution has been cancelled automatically. It can easily be checked that the contraction of either index of the metric $\tilde{g}(p)$ with any vector $B\in\tilde{\mathbb{N}}_p$, including the choice $B=A$ for momentum $p=0$, results in zero. We can thus restrict to vectors $B$ in the horizontal subspace $\tilde{\mathbb{B}}_p$ by imposing either the left or right gauge fixing conditions in Eq.~\eqref{eq:umps:leftgaugeuB} or \eqref{eq:umps:rightgaugeuB}. This considerably simplifies the expression for the metric $\tilde{g}(p)$, since the non-local connected terms and the disconnected term cancel, resulting in 
\begin{equation}
\overline{B}^{\overline{\imath}} \tilde{g}_{\overline{\imath},j}(p) B^{\prime j} = \rbraket{l|\voperator{E}^{B'}_{B}|r}\label{eq:umps:metricsimple}.
\end{equation}

As before, we define a pseudo-inverse metric satisfying
\begin{equation}
\tilde{g}^{i,\overline{\jmath}}(p)\tilde{g}_{\overline{\jmath},k}(p)=\left(P_{\tilde{\mathbb{B}}_p}\right)^{i}_{\; k}=\delta^{i}_{k} - \left(P_{\tilde{\mathbb{N}}_p}\right)^{i}_{\; k}
\end{equation}
and the projector $P_{\tilde{\mathbb{N}}_p}\in\End(\mathbb{A}_p)$ onto the vertical subspace $\tilde{\mathbb{N}}_p\subset \mathbb{A}_p$ is defined by
\begin{equation}
\left(P_{\tilde{\mathbb{N}}_p}\right)^{i}_{\; k} B^k = \tilde{\mathscr{N}}_p^{i}(\tilde{\omega}_p(B)).
\end{equation}
The pseudo-inverse $\tilde{g}^{i,\overline{\jmath}}(p)$ is defined within a single momentum sector. We can extend it as $\tilde{g}^{i,\overline{\jmath}}(p,p')=2\pi\delta(p-p') \tilde{g}^{i,\overline{\jmath}}(p)$ in order to obtain
\begin{equation}
\int\frac{\rmd p'}{2\pi} \tilde{g}^{i,\overline{\jmath}}(p,p')\tilde{g}_{\overline{\jmath},k}(p',p'')=2\pi\delta(p-p'') \left(P_{\tilde{\mathbb{B}}_p}\right)^{i}_{\; k}.
\end{equation}

Finally, we need to discuss how to efficiently parameterize the tensors $B$ that satisfy the left or right gauge fixing conditions in Eq.~\eqref{eq:umps:leftgaugeuB} or \eqref{eq:umps:rightgaugeuB}. A linear parameterization $B=\tilde{\mathscr{B}}_{p}(X)$ depending on a $(q-1)D\times D$ matrix $X$ can be constructed, analogously to the construction in the previous section, but now in a translation invariant setting. We first define the $ D\times D q$ matrices $L$ as
\begin{equation}
[L]_{\alpha;(s,\beta)}= [{A^{s}}^{\dagger} l^{1/2}]_{\alpha,\beta}\label{eq:umps:repdefL}
\end{equation}
and then construct a $D\times (q-1)D$ matrix $V_{L}$ that contains an orthonormal basis for the null space of $L$, \textit{i.e.} $L V_{L}=0$ and $V_{L}^\dagger V_{L}=\one_{(q-1)D}$. Setting $[V^{s}_{L}]_{\alpha,\beta}=[V_{L}]_{(s \alpha);\beta)}$, we then define the representation $\tilde{\mathscr{B}}_{p}(X)$ as
\begin{equation}
\tilde{\mathscr{B}}_{p}(X)= l^{-1/2} V_{L}^{s} X  r^{-1/2}\label{eq:umps:defrepleft}
\end{equation}
in order to obtain
\begin{equation}
\braket{\Phi_{p}(\overline{\tilde{\mathscr{B}}}_{p}(\overline{X}))|\Phi_{p'}(\tilde{\mathscr{B}}_{p'}(Y))}=2\pi\delta(p-p')\tr\left[ X^{\dagger} Y\right],
\end{equation}
in combination with the left gauge fixing condition $\sum_{s=1}^{q} {A^{s}}^{\dagger} l\tilde{\mathscr{B}}_{p}^{s}(X)=0$. The representation $\tilde{\mathscr{B}}_{p}'(X')$ mapping the $D\times (q-1)D$ matrix $X'$ to a tensor $B$ satisfying the right gauge fixing conditions follows similarly. 

\subsection{Levi-Civita connection and parallel transport}
\label{ss:umps:paralleltransport}
We have characterized $\varM_{\mathrm{uMPS}}$ as a K\"{a}hler manifold and defined the K\"{a}hler metric $\tilde{g}_{\overline{\imath},j}$. So far, we have been using a description based on the parameterization of $\varM_{\mathrm{uMPS}}$ via tensors $A\in\manifold{A}_{\mathrm{uMPS}}$. This parameterization is overcomplete, and the $A^i$'s cannot be used as a set of coordinates for $\varM_{\mathrm{uMPS}}$. As a consequence, the metric $\tilde{g}_{\overline{\imath},j}$ is not a proper metric, since it is degenerate.

The most rigorous way to proceed is by introducing a coordinate transform $A^i\leftarrow A^i(\vz,\vw)$ for $\manifold{A}_{\mathrm{uMPS}}$, where the new coordinates $z^{j}$ ($j=1,\ldots,(q-1)D^2$) and $w^k$ ($k=1,\ldots,D^2$) are such that $\partial/\partial z^{j} \in \tilde{\mathbb{B}}_{\mathrm{uMPS}}$ and $\partial/\partial w^{k} \in \tilde{\mathbb{N}}_{\mathrm{uMPS}}$. Hence, the coordinates $\vw$ are related to gauge and scale transformations, whereas the coordinates $\vz$ label the different gauge orbits. They arise as the natural coordinates for the quotient manifold $\manifold{A}_{\mathrm{uMPS}}/\mathsf{S}_{\mathrm{uMPS}}$, or thus, for the manifold $\varM_{\mathrm{uMPS}}$. When expressed solely in terms of the coordinates $\vz$, the pullback metric would be strictly positive.  Since we presently restrict to the representation of uniform MPS, the horizontal and vertical subspace correspond to those defined in the previous subsection at momentum zero: $\tilde{\mathbb{B}}_{\mathrm{uMPS}}=\tilde{\mathbb{B}}_{p=0}$ and $\tilde{\mathbb{N}}_{\mathrm{uMPS}}=\tilde{\mathbb{N}}_{p=0}$. Note that these spaces also depend on the current position $A(\vz,\vw)$. The required properties for a principal bundle connection by which these spaces are defined, ensures that such a coordinate transformation exists. While it is quite easy to find an explicit parameterization for the gauge degrees of freedom, it is more difficult to find an explicit coordinization for the gauge orbits. 

Therefore, we continue with the parameterization of $\varM_{\mathrm{uMPS}}$ based on the original tensors $A\in\manifold{A}_{\mathrm{uMPS}}$, and take into account that this set is overcomplete. A first consequence thereof has already been observed in the previous subsection: the pullback metric $\tilde{g}_{\overline{\imath}j}$ is degenerate and we need to take a pseudo-inverse to define the entries $\tilde{g}^{i,\overline{\jmath}}$. We now proceed by constructing the Levi-Civita connection according to Eq.~\eqref{eq:var:levicivita}. Hereto, we introduce the states
\begin{equation}
\begin{split}
&\ket{\Upsilon(B_1,B_2;A)}=\ket{\Upsilon^{(A)}(B_1,B_2)}=B_1^{i}B_2^{j}\frac{\partial^2\ }{\partial A^{i}\partial A^{j}} \ket{\Psi(A)}=B_1^{i}B_2^{j}\ket{\partial_i\partial_j\Psi(A)}\\
&=\sum_{n_1<n_2\in\mathbb{Z}}\sum_{\{s_{n}\}=1}^{q} \bm{v}_{\mathrm{L}}^{\dagger}\left[\left(\prod_{m<n_1} A^{s_{m}}\right) B_1^{s_{n_1}}\left(\prod_{n_1<m<n_2} A^{s_{m}}\right) B_2^{s_{n_2}}\left(\prod_{m'>n} A^{s_{m'}}\right)\right]\bm{v}_{\mathrm{R}} \ket{\{s_{n}\}}\\
&\quad+\sum_{n_2<n_1\in\mathbb{Z}}\sum_{\{s_{n}\}=1}^{q} \bm{v}_{\mathrm{L}}^{\dagger}\left[\left(\prod_{m<n_2} A^{s_{m}}\right) B_2^{s_{n_2}}\left(\prod_{n_2<m<n_1} A^{s_{m}}\right) B_1^{s_{n_1}}\left(\prod_{m'>n} A^{s_{m'}}\right)\right]\bm{v}_{\mathrm{R}} \ket{\{s_{n}\}}.
\end{split}
\label{eq:umps:defumpstangent}
\end{equation}
We can generalize this definition to obtain arbitrary momentum eigenstates $\ket{\Upsilon_{p_1,p_2}(B_1,B_2;A)}$ with momentum $(p_1+p_2)\mod 2\pi$ by adding a factor $\exp(\rmi p_1 n_1+\rmi p_2 n_2)$ to every term in the definition above. While we restrict to $p_1=p_2=0$ throughout the remainder of this section, the inclusion of momentum factors facilitates keeping track of the different terms in the following expressions. We can easily compute the overlap
\begin{equation}
\begin{split}
\braket{\Psi(\overline{A})|\Upsilon_{p_1,p_2}(B_1,B_2;A)}&=2\pi\delta(p_1+p_2)\\
\times \bigg[&\rme^{\rmi p_2}\rbraket{l|\mathbb{E}^{B_1}_A (\one-\rme^{\rmi p_2}\mathbb{E})^{(-1)}\mathbb{E}^{B_2}_A|r}+\rme^{\rmi p_1}\rbraket{l|\mathbb{E}^{B_2}_A (\one-\rme^{\rmi p_1}\mathbb{E})^{(-1)}\mathbb{E}^{B_1}_A|r}\\
&\quad +(2\pi\delta(p_1)-1)\rbraket{l|\mathbb{E}^{B_1}_A|r}\rbraket{l|\mathbb{E}^{B_2}_A|r}\bigg]
\end{split}\label{eq:umps:overlapgammapsi}
\end{equation}
and with a little bit more algebra
\begin{equation}
\begin{split}
\langle\Phi_{p_3}(\overline{B}_3;\overline{A})&\mid\Upsilon_{p_1,p_2}(B_1,B_2;A)\rangle=2\pi\delta(p_1+p_2-p_3) \\
\times\bigg\{&\rme^{+\rmi p_1} \rbraket{l|\mathbb{E}^{B_2}_{B_3}(\one-\rme^{+\rmi p_1}\mathbb{E})^{(-1)}\mathbb{E}^{B_1}_{A}|r}+\rme^{-\rmi p_1} \rbraket{l|\mathbb{E}^{B_1}_{A}(\one-\rme^{-\rmi p_1}\mathbb{E})^{(-1)}\mathbb{E}^{B_2}_{B_3}|r}\\
&+\rme^{+\rmi p_2} \rbraket{l|\mathbb{E}^{B_1}_{B_2}(\one-\rme^{+\rmi p_2}\mathbb{E})^{(-1)}\mathbb{E}^{B_2}_{A}|r}+\rme^{-\rmi p_2} \rbraket{l|\mathbb{E}^{B_2}_{A}(\one-\rme^{-\rmi p_2}\mathbb{E})^{(-1)}\mathbb{E}^{B_1}_{B_3}|r}\\
&+\rme^{+\rmi p_1+2\rmi p_2}\rbraket{l|^A_{B_3}(\one-\rme^{\rmi p_1+\rmi p_2} \mathbb{E})^{-1} \mathbb{E}^{B_1}_A (\one - \rme^{\rmi p_2}\mathbb{E})^{-1} \mathbb{E}^{B_2}_A |r}\\
&+\rme^{-2\rmi p_2-\rmi p_1}\rbraket{l|\mathbb{E}^{B_2}_{A}(\one-\rme^{-\rmi p_2} \mathbb{E})^{-1} \mathbb{E}^{B_1}_{A} (\one - \rme^{-\rmi p_1-\rmi p_2}\mathbb{E})^{-1} \mathbb{E}^{A}_{B_3} |r}\\
&+\rme^{+\rmi p_2+2\rmi p_1}\rbraket{l|\mathbb{E}^A_{B_3}(\one-\rme^{\rmi p_1+\rmi p_2} \mathbb{E})^{-1} \mathbb{E}^{B_2}_A (\one - \rme^{\rmi p_1}\mathbb{E})^{-1} \mathbb{E}^{B_1}_A |r}\\
&+\rme^{-\rmi p_2-2\rmi p_1}\rbraket{l|\mathbb{E}^{B_1}_{A}(\one-\rme^{-\rmi p_1} \mathbb{E})^{-1} \mathbb{E}^{B_2}_A (\one - \rme^{-\rmi p_1-\rmi p_2}\mathbb{E})^{-1} \mathbb{E}^{A}_{B_3} |r}\\
&+\rme^{-\rmi p_1+\rmi p_2}\rbraket{l|\mathbb{E}^{B_1}_{A}(\one-\rme^{-\rmi p_1} \mathbb{E})^{-1} \mathbb{E}^{A}_{B_3} (\one - \rme^{+\rmi p_2}\mathbb{E})^{-1} \mathbb{E}^{B_2}_{A} |r}\\
&+\rme^{-\rmi p_2+\rmi p_1}\rbraket{l|\mathbb{E}^{B_2}_{A}(\one-\rme^{-\rmi p_2} \mathbb{E})^{-1} \mathbb{E}^{A}_{B_3} (\one - \rme^{+\rmi p_1}\mathbb{E})^{-1} \mathbb{E}^{B_1}_{A} |r}\\
&-\rbraket{l|\mathbb{E}^{B_1}_A|r}\big[ \rbraket{l|\mathbb{E}^{B_2}_{B_3}|r} + \rme^{+\rmi p_2}\rbraket{l|\mathbb{E}^{A}_{B_3}(\one-\rme^{+\rmi p_2}\mathbb{E})^{(-1)}\mathbb{E}^{B_2}_{A}|r}+ \rme^{-\rmi p_2} \rbraket{l|\mathbb{E}^{B_2}_{A}(\one-\rme^{-\rmi p_2}\mathbb{E})^{(-1)}\mathbb{E}^{A}_{B_3}|r}\big]\\
&-\rme^{+\rmi p_1+\rmi p_2}\rbraket{l|\mathbb{E}^{B_1}_A|r}\rbraket{l|\mathbb{E}^{A}_{B_3}(\one-\rme^{+\rmi p_2}\mathbb{E})^{(-1)}(\one-\rme^{+\rmi p_1+\rmi p_2}\mathbb{E})^{(-1)}\mathbb{E}^{B_2}_A|r}\\
&-\rme^{-\rmi p_1-\rmi p_2}\rbraket{l|\mathbb{E}^{B_1}_A|r}\rbraket{l|\mathbb{E}^{B_2}_{A}(\one-\rme^{-\rmi p_2}\mathbb{E})^{(-1)}(\one-\rme^{-\rmi p_1-\rmi p_2}\mathbb{E})^{(-1)}\mathbb{E}^{A}_{B_3}|r}\\
&-\rbraket{l|\mathbb{E}^{B_2}_A|r}\big[ \rbraket{l|\mathbb{E}^{B_1}_{B_3}|r} + \rme^{+\rmi p_1}\rbraket{l|\mathbb{E}^{A}_{B_3}(\one-\rme^{+\rmi p_1}\mathbb{E})^{(-1)}\mathbb{E}^{B_1}_{A}|r}+ \rme^{-\rmi p_2} \rbraket{l|\mathbb{E}^{B_1}_{A}(\one-\rme^{-\rmi p_1}\mathbb{E})^{(-1)}\mathbb{E}^{A}_{B_3}|r}\big]\\
&-\rme^{+\rmi p_1+\rmi p_2}\rbraket{l|\mathbb{E}^{B_2}_A|r}\rbraket{l|\mathbb{E}^{A}_{B_3}(\one-\rme^{+\rmi p_1}\mathbb{E})^{(-1)}(\one-\rme^{+\rmi p_1+\rmi p_2}\mathbb{E})^{(-1)}\mathbb{E}^{B_1}_A|r}\\
&-\rme^{-\rmi p_1-\rmi p_2}\rbraket{l|\mathbb{E}^{B_2}_A|r}\rbraket{l|\mathbb{E}^{B_1}_{A}(\one-\rme^{-\rmi p_1}\mathbb{E})^{(-1)}(\one-\rme^{-\rmi p_1-\rmi p_2}\mathbb{E})^{(-1)}\mathbb{E}^{A}_{B_3}|r}\\
&-\rbraket{l|\mathbb{E}^{A}_{B_3}|r}\big[\rme^{+\rmi p_2}\rbraket{l|\mathbb{E}^{B_1}_{A}(\one-\rme^{+\rmi p_2}\mathbb{E})^{(-1)}\mathbb{E}^{B_2}_{A}|r}+ \rme^{+\rmi p_1} \rbraket{l|\mathbb{E}^{B_2}_{A}(\one-\rme^{+\rmi p_1}\mathbb{E})^{(-1)}\mathbb{E}^{B_1}_{A}|r}\big]\\
&-\rme^{-\rmi p_2}\rbraket{l|\mathbb{E}^{A}_{B_3}|r}\rbraket{l|\mathbb{E}^{B_2}_{A}(\one-\rme^{+\rmi p_1}\mathbb{E})^{(-1)}(\one-\rme^{-\rmi p_2}\mathbb{E})^{(-1)}\mathbb{E}^{B_1}_A|r}\\
&-\rme^{-\rmi p_1}\rbraket{l|\mathbb{E}^{A}_{B_3}|r}\rbraket{l|\mathbb{E}^{B_1}_{A}(\one-\rme^{+\rmi p_2}\mathbb{E})^{(-1)}(\one-\rme^{-\rmi p_1}\mathbb{E})^{(-1)}\mathbb{E}^{B_2}_{A}|r}\\
&+2\rbraket{l|E^{B_1}_A|r}\rbraket{l|E^{B_2}_A|r} \rbraket{l|E^{A}_{B_3}|r}\\
&+2\pi \delta(p_1) \rbraket{l|\mathbb{E}^{B_1}_A|r}\big[ \rbraket{l|\mathbb{E}^{B_2}_{B_3}|r} + \rme^{+\rmi p_2}\rbraket{l|\mathbb{E}^{A}_{B_3}(\one-\rme^{+\rmi p_2}\mathbb{E})^{(-1)}\mathbb{E}^{B_2}_{A}|r}\\
&\qquad \qquad\qquad\qquad\qquad+ \rme^{-\rmi p_2} \rbraket{l|\mathbb{E}^{B_2}_{A}(\one-\rme^{-\rmi p_2}\mathbb{E})^{(-1)}\mathbb{E}^{A}_{B_3}|r}-\rbraket{l|E^{B_2}_A|r} \rbraket{l|E^{A}_{B_3}|r}\big]\\
&+2\pi \delta(p_2) \rbraket{l|\mathbb{E}^{B_2}_A|r}\big[ \rbraket{l|\mathbb{E}^{B_1}_{B_3}|r} + \rme^{+\rmi p_1}\rbraket{l|\mathbb{E}^{A}_{B_3}(\one-\rme^{+\rmi p_1}\mathbb{E})^{(-1)}\mathbb{E}^{B_1}_{A}|r}\\
&\qquad \qquad\qquad\qquad\qquad+ \rme^{-\rmi p_2} \rbraket{l|\mathbb{E}^{B_1}_{A}(\one-\rme^{-\rmi p_1}\mathbb{E})^{(-1)}\mathbb{E}^{A}_{B_3}|r}-\rbraket{l|E^{B_1}_A|r} \rbraket{l|E^{A}_{B_3}|r}\big]\\
&+2\pi \delta(p_1+p_2) \rbraket{l|\mathbb{E}^{A}_{B_3}|r}\big[\rme^{+\rmi p_2}\rbraket{l|\mathbb{E}^{B_1}_{A}(\one-\rme^{+\rmi p_2}\mathbb{E})^{(-1)}\mathbb{E}^{B_2}_{A}|r}\\
&\qquad \qquad\qquad\qquad\qquad+ \rme^{+\rmi p_1} \rbraket{l|\mathbb{E}^{B_2}_{A}(\one-\rme^{+\rmi p_1}\mathbb{E})^{(-1)}\mathbb{E}^{B_1}_{A}|r}-\rbraket{l|E^{B_1}_A|r}\rbraket{l|E^{B_2}_A|r}\big]\\
&+(2\pi)^2 \delta(p_1)\delta(p_2) \rbraket{l|E^{B_1}_A|r}\rbraket{l|E^{B_2}_A|r} \rbraket{l|E^{A}_{B_3}|r}\bigg\}
\end{split}\label{eq:umps:overlapgammaphi}
\end{equation}
Since this subsection is concerned with the translation invariant manifold $\varM_{\mathrm{uMPS}}$, we restrict to the momentum zero states $p_1=p_2=p_3=0$. Clearly, then, the expressions above have a diverging prefactor $2\pi \delta(0)=\lvert \mathbb{Z}\rvert$ corresponding to the infinite number of sites. This is compensated by an factor $\lvert \mathbb{Z}\rvert^{-1}$ in the pseudo-inverse metric $\tilde{g}^{k\overline{m}}$ featuring in the definition of $\tilde{\Gamma}_{ij}^{\ \;k}$. However, Eq.~\eqref{eq:umps:overlapgammaphi} also contains additional divergences on the last four lines within the square brackets, resulting from the disconnected contributions of the transfer matrix. These divergences are precisely cancelled by the other terms in Eq.~\eqref{eq:var:levicivitaaux} resulting in a well-defined and finite Levi-Civita connection for the manifold $\varM_{\mathrm{uMPS}}$.

Many other terms in Eq.~\eqref{eq:umps:overlapgammaphi} can be eliminated by using $B_k$'s ($k=1,2,3$) that live within the horizontal space $\tilde{\mathbb{B}}_{\mathrm{uMPS}}=\tilde{\mathbb{B}}_{0}$ defined by the left or right gauge fixing conditions in Eq.~\eqref{eq:umps:leftgaugeuB} or \eqref{eq:umps:rightgaugeuB}. However, since the Levi-Civita connection is not a proper tensor, we do not expect it to be invariant under gauge transformations $B_k'=B_k + \tilde{\mathscr{N}_0}(\alpha_k,x_k)$. Put differently, the linear map $\mathbb{A}_{\mathrm{uMPS}}\to\hilbert$ obtained by fixing either $B_1$ or $B_2$ in the definition of $\ket{\Upsilon(B_1,B_2;A)}$ does not have $\tilde{\mathbb{N}}_{0}$ as kernel. For example, it can easily be checked that
\begin{displaymath}
\ket{\Upsilon(B_1,A;A)}=\big[2\pi\delta(0)-1\big] \ket{\Phi(B_1;A)}.
\end{displaymath}
By introducing an infinitesimal gauge transformation $G=\exp(\epsilon x)$ in $\ket{\Phi(B^{(G)};A^{(G)})}=\ket{\Phi(B;A)}$, we also obtain
\begin{displaymath}
\ket{\Upsilon(B,\mathscr{N}^{(A)}(x); A}= - \ket{\Phi(\mathscr{N}^{(B)}(x);A)}.
\end{displaymath}
Hence, under a gauge transformation $B_2'=B_2+\tilde{\mathscr{N}}(\alpha,x)$, we obtain
\begin{equation}
\ket{\Upsilon(B_1,B_2';A)}=\ket{\Upsilon(B_1,B_2;A)}-\alpha (2\pi\delta(0)-1)\ket{\Phi(B_1;A)}-\ket{\Phi(\mathscr{N}^{(B)}(x);A}.\label{eq:umps:upsilontransform}
\end{equation}
Note that the additional contributions are contained within the tangent space $T_{\ket{\Psi(A)}}\mathcal{M}_{\mathrm{uMPS}}$. 

Having a Levi-Civita connection at hand, we can now introduce a covariant derivative $\nabla_i$, and use this to define the space of first and second order derivatives properly. Acting with $\nabla_i$ a first time on a scalar function
\begin{displaymath}
\tilde{O}(\overline{A},A)=\frac{\braket{\Psi(\overline{A})|\operator{O}|\Psi(A)}}{\braket{\Psi(\overline{A})|\Psi(A)}}
\end{displaymath}
 on $\tilde{\varM}_{\mathrm{uMPS}}$, allows one to recognize the orthogonal complement of the standard tangent vectors 
\begin{equation}
\ket{\Phi'(B;A)}=\Big[\hat{1}-\hat{P}_{0}(\overline{A},A)\Big] \ket{\Phi(B;A)}.\label{eq:umps:defphiprime}
\end{equation}
Here we have introduced the projector onto the uMPS $\ket{\Psi(A)}$ as
\begin{equation}
\hat{P}_{0}(\overline{A},A)=\frac{\ket{\Psi(A)}\bra{\Psi(\overline{A})}}{\braket{\Psi(\overline{A})|\Psi(A)}}.
\end{equation}
Acting on $\nabla_i \tilde{O}(\overline{A},A) = \partial_i \tilde{O}(\overline{A},A)$ a second time with $\nabla_j$, results in
\begin{displaymath}
\nabla_j\nabla_i \tilde{O}(\overline{A},A)=\partial_i \partial_j \tilde{O}(\overline{A},A)-\tilde{\Gamma}_{ij}^{\ \;k} \partial_k \tilde{O}(\overline{A},A)
\end{displaymath}
from which we can infer the following covariant definition for double tangent vectors
\begin{equation}
\begin{split}
\ket{\Upsilon'(B_1,B_2;A)}=&\Big[\hat{1}-\hat{P}_0(\overline{A},A)\Big]\Big[\hat{1}-\hat{P}_{T_{\ket{\Psi(A)}}\tilde{\varM}_{\mathrm{uMPS}}}(\overline{A},A)\Big]\\
&\quad\times\bigg[\ket{\Upsilon(B_1,B_2;A)}-\ket{\Phi(B_1;A)} \frac{\braket{\Psi(\overline{A})|\Phi(B_2;A)}}{\braket{\Psi(\overline{A})|\Psi(A)}}\\
&\qquad\qquad\qquad\qquad\qquad-\ket{\Phi(B_2;A)} \frac{\braket{\Psi(\overline{A})|\Phi(B_1;A)}}{\braket{\Psi(\overline{A})|\Psi(A)}}\bigg].
\end{split}\label{eq:umps:defupsilonprime}
\end{equation}
We have now also introduced the projector onto the tangent space
\begin{equation}
\hat{P}_{T_{\ket{\Psi(A)}}\tilde{\varM}_{\mathrm{uMPS}}}(\overline{A},A)=\ket{\partial_k\Psi(A)} \tilde{g}^{k\overline{m}}(\overline{A},A) \bra{\opartial_{\overline{m}}\Psi(\overline{A})}.
\end{equation}
If $B_1$ and $B_2$ are such that $\braket{\Psi(\overline{A})|\Phi(B_1;A)}=\braket{\Psi(\overline{A})|\Phi(B_2;A)}=0$, then the states $\ket{\Upsilon'(B_1,B_2;A)}$ are equal to that part of $\ket{\Upsilon(B_1,B_2;A)}$ that is orthogonal to both the original uMPS $\ket{\Psi(A)}$ and all of its tangent vectors. Eq.~\eqref{eq:umps:upsilontransform} illustrates that the states $\ket{\Upsilon'(B_1,B_2;A}$ are invariant under gauge transformations of the $B_k$'s ($k=1,2$), \textit{i.e.} the linear maps obtained by fixing either $B_1$ or $B_2$ have the vertical subspace $\tilde{\mathbb{N}}_{\mathrm{uMPS}}$ as kernel.

\subsection{Riemann curvature tensor, Ricci tensor and scalar curvature}
\label{ss:umps:curvature}
To complete the Riemannian description of the manifold $\varM_{\mathrm{uMPS}}\subset \hilbert$ or $\tilde{\varM}_{\mathrm{uMPS}}\subset P(\hilbert)$, we now also compute the Riemann curvature tensor and its derivates. It can be shown that Eq.~\eqref{eq:var:riemmanprojective} results in
\begin{equation}
\begin{split}
B_1^i \overline{B}_2^{\overline{\jmath}} B_3^k \overline{B}_4^{\overline{l}} \tilde{R}_{i\overline{\jmath}k\overline{l}}=&\frac{\braket{\Upsilon'(\overline{B}_2,\overline{B}_4;\overline{A})|\Upsilon'(B_1,B_3;A)}}{\braket{\Psi(\overline{A})|\Psi(A)}}\\
&\qquad-\frac{\braket{\Phi'(\overline{B}_2;\overline{A})|\Phi'(B_1;A)}\braket{\Phi'(\overline{B}_4;\overline{A})|\Phi'(B_3;A)}}{\braket{\Psi(\overline{A})|\Psi(A)}^2}\\
&\qquad-\frac{\braket{\Phi'(\overline{B}_2;\overline{A})|\Phi'(B_3;A)}\braket{\Phi'(\overline{B}_4;\overline{A})|\Phi'(B_1;A)}}{\braket{\Psi(\overline{A})|\Psi(A)}^2}.
\end{split}
\label{eq:umps:riemanncurvature}
\end{equation}
Given the remarks at the end of the previous subsection, the Riemann tensor $\tilde{R}_{i\overline{\jmath}k\overline{l}}$ is gauge-invariant. An explicit expression of the equation above is given in Appendix~\ref{a:riemannexplicit}, but only for $B_k$'s ($k=1,\ldots,4$) satisfying the left gauge fixing condition Eq.~\eqref{eq:umps:leftgaugeuB}, which allows us to cancel many terms.

From the Riemann curvature tensor we can also define the Ricci tensor as
\begin{equation}
\widetilde{\text{Ric}}_{i\overline{\jmath}}=\tilde{R}^k_{\ ki\overline{\jmath}} = \tilde{g}^{k\overline{l}} \tilde{R}_{k\overline{l}i\overline{\jmath}}=-\partial_i\opartial_{\overline{\jmath}} \log \det [ \tilde{g} ].
\end{equation}
One also defines the \emph{Ricci form} as
\begin{equation}
\tilde{\mathcal{R}}=\rmi\ \widetilde{\text{Ric}}_{i\overline{\jmath}}\ \rmd z^i \wedge \overline{\rmd z}^{\overline{\jmath}},
\end{equation}
which is a real and closed two-form, whose cohomology class corresponds ---up to a constant factor--- to the first Chern class of the canonical line bundle\cite{Moroianu:2007uq}.

Finally, the scalar curvature is obtained by also contracting the Ricci tensor, resulting in
\begin{equation}
\tilde{S}=\tilde{g}^{i\overline{\jmath}}\widetilde{\text{Ric}}_{i\overline{\jmath}}=\tilde{g}^{i\overline{\jmath}}\tilde{g}^{k\overline{l}} \tilde{R}_{i\overline{\jmath}k\overline{l}}.
\end{equation}
The physical significance of the curvature of a variational manifold in relation to the approximation error made by reducing a state in Hilbert space to the manifold was discussed in great detail by \citet{Sidles:2009fk}, and we refer to this publication for more information. 

\section{Summary and outlook}
This article presents a thorough discussion of the mathematical structure of the MPS representation of states in either affine or projective Hilbert space using the language of fiber bundles and complex geometry. We discussed both generic MPS on finite chains with open boundary conditions and translation invariant MPS on chains with periodic boundary conditions or in the thermodynamic limit.

By restricting to the so-called subset of full rank MPS (open boundary conditions) or injective MPS (periodic boundary conditions), we were able to identify this representation with a principal fiber bundle. The variational parameters live in the bundle space. The physical states encoded by the variational parameters are left invariant under a well-understood set of gauge transformations and should therefore be identified with points in the base space, \textit{i.e.} the quotient space of the bundle space and the structure group (gauge group). This identification is bijective, and standard theorems of fiber bundle literature automatically imply that the set of MPS can therefore be given the structure of a (complex) manifold. Since this manifold is embedded in an affine or projective Hilbert space, which is a K\"{a}hler manifold, the manifold of MPS is also a K\"{a}hler manifold. The corresponding K\"{a}hler metric can be obtained by inducing the standard metric of Hilbert space.

A major part of this paper has focussed on the tangent space to the manifold of MPS. This linear subspace of Hilbert space has recently been proven interesting in both the study of time-evolution\cite{2011arXiv1103.0936H} and as a variational ansatz of elementary excitations\cite{2012PhRvB..85j0408H,2012PhRvB..85c5130P}. The gauge invariance of the MPS representation implies that not all partial derivatives with respect to the variational parameters produce linearly independent tangent states of the manifold. Within the fiber bundle context, the tangent map from the tangent bundle of parameter space to the tangent bundle of the base manifold has a non-trivial kernel that is referred to as the vertical subspace. A unique representation of tangent vectors requires the introduction of a principal bundle connection, which defines a complementary horizontal subspace. This can be understood as a canonical representation for MPS tangent vectors, where the connection acts as a gauge fixing prescription. It can be constructed in such a way that many physical expectation values involving the tangent states simplify tremendously. In particular, for every given base point, there exist at least two canonical representations that transform the metric at that point into the unit matrix. While this gauge fixing prescription had already been constructed in previous papers\cite{2011arXiv1103.0936H,2012PhRvB..85j0408H}, we have now shown that it satisfies the required criteria for being a principal bundle connection, so that it transforms equivariantly under gauge transformations. 

Given the physical relevance of the MPS tangent space, its rigorous constructions helps to identify where the naive description breaks down. Non-injective MPS correspond to singular points of the variational set and require a special treatment. However, it is quite trivial to generalize the present construction to the case of $\mathsf{G}$-injective MPS \cite{2010AnPhy.325.2153S}, which will have a smaller structure group given by the quotient group $\mathsf{PGL}(D,\mathbb{C})/\mathsf{G}$. We also expect that our construction trivially generalizes to the set of tree tensor networks\cite{2006PhRvA..74b2320S}, and possibly to the set of injective or $\mathsf{G}$-injective projected entangled pair states\cite{2010AnPhy.325.2153S}.

Aside from the already established applications in studying time-evolution and excitations, we believe that the concepts developed in this paper constitute a basis for many additional developments. For example, it has already been understood that the properties of the structure group $\mathsf{PGL}(D,\mathbb{C})$ are relevant in the classification of possible phases that gapped quantum systems can exhibit \cite{2011PhRvB..83c5107C,2011PhRvB..84p5139S}. A further study of the topological properties of either the structure group or the manifold itself might unveil additional information. On a more general level, by using the rules of geometric quantum mechanics\cite{Giachetta:2011uq,2001JGP....38...19B,Ashtekar:1997kx}, a complete theory of quantum mechanics can be formulated in which the usual projective Hilbert space is replaced by the smaller manifold of MPS. The K\"{a}hler structure of the manifold gives rise to an invariant volume measure that can be used in order to describe mixed states as probability measures over the manifold of MPS. On the more practical side, it has been understood that standard optimization methods such as conjugate gradient or Newton's method benefit greatly from taking geometrical properties of the underlying search space into account\cite{Absil:2009fk}, by replacing ordinary derivatives by covariant derivatives and line searches by searches along geodesic. Clearly, the Levi-Civita connection and the notion of parallel transport play a major role in this development. Applying these methods to the manifold $\varM_{\mathrm{(u)MPS}}$ might result in new algorithms for finding ground states whose efficiency is less susceptible to the magnitude of the gap in the system.

\begin{acknowledgements}
J.H.~greatly acknowledges inspiring discussions with Tom Mestdag and Eduardo Garc\'{i}a Tora\~{n}o and valuable comments of Bart Vandereycken. Christine Tobler is acknowledged for pointing out related work in the field of numerical mathematics, in particular Ref.~\onlinecite{Holtz:2012kx,Uschmajew:2012vn}. Research supported by Research Fund Flanders (M.M.), by the EU grants QUERG and QFTCMPS, by the FWF SFB grants FoQuS and ViCoM and by the cluster of excellence EXC 201 Quantum Engineering and Space-Time Research.
\end{acknowledgements}

\appendix
\section{Graphical notation}
\label{a:graphnot}
The tensor network community has grown accustomed to using a graphical notation to represent tensor networks and visualize tensor identities that need to be imposed or are automatically satisfied. Any shape can in principle be used to represent a tensor. The tensor indices are represented by wires that can range over their respective values. When two wires exiting from two different tensors are connected, this corresponds to a contraction of the corresponding indices.

We now illustrate that such a graphical representation is also very natural in the context of MPS tangent vectors. FIG.~\ref{fig:graph:mps}(a) represents a basic tensor $A$ that could be an element of $\mathbb{A}_{\mathrm{uMPS}}$ and would then feature in the definition of a uniform MPS. The bottom wire corresponds to the physical indices, whereas the left and right wires correspond to the row and column indices of the matrix $A^s$ for a given value $s=1,\ldots,q$ of the physical index. FIG.~\ref{fig:graph:mps}(b) represents the definition of the transfer matrix $\voperator{E}=\sum_{s=1}^q A^s \otimes \overline{A}^s$, corresponding to a contraction of the physical index. Note that complex conjugation is not explicitly denoted in the graphical representation, but is implied by mirroring the tensor $A$ so as to have the physical index as an upper wire. Assuming that $A$ is properly normalized such that the largest eigenvalue of the transfer operator $\voperator{E}$ is $1$, we can graphically denote the eigenvalue equation for the left and right eigenvectors $\rbra{l}$ and $\rket{r}$ corresponding to Hermitian matrices $l$ and $r$ as in FIG.~~\ref{fig:graph:mps}(c) and \ref{fig:graph:mps}(d).

\begin{figure}
\begin{center}
\includegraphics{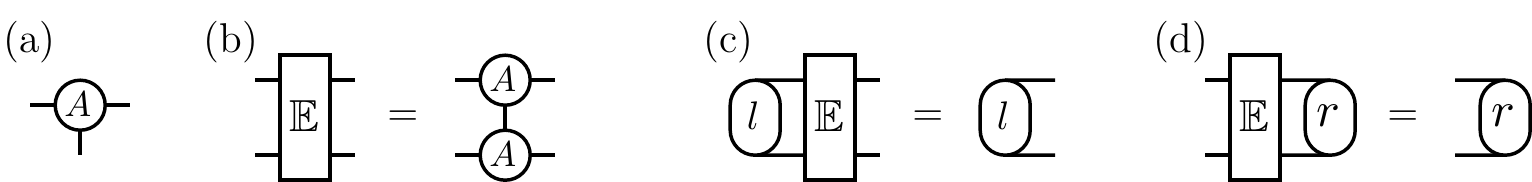}
\caption{Graphical illustration of the basic definitions involving a uniform MPS: (a) the elementary tensor $A\in\mathbb{A}_{\mathrm{uMPS}}$ that forms the basic building block of the uMPS; (b) definition of the transfer matrix $\voperator{E}$ obtained by contracting $A$ and its complex conjugate along the physical index; (c,d) eigenvalue equation for the left (c) and right (d) density matrices $l$ and $r$ under the assumption that $A$ is properly normalized such that $\voperator{E}$ has largest eigenvalue $1$.}
\label{fig:graph:mps}
\end{center}
\end{figure}

This graphical notation is particularly convenient to define \textit{e.g.} the left canonical form for the tensor $B$ used in the representation of tangent vectors $\ket{\Phi_p(B)}$, as is shown in FIG.~\ref{fig:graph:tangent}. In addition, we can expand full expectation values such as the (squared) norm $\braket{\Phi_p(B)|\Phi_p(B)}$ in terms of these diagrams, as in FIG.~\ref{fig:graph:metric}, and easily infer how this expression simplifies by using the defining equations of the tensor $V_L$ in FIG.~\ref{fig:graph:tangent}(b) and \ref{fig:graph:tangent}(c).

\begin{figure}
\begin{center}
\includegraphics{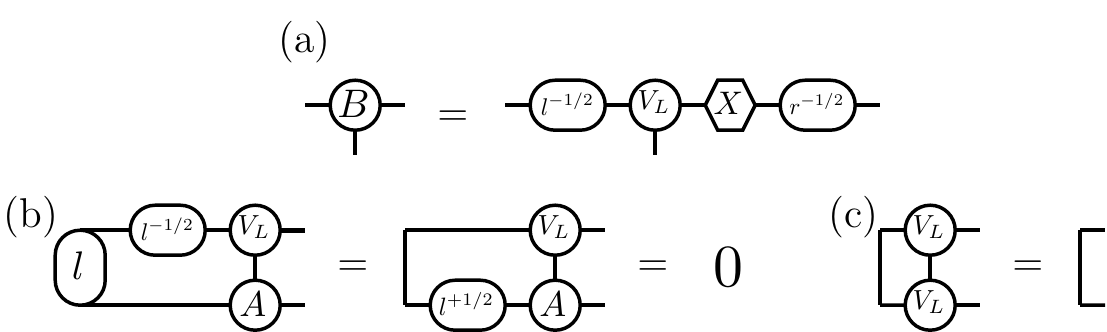}
\caption{Graphical representation of the left canonical form for the parameterization of tangent vectors to the manifold of uniform MPS $B=\tilde{\mathscr{B}}_p(x)$ according to Eq.~\eqref{eq:umps:defrepleft} (a), where the tensor $V_L$ is determined by the equations denoted in (b) and (c).}
\label{fig:graph:tangent}
\end{center}
\end{figure}

\begin{figure}
\begin{center}
\includegraphics[width=\textwidth]{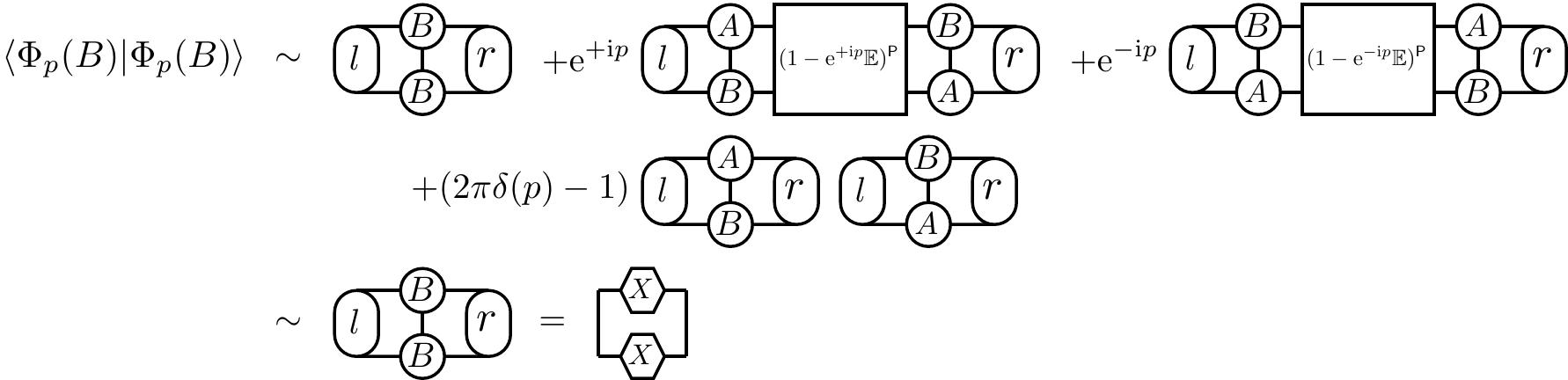}
\caption{Graphical representation of the norm of a tangent vector $\ket{\Phi_p(B)}$ [up to the diverging factor $2\pi \delta(0)$] and its simplification after the left canonical form of FIG.~\ref{fig:graph:tangent} has been inserted.}
\label{fig:graph:metric}
\end{center}
\end{figure} 

\section{The Riemann curvature tensor}
\label{a:riemannexplicit}
The Riemann curvature tensor of the manifold of (projective) uniform MPS at a point $\ket{\Psi(A)}$ was given in the main text in Eq.~\eqref{eq:umps:riemanncurvature}, using the covariantly defined states $\ket{\Phi'(B;A)}$ and $\ket{\Upsilon'(B_1,B_2;A)}$. Its full expression contains many terms, most of which are to ensure that $\tilde{R}_{i\overline{\jmath}k\overline{l}}$ acts as a proper gauge-invariant tensor, so that it is zero whenever one of the indices is contracted with an element of the vertical subspace $\tilde{\mathbb{N}}_{\mathrm{uMPS}}$. If we restrict to contraction with elements $B_1$, $B_2$, $B_3$ and $B_4$ within the horizontal subspace defined by either the left or right gauge fixing conditions [Eq.~\eqref{eq:umps:leftgaugeuB} and Eq.~\eqref{eq:umps:rightgaugeuB} respectively], then many of these terms cancel automatically. If for example the left gauge fixing conditions are satisfied, the resulting tensor elements are shown in FIG.~\ref{fig:graph:riemann}. The state $\ket{\Upsilon'(B_1,B_2)}$ contains the ordinary second derivative $\ket{\Upsilon(B_1,B_2)}$ minus its projection its projection onto the first tangent space. The former produces the first six lines, whereas the latter corresponds to lines number 7 and 8, and is formulated using the tensor $V_L$ which was defined graphically in the previous appendix. The last line corresponds to the disconnected contributions to the Riemann curvature tensor.

\begin{figure}
\begin{center}
\includegraphics[width=\textwidth]{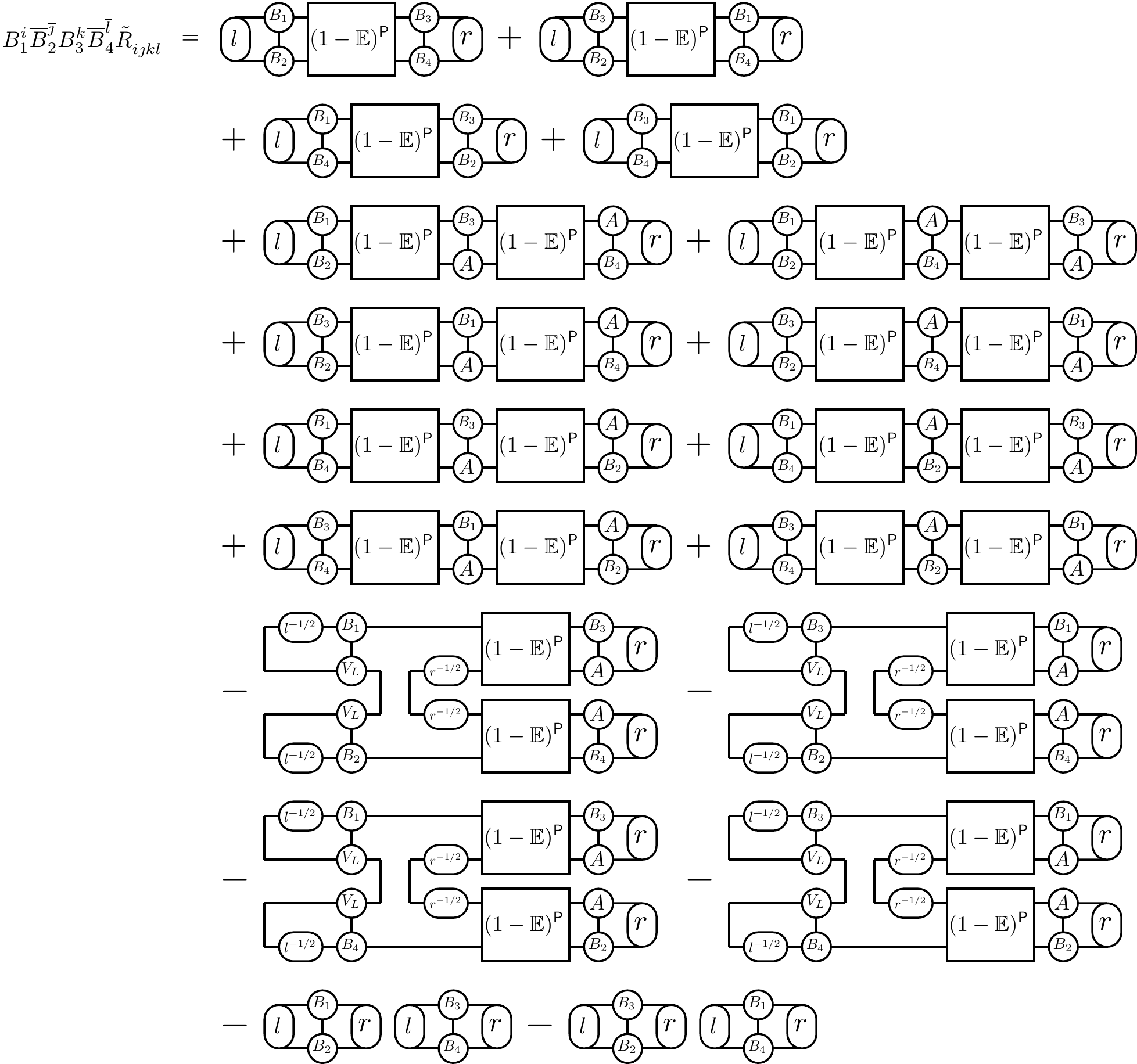}
\caption{Graphical representation of the Riemann curvature tensor $\tilde{R}_{i\overline{\jmath}k\overline{l}}$ when contracted with vectors $B\in\tilde{\mathbb{B}}_{\mathrm{uMPS}}$ satisfying the left gauge fixing condition of Eq.~\eqref{eq:umps:leftgaugeuB}.}
\label{fig:graph:riemann}
\end{center}
\end{figure}

\bibliography{paperslibrary,manuallibrary,books}
\bibliographystyle{aipauth4-1}

\end{document}